\documentclass[11pt]{amsart}

\usepackage{hyperref}
\usepackage{tikz}
\usetikzlibrary{arrows,decorations.markings,shapes.arrows,patterns,calc}
\usepackage[all]{xy}

\tikzset{% arrow close to the source: the 0.2 determines where the arrow is drawn
  ->-/.style={decoration={markings, mark=at position 0.5 with {\arrow{to}}},
              postaction={decorate}},
}

\tikzset{% arrow close to the source: the 0.2 determines where the arrow is drawn
  -<-/.style={decoration={markings, mark=at position 0.5 with {\arrow{to reversed}}},
              postaction={decorate}},
}

\tikzset{% arrow close to the source: the 0.2 determines where the arrow is drawn
  dbl->-/.style={
double, 
double equal sign distance,
shorten >= 1pt,
shorten <= 1pt,
 decoration={markings, mark=at position 0.5 with {\arrow{implies}}},
              postaction={decorate}},
}

\tikzset{% arrow close to the source: the 0.2 determines where the arrow is drawn
  dbl-<-/.style={
double, 
double equal sign distance,
shorten >= 1pt,
shorten <= 1pt,
 decoration={markings, mark=at position 0.5 with {\arrowreversed{implies}}},
              postaction={decorate}},
}

\usepackage{amsmath,amsthm}
\pdfmapfile{+mathpple.map}
\usepackage{mathrsfs}
\usepackage{amscd,amssymb, amsfonts, verbatim,subfigure, enumerate}
\usepackage[mathcal]{eucal}
\usepackage[super]{nth}

\usepackage{mathpazo}

\usepackage{color,slashed}

\newcommand{\Lap}{\tr}
\newcommand{\dbar}{\br{\partial}}

\newcommand{\zbar}{\br{z}}

\newcommand{\dpa}[1]{\frac{\partial}{\partial #1}}

\newcommand{\Obs}{\op{Obs}}

\newcommand{\eps}{\epsilon}
\newcommand{\g}{\mathfrak{g}}

\newcommand{\xto}{\xrightarrow}

\newcommand{\what}{\widehat}
\newcommand{\tr}{\triangle}

\newcommand{\til}{\widetilde}
\newcommand{\mscr}{\mathscr}

\newcommand{\br}{\overline}

\newcommand{\iso}{\cong}
\newcommand{\C}{\mathbb C}

\newcommand{\Oo}{\mscr O}
\newcommand{\Z}{\mathbb Z}

\newcommand{\into}{\hookrightarrow}

\newcommand{\op}{\operatorname}

\newcommand{\mbb}{\mathbb}
\newcommand{\mf}{\mathfrak}
\newcommand{\mc}{\mathcal}

\newcommand{\ip}[1]{\left\langle #1 \right\rangle}
\newcommand{\abs}[1]{\left| #1 \right|}

\newcommand{\R}{\mbb R}
\renewcommand{\d}{\mathrm{d}}

\DeclareMathOperator{\Sym}{Sym} \DeclareMathOperator{\Hom}{Hom}

\newtheoremstyle{thm}% name
  {7pt}%      Space above
  {7pt}%      Space below
  {\itshape}%         Body font
  {}%         Indent amount (empty = no indent, \parindent = para indent)
  {\bf}% Thm head font
  {.}%        Punctuation after thm head
  {5pt}%     Space after thm head: " " = normal interword space;
  {\thmnumber{#2 }\thmname{#1}\thmnote{ (#3)}}%         Thm head spec (can be left empty, meaning `normal')

\newtheoremstyle{def}% name
  {7pt}%      Space above
  {10pt}%      Space below
  {\itshape}%         Body font
  {}%         Indent amount (empty = no indent, \parindent = para indent)
  {\bf}% Thm head font
  {.}%        Punctuation after thm head
  {5pt}%     Space after thm head: " " = normal interword space;
  {\thmnumber{#2} \thmname{#1}\thmnote{ (#3)}}%         Thm head spec (can be left empty, meaning `normal')

\newtheoremstyle{rem}% name
  {4pt}%      Space above
  {10pt}%      Space below
  {}%         Body font
  {}%         Indent amount (empty = no indent, \parindent = para indent)
  {\itshape}% Thm head font
  {:}%        Punctuation after thm head
  {3pt}%     Space after thm head: " " = normal interword space;
  {}%         Thm head spec (can be left empty, meaning `normal')

\newtheoremstyle{texttheorem}% name
  {8pt}%      Space above
  {8pt}%      Space below
  {\itshape}%         Body font
  {}%         Indent amount (empty = no indent, \parindent = para indent)
  {\bf}% Thm head font
  {. \hspace{5pt}}%        Punctuation after thm head
  {3pt}%     Space after thm head: " " = normal interword space;
  {}%         Thm head spec (can be left empty, meaning `normal')

\theoremstyle{thm}

\newtheorem*{claim}{Claim}
\newtheorem*{theorem*}{Theorem}
\newtheorem*{lemma*}{Lemma}
\newtheorem*{corollary*}{Corollary}
\newtheorem*{proposition*}{Proposition}
\newtheorem*{definition*}{Definition}

\newtheorem*{conjecture}{Conjecture}

\newtheorem{theorem}{Theorem}[subsection]
\newtheorem{thm-def}{Theorem/Definition}[theorem]
\newtheorem{proposition}[theorem]{Proposition}

\newtheorem*{question*}{Question}
\newtheorem{lemma}[theorem]{Lemma}

\newtheorem{corollary}[theorem]{Corollary}

\numberwithin{equation}{subsection}

\theoremstyle{def}
\newtheorem{definition}[theorem]{Definition}

\theoremstyle{rem}
\newtheorem*{remark}{Remark}

\newcommand{\cinfty}{C^{\infty}}

\usepackage{stmaryrd}
\date{}
\newcommand{\gl}{\mf{gl}}

\title{Holography and Koszul duality: the example of the $M2$ brane}
\author{Kevin Costello}
\thanks{}

\address{Perimeter Institue for Theoretical Physics}
\email{kcostello@perimeterinstitute.ca}

\begin{document}

\begin{abstract}
Si Li and author suggested in that, in some cases, the AdS/CFT correspondence can be formulated in terms of the algebraic operation of Koszul duality.  In this paper this suggestion is checked explicitly for $M2$ branes in an $\Omega$-background. The algebra of supersymmetric operators on a stack of $K$ $M2$ branes is shown to be Koszul dual, in large $K$, to the algebra of supersymmetric operators of $11$-dimensional supergravity in an $\Omega$-background (using the formulation of supergravity in an $\Omega$-background presented in \cite{Cos16}).  

The twisted form of supergravity that is used here can be quantized to all orders in perturbation theory.  We find that the Koszul duality result holds to all orders in perturbation theory,  in both the gravitational theory and the theory on the $M2$.  (However, there is a certain non-linear identification of the coupling constants on each side which I was unable to determine explicitly).

It is also shown that the algebra of operators on $K$ $M2$ branes, as $K \to \infty$, is a quantum double-loop algebra (a two-variable analog of the Yangian).  This algebra is also the Koszul dual of the algebra of operators on the gravitational theory. An explicit presentation for this algebra is presented, and it is shown that this algebra is the unique quantization of its classical limit.     Some conjectural applications to enumerative geometry of Calabi-Yau threefolds are also presented.  
\end{abstract}

\maketitle
\tableofcontents
\section{Introduction}
The AdS/CFT correspondence has been very influential in theoretical physics. However, compared to other aspects of string theory, the impact of holography on mathematics has been more limited.  Perhaps one reason is that, previously, a twisted form of the AdS/CFT correspondence has been unavailable.

Such a twisted AdS/CFT correspondence has been proposed by the author and Si Li \cite{CosLi16}.  The twisted AdS/CFT correspondence relates a twist of a supersymmetric gauge theory (in the sense of \cite{Wit88a}) with a twist of supergravity.  Twisted supergravity was introduced by the author and Si Li in \cite{CosLi15,CosLi16}. Although twisted supergravity is non-renormalizable, we were able to construct it at the quantum level in perturbation theory in some important examples.

Li and I came up with the following mathematical formulation of (part of) the AdS/CFT correspondence, which we state as a conjecture (and at a somewhat impressionistic level).
\begin{conjecture}
Consider a $D$-brane in type IIA or IIB string theory, a brane the topological string, or a brane in $M$-theory. Following Polchinski, we can view the $D$-brane as a defect in the string theory or supergravity theory.  Then, we can consider:
\begin{enumerate} 
\item The algebras of operators of the supersymmetric gauge theory living on stack of $N$ $D$-branes, after performing a twist and sending $N \to \infty$. 
\item The algebra of operators in perturbative twisted supergravity living on the location of the defect given by the stack of $D$-branes. 
\end{enumerate} 
Then, these two algebras are Koszul dual\footnote{As stated, the conjecture is imprecise. Only in the case that the brane sources no flux in the supergravity theory do we expect an exact match with the Koszul dual. In general, we expect a deformation of the Koszul dual.   This unusual situation can be engineered by considering a stack of the same number of branes and anti-branes.  We also find that $M2$ branes in $11$-dimensional supergravity in the $\Omega$-background do not source a flux. } . 
\end{conjecture}
\begin{remark}
 All the examples we have, and all of the theoretical evidence supporting this conjecture, only hold after performing some (not necessarily topological) twist to both the theory on the brane and the dual supergravity theory (where the supergravity theory is twisted in the sense of  \cite{CosLi16}).  Maybe some  version of this conjecture holds for the physical theory, before twisting, but we do not have good evidence for that right now. 
\end{remark} 
It may seem rather outrageous to conjecture that holography can be understood in terms of an abstract homological-algebra operation like Koszul duality.  Such a claim needs to be supported by strong evidence.   

One strand of evidence we plan to present for this conjecture is the relationship between this statement and Witten's \cite{Wit98} prescription for computing gauge-theory correlators in terms of the supergravity dual.   

The other evidence we can present is to verify the conjecture in detail in examples.  This is (part of) the purpose of this paper.  The example considered here comes from $M$-theory instead of string theory. We consider a stack of $K$  $M2$ branes at the tip of an $A_{N-1}$ singularity, where both the $M2$ brane and $11$-dimensional supergravity are placed in an $\Omega$-background. (The concept of supergravity in the $\Omega$-background is developed in \cite{Cos16}).  

Our conjecture now states that the algebra of operators on the stack of $M2$ branes, as $K \to \infty$, becomes Koszul dual to the algebra of operators of $11$-dimensional supergravity in the $\Omega$-background.  We find, after a lot of work, that this statement holds, \emph{even when the  supergravity theory is treated at the quantum level}.  

We use the description of $11$-dimensional supergravity in the $\Omega$-background presented in \cite{Cos16}, where it is found that it is equivalent to a $5$-dimensional non-commutative gauge theory.  There, it was also shown that this $5$-dimensional gauge theory can be quantized in perturbation theory; this is what allows us to access part of $11$-dimensional supergravity theory at the quantum level. 

Incidentally, this statement does not contradict the standard idea that one should not find any local operators in a gravity theory.  The $5$-dimensional gauge theory which is our avatar of $11$-dimensional supergravity has no local operators \emph{of ghost number zero}.  However, it does have local operators of positive ghost number, which  come from the symmetries of the background we use. The operation of Koszul duality can turn an algebra built from elements of positive cohomological degree into one concentrated in cohomological degree $0$, so that operators of positive ghost number on the gravitational side are related to operators of ghost number zero living on the brane.

\subsection{} 
  The proof of this result is very much non-trivial: we have to work very hard to constrain the operator products that appear in the gravitational system. In the end, the match we obtain between the two systems involves a non-linear identification of the coupling constants.  This is unique, and can in principle be determined by Feynman diagram computations, but I was unable to compute it.  

Despite the motivation from physics, this paper is entirely mathematical, and in fact mostly algebra. Mathematicians should not be discouraged by the mentions of $M$-theory: I will shortly explain the main results of the paper from a purely mathematical point of view.

\subsection{A $5$-dimensional gauge theory from $M$-theory}
Let me now describe the $5$-dimensional gauge theory which will be our substitute for $M$-theory. In \cite{Cos16}, I gave a definition of supergravity in the $\Omega$ background.  I also analyzed $11$-dimensional supergravity on a product of $TN_k$, the Taub-NUT manifold with an $A_k$ singularity at the origin, with $\R^3 \times \C^2$.  I showed that, when placed in an $\Omega$-background arising from rotating the circle fibre of the Taub-NUT with speed $\eps$ and $\R^2$ with speed $\delta$ (while preserving the holomorphic volume form on $TN_k \times \R^2$), $M$-theory reduces to a $5$-dimensional non-commutative gauge theory.

 The fundamental field of this gauge theory is a $3$-component connection which, on flat space $\R \times \C^2$ takes the form
$$
A = A_t \d t + A_{\zbar_1} \d \zbar_1 + A_{\zbar_2} \d \zbar_2
$$ 
where the components are smooth functions valued in $\mf{gl}_N$.
 
The action functional is the non-commutative Chern-Simons functional
$$
S(A) = \frac{1}{\delta}  \int \d z_1 \d z_2 \left\{ \tfrac{1}{2} \op{Tr} ( A \ast_\eps \d A)  + \tfrac{1}{3} \op{Tr} (A \ast_\eps A \ast_\eps A) \right\}k
$$
where $\ast_\eps$ indicates the Moyal product on the plane $\C^2$ with parameter of non-commutativity $\eps$.  The theory has two coupling constants $\delta$ and $\eps$, and we work in a perturbative regime where both $\delta$ and $\eps$ are small but $\abs{\delta}$ is much less than $\abs{\eps}$.   These coupling constants are the same as the $\Omega$-background parameters mentioned above.

This theory can be placed on $\R \times X$ where $X$ is a  hyper-K\"ahler $4$-manifold which, when viewed as a complex symplectic surface, is equipped with a deformation quantization of the sheaf of holomorphic functions.  For example, $X$ could be a resolution of an ADE singularity or the cotangent bundle of a Riemann surface. In the limit when $\eps = 0$, the equations of motion of this theory describe a bundle on $\R \times X$ with a holomorphic structure in the $X$ directions and a flat connection in the $\R$ direction. When $\eps$ is non-zero, the equations of motion are modified so that the holomorphic bundle on $X$ is replaced by a module over the non-commutative structure sheaf of $X$.

This $5$-dimensional gauge theory is our substitute for $M$-theory. For the rest of the paper, I will only use this $5$-dimensional gauge theory, and refrain from referring back to the full $11$-dimensional supergravity. Mathematicians who are not familiar with $M$-theory can take this $5$-dimensional gauge theory as a definition (for the purposes of this paper).  

$11$-dimensional supergravity has two types of extended objects, namely $M2$ and $M5$ branes. In this paper we are exclusively interested in $M2$ branes.   These extended objects have residues in the $5$-dimensional gauge theory. We will discuss  the theory on an $M2$ brane in detail shortly. 

%An $M2$ brane wrapping $\R^2_{-2 \delta} \times \R$ in the $M$-theory background $TN_{N-1, \eps,\delta} \times \R^2_{-2\delta} \times \R \times X$ has a residue in the $5$-dimensional gauge theory as an object supported on $\R \times p$ inside $\R \times X$ (where $p \in X$).  In \cite{Cos16} it was argued that this object has a semi-classical description as an instanton supported on $\R \times p$.  Thus, $11$-dimensional supergravity in the presence of $K$ $M2$ branes can be described by the $5$-dimensional gauge theory on $\R \times \C^2$ in perturbation theory around a background connection whose second Chern class is equivalent to $K$ times the fundamental  class of $\R \times p$ (for $p \in \C^2$).   

%There is also a theory on a stack of $K$ $M2$ brane. In the $\Omega$-background, this is a one-dimensional field theory, that is, a quantum mechanical system.  In \cite{Cos16} it is argued that we can describe this by a certain quiver quantum mechanics associated to the ADHM quiver. The algebra of operators of this quiver quantum mechanics is a deformation quantization of the algebra of holomorphic functions on the Nakajima quiver variety associated ot the ADHM quiver. That is, it is a deformation quantization of the moduli of framed torsion-free sheaves on a non-commutative $\C^2$ of rank $N$ and charge $K$.  
\subsection{Quantizing the $5d$ gauge theory}
Let us now describe the mathematical results we can prove about this $5$-dimensional quantum field theory. 
\begin{theorem*} 
Let $X$ be a complex symplectic surface which is either $\C^2$, a resolution of an ADE singularity, or the cotangent bundle of a Riemann surface. Then, the $5$-dimensional gauge theory on $\R \times X$  can be quantized essentially uniquely in perturbation theory. 
\end{theorem*}
This theorem is proved in \cite{Cos16},  using the method for perturbative renormalization in the BV formalism developed in \cite{Cos11}.   Like any non-commutative gauge theory, our theory is naively non-renormalizable: there are terms in the interaction with arbitrary high number of derivatives.   (We can still treat it as a local theory as long as we work in perturbation theory in the parameter $\eps$.)  We find that the quantum master equation for the theory is strong enough to constrain all possible counter-terms uniquely, up to a redefinition of the coupling constants $\eps$ and $\delta$. 

The theory developed in \cite{Cos11,CosGwi11} allows one to encode the algebra of operators (or observables) of the theory into an object called a factorization algebra. In particular, the operator product expansion in the $\R$ direction, in which the theory is topological, makes the space of local operators of the theory into a homotopy associative algebra, say $\Obs_{\eps,\delta}$.  

At $\delta = 0$, this algebra has a simple description. Let $\Oo_\eps(\C^2)$ denote the non-commutative deformation of the algebra $\C[z_1,z_2]$ of polynomial functions on $\C^2$ where the variables $z_i$ satisfy the commutation relation 
$$
[z_1,z_2] = \eps. 
$$
Then,
$$
\Obs_{\eps,\delta = 0} \simeq C^\ast ( \Oo_\eps(\C^2) \otimes \mf{gl}_N) 
$$
where $C^\ast$ indicates the Chevalley-Eilenberg Lie algebra cochain complex.  One arrives at this description as follows.  First, we note that  locally, every solution to the equations of motion of our $5$-dimensional gauge theory is trivial. However, the trivial solution has a very large Lie algebra of infinitesimal gauge symmetries, namely $\Oo_\eps(\C^2) \otimes \mf{gl}_N$.   Functions of the ghosts for these gauge symmetries,  equipped with the natural BRST differential,  give us the Chevalley-Eilenberg cochain algebra. 

For non-zero $\delta$, the algebra $\Obs_{\eps,\delta}$ will be some interesting deformation of this Chevalley-Eilenberg cochain complex.

\subsection{Koszul duality and quantum double loop algebras}
We have proposed that one can understand the AdS/CFT correspondence at the level of operators as being Koszul duality.  Thus, the Koszul dual of the algebra $\Obs_{\eps,\delta}$ will be an important part of the story.   

At $\delta=0$, the theory is classical and Koszul dual is easily understood. Standard results tell us that, for any Lie algebra $\mf{g}$, the Koszul dual of $C^\ast(\mf{g})$ is the universal enveloping algebra $U(\mf{g})$.  It follows that the Koszul dual of $\Obs_{\eps,0}$ is the universal enveloping algebra $U(\Oo_{\eps}(\C^2) \otimes \mf{gl}_N)$, and that the Koszul dual of $\Obs_{\eps,\delta}$ is some deformation of this algebra.  We will denote this deformation by 
$$
U_{\delta}^{QFT}( \Oo_{\eps}(\C^2) \otimes \mf{gl}_N). 
$$ 
The superscript ``QFT'' is to distinguish this from another deformation of  $U(\Oo_\eps (\C^2) \otimes \mf{gl}_N)$ we will consider momentarily.  An explicit calculation of operator products will show us that $U_{\delta}^{QFT}(\Oo_\eps(\C^2) \otimes \mf{gl}_N)$ is a non-trivial deformation to first order in $\delta$. 
\begin{remark}
A very similar quantum double loop group was introduced by N. Guay \cite{Gua07}.  The algebras studied here are almost certainly isomorphic to those studied by Guay. Indeed, Guay shows that at a certain specialization of his parameters, his algebra is isomorphic to $U(\mf{gl}_N  \otimes \op{Diff}(\C))$. The formula for the first-order deformation away from this specialization coincides with the first-order deformation  obtained here.  Some work remains to prove the algebras are isomorphic: Guay's algebras are not presently known to satisfy the criteria of the uniqueness theorem  I prove for the quantization of $U(\Oo_\eps(\C^2) \otimes \mf{gl}_N)$.
\end{remark}

This follows a pattern studied in \cite{Cos13}.  In that paper, I constructed a $4$-dimensional field theory such that the Koszul dual of the algebra of local operators is the Yangian.  I showed that this $4$-dimensional theory is related to the Yangian in the same way that Chern-Simons theory is related to the quantum group.  The $5$-dimensional Chern-Simons type theory considered in this paper thus bears the same relationship to the two-variable counterpart to the Yangian.

\subsection{$M2$ branes}
The other part of our story concerns $M2$ branes.  As before, let us consider $M$-theory on $TN_{N-1, \eps,\delta} \times \R^2_{-2\delta} \times \R \times \C^2$,  and let us place an $M2$ brane wrapping $\R^2_{-2\delta} \times \R$, and placed at the singular point of $TN_{N-1}$ and some point in $\C^2$.  Since we are in the $\Omega$-background, we find that the $M2$ brane has become effectively a one-dimensional theory.  We will show, using standard techniques from string theory and $M$-theory, that this theory is a holomorphic version of quantum mechanics built from a moduli space of instantons on $\C^2$.  

Let us describe this in more detail.  Let $\mc{M}_{N,K}^\eps$ denote the moduli space of torsion-free sheaves of rank $N$ and charge $K$ on the non-commutative deformation of $\C^2$ where $[z_1,z_2] = \eps$.  These sheaves are also framed at $\infty$.  We can describe $\mc{M}_{N,K}^{\eps}$ as the Nakajima quiver variety associated to the ADHM quiver 
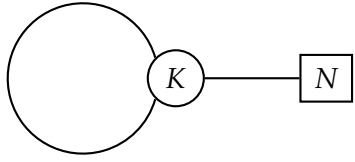
\begin{center}
\begin{tikzpicture}
  \node[circle,draw,thick](NL) at (0,0){$K$};
  \node[rectangle,draw,thick, inner sep=5](NR) at (2,0){$N$};
  \draw[thick](NL.east)--(NR.west);
  
  \draw[thick](NL.north west) arc (16:344:1) ;  

\end{tikzpicture}
\end{center}
where we set the moment map to $\eps \op{Id}_{gl(K)}$ when we form the symplectic quotient defining the Nakajima quiver variety.  

More concretely, $\mc{M}_{N,K}^{\eps}$ is the variety consisting of quadruples
\begin{align*} 
B_1,B_2& \in \mf{gl}(K) \\
I& \in \Hom(\C^N,\C^K) \\
J& \in \Hom(\C^K, \C^N),   
\end{align*}
satisfying the moment map equation
$$
[B_1,B_2] + IJ = \eps \op{Id}_{\mf{gl}(K)} 
$$
and taken up to the natural action of $GL(K)$. 

The algebra of local operators on  the $M2$ brane, in the $\Omega$-background we consider, is a non-commutative deformation $\Oo_{\delta}(\mc{M}_{N,K}^{\eps})$ of the algebra $\Oo(\mc{M}_{N,K}^{\eps})$ of polynomial holomorphic  functions on the variety $\mc{M}_{N,K}^{\eps}$.  This non-commutative deformation is defined by quantum Hamiltonian reduction. 

Generators for the algebra of functions $\Oo_{\delta}(\mc{M}_{N,K}^{\eps})$ are given by the expressions 
$$
\op{Tr}_{\C^N} \left( M  J B_1^k B_2^l I\right) 
$$
for a matrix $M \in \mf{gl}_N$, and integers $k,l \in \Z_{\ge 0}$.  For $K$ sufficiently large these generators become independent and match, in a certain limit, with the generators $M z_1^k z_2^l$ of $U(\mf{gl}_N[z_1,z_2])$.   

From the point of view of physics, these operators give holomorphic functions on the Higgs branch  of the moduli of the $3d$ $N=4$ theory on the $M2$ brane.  From the standard picture of holography, all these functions can be thought of in terms of the open-string states on a $D6$ brane. This is why they are captured by our $5d$ gauge theory, which is the theory on a $D6$ brane in an $\Omega$-background and in the presence of a $B$-field.  Of course, all the states on the $D6$ brane come from states of $M$-theory on the Taub-NUT. In \cite{Cos16} it was shown that variation of the Taub-NUT radius is $Q$-exact for the supersymmetry we are using, so that the $D6$ brane captures all the relevant information.      

Let me mention some literature relevant to this study of the large $K$ limit of the algebra of operators on the $M2$ brane. Our analysis is connected, by a chain of string dualities,  to a somewhat related large $K$ analysis performed by Koroteev and Sciarappa in their interesting series of papers \cite{KorSci16}. In addition, the algebras $\Oo_{\delta}(\mc{M}_{N,K}^{\eps})$ and their representation theory have been studied in detail by I. Losev \cite{Los16} in the context of symplectic duality.

\subsection{}
To prove our formulation of the AdS/CFT correspondence, we need to relate the large $K$ limit of the algebra $\Oo(\mc{M}_{N,K}^{\eps})$ to the algebra $U^{QFT}_{\delta}(\Oo_\eps(\C^2) \otimes \mf{gl}(N))$, which is the Koszul dual of the algebra of observables of the $5$-dimensional non-commutative gauge theory. 

We do this via the introduction of an intermediate algebra.  In section \ref{section_combinatorial_algebra} I construct a purely combinatorial definition of an algebra $U^{comb}_{\delta}(\Oo_\eps(\C^2)\otimes \mf{gl}(N) )$, defined for $\eps \neq 0$, which deforms the universal enveloping algebra $U(\Oo_\eps(\C^2) \otimes \mf{gl}(N) )$.  There are algebra homomorphisms
$$
U^{comb}_{\delta}(\Oo_\eps(\C^2) \otimes \mf{gl}(N) )\to \Oo_\delta(\mc{M}_{N,K}^{\eps})
$$ 
for all $K$.  There is a central parameter $\kappa \in  U^{comb}_{\delta}(\Oo_\eps(\C^2)\otimes \mf{gl}(N) )$, lifting the central element $1 \otimes \op{Id}_{\mf{gl}(N)}$ in the Lie algebra $\Oo_\eps(\C^2) \otimes \mf{gl}(N)$.  This homomorphism sends $\kappa$ to the constant $K$. 

The following result allows us to think of $U^{comb}_{\delta}(\Oo_\eps(\C^2) \otimes \mf{gl}(N))$ as a kind of large-$K$ limit of the algebras $\Oo_{\delta}(\mc{M}_{N,K}^{\eps})$.  
\begin{theorem} 
The map
$$
U^{comb}_{\delta}(\Oo_\eps(\C^2) \otimes \mf{gl}(N) )\to \Oo_\delta(\mc{M}_{N,K}^{\eps})
$$
is surjective for all $K$, and as $K \to \infty$ it identifies $\Oo_{\delta}(\mc{M}_{N,K}^{\eps})$ with the quotient of $U^{comb}_{\delta}(\Oo_\eps(\C^2) \otimes \mf{gl}(N))$ where we set the central parameter $\kappa$ to $K$.
\end{theorem} 

\subsection{}
This theorem reduces the proof of our formulation of the AdS/CFT correspondence to finding a relationship between the two different deformations $U^{comb}_{\delta}(\Oo_\eps(\C^2) \otimes \mf{gl}(N))$ and $U^{QFT}_{\delta}(\Oo_\eps(\C^2) \otimes \mf{gl}(N))$  of the universal enveloping algebra $U(\Oo_\eps(\C^2) \otimes \mf{gl}(N))$. This turns out to be the hardest part of the argument.  

The main difficulty in finding such an isomorphism is that, while the combinatorially defined algebra is completely explicit, it is very hard to compute exactly with the algebra defined from quantum field theory.  This is because the structure constants of this algebra are expressed in terms of Feynman diagrams which can have an arbitrary number of loops.   

We thus have to resort to more abstract methods to relate these two algebras.  The idea is to try to show that there is a unique, or at least a small number, of deformations of the algebra $U(\Oo_\eps(\C^2) \otimes \mf{gl}(N))$, and use this to show that the two algebras must  be isomorphic.  In physics terminology, this kind of argument is related to the ``bootstrap'' method, which uses associativity properties of the OPE of a conformal field theory to constrain the theory. 

Deformations of $U(\Oo_\eps(\C^2) \otimes \mf{gl}(N))$ are described by Hochschild cohomology of this algebra.  Unfortunately, the Hochschild cohomology of this algebra is very hard to compute, at least for small $N$. For example, when $N = 1$, it is closely related to the Lie algebra cohomology of the Lie algebra $\Oo_\eps(\C^2)$, the computation of which is a famously difficult problem \cite{Fuk12}.   When $N$ is large, however, these Hochschild cohomology groups are more tractable, and can be computed using a version of the Loday-Quillen-Tsygan \cite{LodQui84,Tsy83} theorem.   

We solve the problem of constraining deformations of $U(\Oo_\eps(\C^2) \otimes \mf{gl}(N))$ for all values of $N$ by studying the corresponding problem when the Lie algebra $\mf{gl}(N)$ is replaced by the super Lie algebra $\mf{gl}(N+R \mid R)$.  Any deformation of $U(\Oo_\eps(\C^2) \otimes \mf{gl}(N+R \mid R) )$ gives rise to one of $U(\Oo_\eps(\C^2) \otimes \mf{gl}(N+R-1\mid R-1))$, using a trick I learned in \cite{MikWit14}. If we  choose an element $Q \in \mf{gl}(N+R \mid R)$ of rank $(0\mid 1)$, then $Q^2 = 0$ and the $Q$-cohomology of $\mf{gl}(N+R \mid R)$ is $\mf{gl}(N+R-1 \mid R-1)$.  Similarly, the $Q$-cohomology of any deformation of $U(\Oo_\eps(\C^2) \otimes \mf{gl}(N+R \mid R) )$ will be a deformation of $U(\Oo_\eps(\C^2) \otimes \mf{gl}(N+R -1 \mid R-1) )$.  

One can thus ask for compatible sequences of deformations of $U(\Oo_\eps(\C^2)\otimes \mf{gl}(N+R \mid R) )$ for all $R$.   The group describing such deformations is the large-$R$ limit of the Hochschild cohomology groups of $U(\Oo_\eps (\C^2) \otimes \mf{gl}(N+R \mid R) )$. We can compute these Hochschild cohomology  groups completely explicitly in the large $R$, using a version of the Loday-Quillen-Tsygan theorem.  Before stating the theorem, note that since $U(\Oo_\eps(\C^2) \otimes \mf{gl}(N+R \mid R))$ has a central element $\kappa$,  we can multiply every deformation by any power of $\kappa$. Thus, the space of first order deformations is a module for $\C[\kappa]$.  
\begin{theorem*}
The uniform-in-$R$ second Hochschild cohomology group of $U(\Oo_\eps(\C^2) \otimes \mf{gl}(N+R \mid R))$ is the rank $1$ module $\C[\kappa]$ for the ring $\C[\kappa]$. 
\end{theorem*}  
This is the best result one can hope for: up to multiplying by powers of $\kappa$ there is only one first-order deformation that is defined uniformly in $R$.

We show that both our combinatorially defined algebra $U^{comb}_{\delta}(\Oo_\eps(\C^2)\otimes \mf{gl}(N))$ and the algebra $U^{QFT}_{\delta}(\Oo_\eps(\C^2)\otimes \mf{gl}(N))$  defined using quantum field theory can also be defined if we replace $\mf{gl}(N)$ by $\mf{gl}(N+R \mid R)$, and they are related for different values of $R$ in the way described above.  In each algebra, let us fix $\eps$ and view the algebra as a family over $\C[[\delta]]$.  The uniqueness statement for uniform-in-$R$ deformations thus leads to the following.
\begin{corollary*}
For all $N$, there is an isomorphism of associative algebras
$$
\Phi : U^{QFT}_{\delta_{QFT}}(\Oo_\eps(\C^2)\otimes \mf{gl}(N)) \iso U^{comb}_{\delta_{comb}}(\Oo_\eps(\C^2)\otimes \mf{gl}(N)).
$$
This isomorphism sends
$$
\delta_{QFT} \mapsto \lambda \delta_{comb} + f_2(\kappa) \delta_{comb}^2 + f_3(\kappa) \delta_{comb}^3 + \dots
$$
where $f_i(\kappa)$ are polynomials in the central element $\kappa$ of degree at most $i-1$, and $\lambda$ is a non-zero constant which is determined by a Feynman diagram computation to be
$$
\lambda = 2^{5} \pi^{2}. 
$$ 
\end{corollary*} 
This corollary thus tells us that the two algebras are isomorphic after performing some change of coordinates on the central parameters $\delta$ and $\kappa$. The fact that the polynomials $f_i(\kappa)$ are of order at most $i-1$ is not a formal consequence of the calculation of Hochschild cohomology.  Instead, it follows from a more refined analysis of the algebras we are considering. 

Since the algebra $U_{\delta}^{comb}(\Oo_\eps(\C^2) \otimes \mf{gl}(N))$ maps to the algebras $\Oo_{\delta}(\mc{M}^\eps_{N,K})$ for all $K$, in such a way that the central parameter $\kappa$ gets sent to $K$, we deduce from this that there are algebra homomorphisms
$$
 U^{QFT}_{\delta}(\Oo_\eps(\C^2)\otimes \mf{gl}(N)) \to  \Oo_{F(\delta,K)}(\mc{M}^\eps_{N,K}) 
$$
where
$$
F(\delta,K) = 2^{-5} \pi^{-2}  \delta + \delta^2 g_2(K) + \dots
$$
and the $g_i$ are polynomials in $K$ of degree at most $i-1$.   These homomorphisms are all surjective, and as $K \to \infty$ the kernel of this homomorphism becomes the ideal which sets the central parameter $\Phi^{-1}(\kappa)$ to $K$. 

To translate into physics terminology, what this shows is that the algebra we find from our non-commutative gauge theory is a scaling limit of the algebra of operators of the theory on a stack of $K$ $M2$ branes, where the coupling constant of the theory on the $M2$ branes depends on $K$ and  $\delta$ in a non-trivial way.  

\begin{remark}
There is an ambiguity in the quantization of the $5$-dimensional gauge theory, but two different quantizations  are related by a reparameterization of $\delta$.  This tells us that, when  we set $\kappa$ to zero, we can choose a quantization of the $5$-dimensional gauge theory so that there is an isomorphism of algebras
$$
U^{QFT}_{2^{-5} \pi^{-2} \delta} (\Oo_\eps(\C^2) (\mf{gl}_N) / (\kappa) \iso U^{comb}_{\delta}(\Oo_\eps(\C^2) \otimes \mf{gl}_N )/(\kappa) 
$$
with no unknown change of coordinates. Since any quantization of the field theory is as good as any other (they only differ by how we coordinatize the space of coupling constants), we can always quantize our field theory so that this holds.  

When we do not set $\kappa = 0$, we can in the same way choose the quantization of our field theory so that the constant term in each polynomial $f_i(\kappa)$ is zero.  It is not possible to adjust the remaining terms in these polynomials by different choices of quantization of the field theory.  

It is natural to conjecture that these polynomials are all zero. I have no idea how one could attempt to prove this conjecture except by brute-force computations in the quantum field theory.  
\end{remark}

\subsection { Categorified Donaldson-Thomas invariants and $5$-dimensional gauge theory}
\label{subsection_DT_intro}
Now that I have outlined the technical results contained in this paper, let me discuss someconjectural links between the algebras considered here and enumerative geometry.

In \cite{GopVaf98}, Gopakumar and Vafa analyzed the Gromov-Witten invariants of a Calabi-Yau three-fold by considering the $5$-dimensional field theory obtained from $M$-theory compactified on the Calabi-Yau.    Their analysis involved BPS objects in this $5$-dimensional theory. Their work led to many important developments in mathematics, including the introduction of Donaldson-Thomas and related invariants. 

The $5$-dimensional theory considered here is, according to the analysis of \cite{Cos16}, a supersymmetric twist of the theory obtained from an compactifying $M$-theory on certain toric Calabi-Yaus, in an $\Omega$-background (that is, equivariantly).  The $5$-dimensional theory has parameters corresponding to the two equivariant parameters in the toric Calabi-Yau.   It is thus natural to expect that the $5$-dimensional theory associated to a toric Calabi-Yau ``controls'' the equivariant enumerative geometry of the same Calabi-Yau.

In Gopakumar-Vafa's analysis, curve-counting invariants of the Calabi-Yau  arise from $M$-theory compactified on a circle.  If we did not put $M$-theory on a circle, and instead considered the Hilbert space of states, we would find categorified versions of curve-counting invariants, such as categorified Donaldson-Thomas invariants.  We would thus expect that the $5$-dimensional field theories studied here, and the corresponding two-variable quantum groups, should be related to categorified DT invariants.    In this section I will state a conjecture along these lines.

Let 
$$
X_N = \til{\C^2 / \Z_N}
$$
be the resolution of the $A_{N-1}$ surface singularity.   Consider the equivariant, categorified Donaldson-Thomas theory of $\C \times X_N$.  This is the equivariant cohomology of the Donaldson-Thomas moduli of ideal sheaves on $\C \times X_N$ with coefficients in the appropriate sheaf of vanishing cycles.  We work equivariantly with respect to the action of $\C^\times \times \C^\times$ which preserves the Calabi-Yau volume form on $\C \times X_N$.  We let $\delta$ denote the equivariant parameter associated to the $\C^\times$ action which rotates $\C$ and scales the symplectic form on $X_N$, and $\eps$ the equivariant parameter associated to the $\C^\times$ action which preserves $\C$ and fixes the symplectic form on $X_N$. We invert the equivariant parameter $\eps$ but not $\delta$. 

The moduli of ideal sheaves on $\C \times X_N$ has components labelled by the Chern classes of the ideal sheaf.    We denote the equivariant categorified Donaldson-Thomas invariants of $\C \times X_N$, summing over all components, by $\mc{H}(\C \times X_N)$.  This is a module over the ring $\C[[\delta]]((\eps))$ of series in $\eps,\delta$. 

The module $\mc{H}(\C \times X_N)$ is $\Z^N$-graded, where one of the $\Z$-gradings corresponds to the number of points and the remaining $N-1$ $\Z$-gradings correspond to curve classes in $H_2(X_N,\Z) = \Z^{N-1}$.  We interpret this grading as a $\C^\times \times (\C^\times)^{N-1}$ action on $\mc{H}(\C \times X_N)$. We identify $(\C^\times)^{N-1}$ with the maximal torus $H \subset SL_N$. 

The algebra $U_\delta(\Oo_\eps(\C^2) \otimes \mf{gl}_N)$ is an algebra over the ring $\C[[\delta]]((\eps))$.    This algebra also has an action of $\C^\times \times H$, where $\C^\times$ acts as a symplectic symmetry of $\C^2$, and the maximal torus $H \subset SL_N$ acts on $\mf{gl}_N$ by the adjoint action.
\begin{conjecture}
\begin{enumerate}
\item
The algebra 
$$U_\delta(\Oo_\eps(\C^2) \otimes \mf{gl}_N)/(\kappa = 0)$$ acts on the equivariant categorified Donaldson-Thomas space $\mc{H}(\C \times X_N)$ by correspondences.  
\item This action takes place in the category of $\C[[\delta]]((\eps))$-modules.  
\item The action is compatible with the action of the torus $\C^\times \times H$ on both sides.
\item There is an embedding of Kontsevich-Soibelman's cohomological Hall algebra \cite{KonSoi10} associated to $\C \times X_N$ into the algebra $U_\delta(\Oo_\eps(\C^2) \otimes \mf{gl}_N)$, in such a way that the action of  $U_\delta(\Oo_\eps(\C^2) \otimes \mf{gl}_N)$ on the categorified Donaldson-Thomas invariants restricts to the natural action of the cohomological Hall algebra on the same space.  
\end{enumerate}
\end{conjecture}
\begin{remark}
\begin{enumerate}
\item In this conjecture I do not distinguish between the two versions of the algebra $U_\delta(\Oo_\eps(\C^2) \otimes \mf{gl}_N)$ we have discussed above. At $\kappa = 0$ the algebras are isomorphic up to a change of coordinates in the parameter $\delta$.  In practise, the conjecture can really only be checked for the explicit combinatorial version of the algebra. 
\item The choice of stability structure plays a role in the definition of the categorified DT invariants and the Konstevich-Soibelman cohomological Hall algebra (COHA). I expect that the statement should hold for all stability structures, so that the COHA for any choice of stability structure is embedded in $U_\delta(\Oo_\eps(\C^2) \otimes \mf{gl}_N)$.    As we will see shortly, the choice of stability structure seems to be  related to the choice of positive roots in the algebra $U_\delta(\Oo_\eps(\C^2) \otimes \mf{gl}_N)$.  
\end{enumerate}
\end{remark}
There are also variants of this conjecture which apply to Pandharipande-Thomas rather than Donaldson-Thomas invariants.    In certain cases, as I explain in section \ref{section_DT}, the PT version of the conjecture can be derived from the results of this paper together with the $3$-dimensional mirror symmetry conjecture.  

There are further generalizations of this conjecture in which we study the $5$-dimensional gauge theory, not on $\R \times \C^2$, but on $\R \times X_K$, where $X_K$ is a resolution of an $A_{K-1}$ singularity.  In this case, we expect to be able to construct an algebra $U_\delta(\Oo_\eps(X_K) \otimes \mf{gl}_N)$ from the quantum field theory, and this algebra should  act on rank $K$ categorified DT invariants of $\C \times X_{N}$, in the same way.

\section{Open problems}
Before moving to the proofs of the results mentioned in the introduction, let me list some natural questions and conjectures that this work suggests. Most of these are questions in pure algebra (with no reference to $M$-theory).  

\subsection{}
This paper is not the first paper to consider two-variable quantum groups deforming $U(\mf{g}[z_1,z_2])$ for a Lie algebra $\mf{g}$: see \cite{Gua07, GuaYan16}.
\begin{question*}
Can one prove that the combinatorial algebra $U^{comb}_{\delta}(\Oo_\eps(\C^2) \otimes \mf{gl}_N)$ is isomorphic to the quantum double current algebras considered by Guay and his collaborators?  
\end{question*}
  
\subsection{}
Our $5$-dimensional quantum field theory is defined not just on $\R \times \C^2$, but on a variety of spaces of the form $\R \times X$ for a complex symplectic surface $X$.  In particular we can take $X$ to be a resolution of an ADE singularity. The algebra of functions $\Oo(X)$ on the resolved ADE singularity has a deformation with $r+1$ parameters (where $r$ is the rank of the ADE group). Let us call this deformation $\Oo_{\eps,\lambda_1,\dots,\lambda_r}(X)$.  By applying Koszul duality to the algebra of observables of the theory on $\R \times X$, one finds a deformation $U^{QFT}_{\delta} ( \Oo_{\eps,\lambda_1,\dots,\lambda_r}(X) \otimes \mf{gl}_N)$ of the universal enveloping algebra of $\Oo_{\eps,\lambda_1,\dots,\lambda_r}(X)\otimes \mf{gl}(N)$.   It is natural to make the following conjecture.
\begin{conjecture}
When $X$ is a resolution of an ADE singularity, the algebra $U^{QFT}_{\delta} ( \Oo_{\eps,\lambda_1,\dots,\lambda_r}(X) \otimes \mf{gl}_N)$ is the large $K$ limit of a deformation quantization of the moduli of framed torsion-free sheaves of charge $K$ and rank $N$ on the non-commutative deformation of $X$.

In particular, in the case $N=1$, the algebra $U^{QFT}_{\delta} ( \Oo_{\eps,\lambda_1,\dots,\lambda_r}(X))$ should be the large $K$ limit of the symplectic reflection algebra associated to the action of the wreath product $S_K \ltimes (\Gamma)^K$ on $\C^{2K}$, where $\Gamma$ is the finite subgroup of $SU(2)$ associated to the ADE singularity.  
\end{conjecture}
When $X$ is a resolved ADE singularity, an algebra quantizing the universal enveloping algebra of a central extension of $\Oo(X) \otimes \mf{sl}_N$ was studied by Guay, Hernandez and Loktev in \cite{GuaHerLok09}.  It is natural to hope that the algebra studied by these authors is related to the one given by quantum field theory.

Hopefully, one could generalize the arguments presented in this paper to prove a statement along these lines, but it seems to me to be non-trivial. As a first step, one should find a combinatorially-defined algebra which is the large $K$ limit of a deformation quantization of the moduli of sheaves on $X$.  One should then show that in an appropriate limit, this combinatorial algebra degenerates to the universal enveloping algebra of $\Oo_{\eps,\lambda_1,\dots,\lambda_r}(X) \otimes \mf{gl}_N$. Next, one should try to generalize the uniqueness result used here to show that the universal enveloping algebra of $\Oo_{\eps,\lambda_1,\dots,\lambda_r}(X) \otimes \mf{gl}_N$ admits an $r+1$-dimensional family of deformations which are defined uniformly in $N$, as we saw in the case that $X = \C^2$.  All of these steps require some work, but entirely in algebra.  

Arguments similar to the one presented in this paper suggest that, in the case $X$ is a resolved $A_{r-1}$ singularity, the the algebra $U_\delta( \Oo_{\eps,\lambda_1,\dots,\lambda_r}(X) \otimes \mf{gl}_N)$ acts on the moduli space of quasi-maps to the moduli of rank $r$ torsion free sheaves on a resolved $A_{N-1}$ singularity.

\subsection{}
The $5$-dimensional theory considered here is related to the $4$-dimensional theory of \cite{Cos13} via a certain dimensional reduction. This $4$-dimensional theory is related to the Yangian.  Consideration of how dimensional reduction acts on the algebra of operators leads to the following conjecture relating the algebras studied here to the Yangian. 

Let us give the algebra $U^{comb}_{\delta}(\Oo_\eps(\C^2) \otimes \mf{gl}_N)$ a $\C^\times$ action arising from the $\C^\times$ action on $\C^2$ which scales the coordinate with opposite weight. Given any algebra $A$ with a $\C^\times$ action, or equivalently a grading, one can define a new algebra $\mc{B}(A)$, called the  $B$-algebra, by
$$
\mc{B}(A) = A^0 / \left(\oplus_{i > 0} A^i \cdot A^{-i}  \right). 
$$
Here $A^i$ refers to the $i$th graded piece. The algebra is the quotient of $A^0$ by the two-sided ideal consisting of the product of an element of degree $i$ with an element of degree $-i$. 
\begin{conjecture}
There is an isomorphism  between
$$
\mc{B}(U^{comb}_{\delta}(\Oo_\eps(\C^2) \otimes \mf{gl}_N) ) \iso Y(\mf{gl}_N) 
$$
between the $B$-algebra of $ U^{comb}_{\delta}(\Oo_\eps(\C^2) \otimes \mf{gl}_N)$ and the Yangian for $\mf{gl}_N$.  
\end{conjecture} 
More generally, given a cocharacter $\rho$ of $SL(N)$, one gets a $\C^\times$ action on the algebra $U^{comb}_{\delta}(\Oo_\eps(\C^2) \otimes \mf{gl}_N)$ by combining the action on the plane $\C^2$ used before with the adjoint action of the chosen cocharacter on $\mf{gl}_N$.   Using this $\C^\times$-action one finds a different $B$-algebra, which we denote by $ \mc{B}_{\rho}(U^{comb}_{\delta}(\Oo_\eps(\C^2) \otimes \mf{gl}_N) )$ .
\begin{conjecture}
There is an isomorphism
$$
\mc{B}_{\rho}(U^{comb}_{\delta}(\Oo_\eps(\C^2) \otimes \mf{gl}_N) ) \iso Y_{\rho}(\mf{gl}_N)
$$
where $Y_\rho(\mf{gl}_N)$ is the shifted Yangian defined  by \cite{BruKle06,KamWebWeeYac14} associated to the cocharacter $\rho$.  
\end{conjecture}
Both of these conjectures could, in principle, be checked by analyzing the combinatorial presentation of the algebra $U^{comb}_{\delta}(\Oo_\eps(\C^2))$ given in this paper. 

\subsection{}
Given a Riemamn surface $\Sigma$, the $5$-dimensional theory we are considering can be defined on $\R \times T^\ast \Sigma$, at the quantum level \cite{Cos16}.  If $\Sigma$ is an affine algebraic curve, then the observables of the theory on an interval in $\R$ times all of $T^\ast \Sigma$ can be identified with $C^\ast(\mf{gl}_N \otimes \op{Diff}(\Sigma))$.

One can use Koszul duality to construct from this a deformation of $U(\mf{gl}_N \otimes \op{Diff}(\Sigma) )$.  (Applying Koszul duality to the algebra constructed from the field theory is not entirely trivial, because it is built from holomorphic functions on $T^\ast \Sigma$ rather than polynomials.  One can show, however, that the quantum field theory constructions make sense if one considers fields of polynomial growth at at the boundary of $T^\ast \Sigma$. If one does this, one finds that the algebra of observables is quasi-isomorphic to cochains of polynomial differential operators on $\Sigma$, so there are no difficulties with defining Koszul duality).  

These quantum-field theory considerations lead us to the following conjecture.
\begin{conjecture}
There is a deformation of $U(\mf{gl}_N \otimes \op{Diff}(\Sigma) ) $ coming from the $5$-dimensional quantum field theory on $\R \times T^\ast\Sigma$, which can be described in terms of the large $K$ limit of a deformation quantization of the moduli of instantons on $T^\ast \Sigma$.  
\end{conjecture}
To be precise, the existence of the algebra can be proved by QFT considerations.  What one does not get for free from QFT is a simple presentation of the algebra, or its relation to moduli of instantons on $R \times T^\ast \Sigma$.

In the case $N=1$, the deformation of $U(\op{Diff}(\Sigma))$ should be related to the large $K$ limit of the spherical Cherednik algebra associated to the curve $\Sigma$ in the work of Etingof \cite{Eti04}.  

It is natural to further ask whether these algebras can be  defined when we use other affine holomorphic symplectic surfaces, such as the complement of an  elliptic curve in $\mbb{P}^2$ or of a divisor in a $K3$.  I don't know the answer to this, because currently I don't know  whether the  $5$-dimensional quantum field theory can be defined on such backgrounds. 

\subsection{}
In \cite{Cos13} I explained how, in a certain four-dimensional theory which is topological in two directions and holomorphic in one complex direction, the Yangian arises as the Koszul dual of the algebra of operators, equipped with the operator product in a topological direction.  Moreover, the $R$-matrix of the Yangian arises  from the operator product in the holomorphic direction.

This suggests that there should be a similar here. Our $5$-dimensional gauge theory is topological in one direction and holomorphic in two complex directions. The two-variable quantum group  we consider is the Koszul dual of the algebra of operators equipped with the operator product in the topological direction.  The OPE in the two complex directions should equip this algebra with an interesting and novel algebraic structure, which should be worth exploring.  I hope to discuss this structure in subsequent work.

\section{How to read this paper}
This paper is mostly algebra, but also uses the $5$-dimensional quantum field theory related to $M$-theory, and its algebra of observables.  Mathematicians who have no interest in quantum field theory could skip sections \ref{section_5d_gauge_theory}--\ref{section_ope}, which are the only sections of the paper dealing with quantum field theory.  Starting at section \ref{section_combinatorial_algebra}, the paper is exclusively algebraic.  Section \ref{section_combinatorial_algebra} gives an explicit description of an algebra deforming $U(\mf{gl}_N \otimes \op{Diff}(\C))$.  The subsequent sections explore the relation of this algebra to  deformation quantizations of instanton moduli spaces, and prove a theorem regarding the  uniqueness of this deformation of $U(\mf{gl}_N \otimes \op{Diff}(\C)) $.

\section{Acknowledgements}
I am very grateful to many people for helpful conversations about the mathematical and physical aspects of this project. I'd like to particularly thank Sasha Braverman, Tudor Dimofte, Chris Dodd, Pavel Etingof, Davide Gaiotto, Sachin Gautam, Nicolas Guay, Si Li, Ivan Loseu and Yaping Yang.

This work is supported by Perimeter Institute for Theoretical Physics. Research at Perimeter Institute is supported by the Government of Canada through Industry Canada and by the Province of Ontario through the Ministry of Research and Innovation.  This work was also supported in part by the NSERC Discovery Grant program.
 
\section{A $5d$ gauge theory}
\label{section_5d_gauge_theory}
In this section we will describe the $5d$ gauge theory for the group $\gl_N$ which is obtained by an equivariant reduction of $M$-theory on an $A_{N-1}$ singularity times $\R^2$.

Let's describe the theory first in the traditional formulation, with fields, an action functional, and gauge  symmetry. Then we will explain how to rewrite it in the BV formalism.   We will start by describing the theory on $\R \times \C^2$, and then explain how to write it on $\R \times X$ for a class of complex surfaces $X$.

The theory is a non-commutative gauge theory, meaning that the complex plane $\C^2$ is non-commutative.  If $\Oo(\C^2)$ denotes the ring of holomorphic functions on $\C^2$, recall that the Moyal product on $\Oo(\C^2)$ is a non-commutative product of the form
$$
f \ast_c g = f g + c  \tfrac{1}{2} \eps_{ij} \dpa{z_i} f \dpa{z_j} g + c^2 \tfrac{1}{2^2 \cdot 2!} \eps_{i_1 j_1} \eps_{i_2 j_2} \left( \dpa{z_{i_1}} \dpa{z_{i_2}} f\right) \left( \dpa{z_{j_1}} \dpa{z_{j_2}} g \right) + \dots
$$
where $c$ is a formal parameter, $\eps_{ij}$ is the alternating symbol, and we have used the summation convention.    The coefficient of $c^n$ in the expansion is 
$$
\frac{1}{2^n n!} \prod_{r = 1}^{n} \eps_{i_r j_r} \left(  \prod_{r = 1}^{n} \dpa{z_{i_r}} f  \right) \left( \prod_{r = 1}^{n} \dpa{z_{j_r}}  g \right). 
$$
This product can be extended to a product on the Dolbeault complex $\Omega^{0,\ast}(\C^2)$, by the same formula except that product of holomorphic functions is replaced by wedge product of Dolbeault forms.  The $\dbar$ operator on $\Omega^{0,\ast}(\C^2)$ is a derivation for the Moyal product, so that $\Omega^{0,\ast}(\C^2)[[c]]$ becomes a non-commutative differential graded algebra.

The fundamental field of our $5$-dimensional theory is a partial connection
$$
A \in \left(  \Omega^1(\R \times \C^2) / \left(\d z_1 \Omega^0 \oplus \d z_2 \Omega^0 \right) \right) \otimes \gl_N. 
$$
Thus, $A$ has $3$ components,
$$
A = A_0 \d t + A_{1} \d \zbar_1 + A_{2} \d \zbar_2
$$
where $A_0$, $A_1$ and $A_2$ are $\gl_N$-valued smooth functions on $\R \times \C^2$.

The action functional is 
\begin{align*} 
S(A) &= \int_{\R \times \C^2} \d z_1 \d z_2 \op{Tr} \left\{ \tfrac{1}{2} A \d A + \tfrac{1}{3} A (A \ast_c A) \right\}  \\
     &=  \int_{\R \times \C^2} \d z_1 \d z_2 \op{Tr} \left\{ \tfrac{1}{2} A \ast_c \d A + \tfrac{1}{3} A \ast_c A \ast_c A) \right\} 
\end{align*}
where $c$ is a coupling constant, and is treated as a formal parameter.   The expressions on the first and second lines are equivalent because the difference between them is a total derivative. 

Note that if we try to treat $c$ as a non-zero number, we find a non-local action functional.  While each term in the expansion of $c$ is local, the coefficient of $c^n$ has $2n$ derivatives, so that summing up the coefficients of $c^n$ will lead to a non-local expression.  Also, note that this theory is very non-renormalizable, because the classical action functional contains arbitrarily high derivatives.  Nevertheless, it was proved in \cite{Cos16} that the theory can be quantized (using techniques developed in \cite{CosLi15}).

The Lie algebra of infinitesimal gauge transformations of the theory is the space $\Omega^0(\R \times \C^2) \otimes \gl_N$, but equipped with the Lie bracket
$$
[f,g] = f \ast_c g - g \ast_c f
$$
which is the commutator in the tensor product of $\Omega^0(\R \times \C^2)$ with it's $c$-dependent Moyal product with the associative algebra $\gl_N$.  An infinitesimal gauge transformation $f$ acts on a field $A$ by 
$$
A \mapsto A + \eps \left( \d f + [f,A] \right) 
$$
where, again, the commutator $[f,A]$ uses the Moyal product.  Here by $\d f$ we mean the image of the de Rham differential in $f$ in the quotient of $\Omega^1(\R \times \C^2)$ by the subspace spanned by $\d z_1$ and $\d z_2$.

Let us now explain how one can put this theory in the BV formalism.  Let us introduce the space
$$
\mc{A} = \Omega^\ast(\R \times \C^2) / \langle \d z_1, \d z_2 \rangle
$$
which is the quotient of the de Rham complex by the differential ideal generated by the $1$-forms $\d z_1$, $\d z_2$.  Thus, 
 $$
\mc{A} = \cinfty(\R \times \C^2) [ \d t, \d \zbar_1, \d \zbar_2] 
 $$
 where the parameters $\d t$, $\d \zbar_1$ and $\d \zbar_2$ are of degree $1$. 

 Evidently, $\mc{A}$ is a commutative algebra and the de Rham operator $\d_{dR}$ descends to a differential $\d_{\mc{A}}$ on $\mc{A}$.  Explicitly, 
 $$
\d_{\mc{A}} = \d t \dpa{t} + \sum  \d \zbar_i \dpa{\zbar_i}. 
 $$
 Further, the Moyal product on $\C^2$ gives a map
 \begin{align*} 
 \mc{A} \otimes \mc{A} &\to \mc{A}[[c]]\\
     \alpha \otimes \beta & \mapsto  \alpha \ast_c \beta
 \end{align*}
 defined by the same formula we gave earlier, but where product of holomorphic functions on $\C^2$ is replaced by the wedge-product in $\mc{A}$. This Moyal product makes $\mc{A}[[c]]$ into a differential-graded commutative algebra over $\C[[c]]$.

 Let $\mc{A}_c$ be the quotient of the space $\Omega^\ast_c(\R \times \C^2)$ of compactly supported forms by the ideal generated by $\d z_1$ and $\d z_2$.  Then, $\mc{A}_c[[c]]$ is equipped with an integration map
 \begin{align*} 
 \int : \mc{A}_c[[c]] &\to \C [[c]]\\ 
     \alpha & \mapsto  \int_{\R \times \C^2} \d z_1 \d z_2 \alpha.
 \end{align*}
 Note that $\int \d_{\mc{A}}\alpha = 0$ and 
 $$
\int \alpha \ast_c \beta = \pm \int \beta \ast_c \alpha
 $$
 for $\alpha,\beta \in \mc{A}$.  

 Thus, the space $\mc{A}_c$ has all the structure needed to define the kind of Chern-Simons action functional that appears in open-string field theory. The space of fields for this action functional is $\mc{A}_c \otimes \gl_N[1]$, and the action functional is
 \begin{align*} 
 S(\alpha) & = \tfrac{1}{2} \int \alpha \ast_c \d \alpha + \tfrac{1}{3} \int \alpha \ast_c \alpha \ast_c \alpha\\
               &= \tfrac{1}{2} \int \alpha  \d \alpha + \tfrac{1}{3} \int \alpha ( \alpha \ast_c \alpha).
 \end{align*}
The first and second lines are the same because the difference between them is a total derivative.

Then, $\mc{A}_c \otimes \gl_N[1]$ is the space of fields of our $5$-dimensional gauge theory in the BV formalism.  The BV action functional is the functional $S$ above. The odd symplectic structure is given by the formula
$$
\ip{\alpha,\beta} = \int \alpha \beta
$$
for $\alpha,\beta \in \mc{A}_c \otimes \gl_N[1]$.

\section{The theory on more general manifolds}
Let $X$ be a holomorphic symplectic complex surface, and suppose that $X$ is equipped with a $\C^\times$ action which scales the holomorphic symplectic form.  (Such an $X$ is called \emph{conical}).  For example, $X$ could be the cotangent bundle of a Riemann surface, or the resolution of an ADE singularity. In this section, we will explain how to put our theory on $\R \times X$.  
\begin{definition}
A $\ast$-product  on the sheaf $\Oo_X$ of holomorphic functions on $X$ is a map of sheaves $\Oo_X \otimes_{\C} \Oo_X \to \Oo_X[[c]]$ which satisfies the following properties.
\begin{enumerate} 
 \item It makes $\Oo_X[[c]]$ into a sheaf of associative algebras, quantizing the sheaf of Poisson algebras $\Oo_X$.
 \item The coefficient of $c^n$ in the product is given by a holomorphic bi-differential operator of finite order.
\item Suppose the $\C^\times$ action on $X$ scales the holomorphic symplectic form by some non-zero weight $k$. Let us give $\Oo_X[[c]]$ an action of $\C^\times$ by combining the given action on $X$ with the action on $\C[[c]]$ which scales $c$ with weight $k$.  We require that the associative product on $\Oo_X[[c]]$ must be compatible with this $\C^\times$-action.   
\end{enumerate}
\end{definition}
The third condition severely restricts the moduli of $\ast$-products on $\Oo_X$. If $X$ is the cotangent bundle of a Riemann surface or a resolution of an ADE singularity, then this condition implies that the space of $\ast$-products is finite-dimensional and is a torsor for $H^2(X)$.

In this situation, we can define our $5$-dimensional gauge theory on $\R \times X$, just as before.  We let $\mc{A}^{X}$ denote the quotient of $\Omega^\ast(\R \times X)$ by the differential ideal generated by $\Omega^{1,0}(X)$, and $\mc{A}_c^{X}$ be the corresponding quotient of $\Omega^\ast_c(R \times X)$. Thus,
   $$
   \mc{A}^{X} = \Omega^\ast(\R) \what{\otimes} \Omega^{0,\ast}(X)
   $$
where $\what{\otimes}$ is the completed projective tensor product. (Using this completed tensor products simply means that the fields are the sections of a graded vector bundle on $\R \times X$ with fibres $\wedge^\ast (T^\ast \R \oplus \br{T}^\ast X)$. 

The fact that the $\ast$-product on $\Oo_X$ is implemented by holomorphic bi-differential operators means that it extends in a natural way to a $\ast$-product on $\mc{A}^{X}$, making $\mc{A}^{X}[[c]]$ into a differential graded associative algebra.  As before, there is an integration map on $\mc{A}_c^{X}$ defined by the formula
$$ 
\int \alpha = \int_{\R \times X} \omega \alpha 
$$
where $\omega$ is the holomorphic volume form on $X$. 

Thus, we can define a field theory where the fields are $\mc{A}_c^{X} \otimes \gl_N[1]$, and the action functional as before is
$$
S(\alpha)  = \tfrac{1}{2} \int \alpha \ast_c \d \alpha + \tfrac{1}{3} \int \alpha \ast_c \alpha \ast_c \alpha.
$$
This describes the theory on $X$ in the BV formalism.

\section{Quantization}
In this section we will state a theorem, proved in \cite{Cos16}, about when this theory can be quantized.  As we mentioned above, this theory is highly non-renormalizable, as the classical action functional has infinitely many terms with more and more derivatives. Even so, we find that consistency of the quantization -- the BV quantum master equation -- constrains the problem of quantization so tightly that there are only $2$ free parameters.

To quantize the theory, we will introduce a loop expansion parameter $\hbar$ and use the action $\hbar^{-1} S$.   We will then quantize order by order in $\hbar$.  In order to cancel certain higher loop anomalies, we will also need to allow negative powers of $c$.  This means that the quantum-corrected action functional is a series in $c$ and $\hbar$ where each positive power of $\hbar$ can be accompanied by some finite number of negative powers of $c$. More formally, this means we have a theory over the sub-ring of the base ring $\C((c))[[\hbar]]$ consisting of series $\sum \hbar^i f_i(c)$ where $f_i(c) \in \C((c))$ and $f_0(c) \in \C[[c]]$.

There are two versions of our quantization theorem, depending on whether we work on a conical surface $X$ or on $\C^2$. In the case that we work on $\C^2$, it is natural to ask that we quantize the theory in a way compatible with all the symmetries of $\C^2$, whereas on a general conical $X$ we can only ask that the quantization is compatible with the $\C^\times$ action on $X$.  Asking that our quantization is compatible with these extra symmetries on $\C^2$ means that the theorem is slightly different in this case. 

Let us first state the version of our theorem that applies for a conical symplectic surface $X$.  In fact, we will restrict to the case that $X$ is a resolution of an ADE singularity or the cotangent bundle to a Riemann surface.  To state the theorem, we should recall that every class $w \in H^2(X)$ gives rise to a first-order deformation of the $\ast$-product on the sheaf $\Oo_X$ of holomorphic functions on $X$ which we write as
$$
\alpha \ast_c \beta + \eps \alpha \ast_{c,w} \beta.
$$ 
This first-order deformation is compatible with the $\C^\times$-action on $X$ in the sense we discussed above. Further, every $\C^\times$-equivariant deformation of the $\ast$-product on $X$ is of this form.  
\begin{theorem}
Let $X$ be a conical complex symplectic surface which is either the cotangent bundle to a Riemann surface or a resolution of an ADE singularity. Consider our gauge theory on $\R \times X$ with gauge group $\mf{gl}_N$. Let us consider quantizing the theory in perturbation theory in the loop expansion parameter $\hbar$, and let us work ``uniformly in $N$'' as discussed in detail in \cite{CosLi15}.  Let us also ask that the quantization is compatible with the $\C^\times$-action on $X$. Finally, let us allow negative powers of $c$ as long as they are accompanied by positive powers of $\hbar$. Heuristically, this means that we should treat $c$ and $\hbar$ as both being small but with $\hbar$ much smaller than any positive power $c^n$ of $c$.  

Then, there are no obstructions (i.e.\ anomalies) to producing a consistent quantization.  At each order in $\hbar$, we are free to add $\op{dim} H^2(X) + 1$ independent terms to the action functional, leading to an ambiguity in the quantization.   Explicitly, the deformation of the action functional we can incorporate at $k$ loops is
 \begin{align*} 
 S & \mapsto S  + \hbar^{k}c^{-k}\left( \tfrac{1}{2} \int \alpha \d \alpha + \tfrac{1}{3} \int \alpha \ast_c \alpha \ast_c \alpha\right)\\
 S & \mapsto S +   \hbar^k c^{-k}\int \alpha \ast_{c} \alpha \ast_{c,w} \alpha 
\end{align*}
where $\ast_{c,w}$ indicates the first-order deformation of the $\ast$-product on $X$ associated to a class $w \in H^2(X)$. 

\end{theorem}

One way to summarize the result is the following. Recall that to define the action functional we need to give a $\ast$-product on the sheaf $\Oo_X$ of holomorphic functions on $X$, together with a holomorphic volume form $\omega$ on $X$.  To specify a quantization, we need to specify $\ast$-product and holomorphic volume form which both depend on $\hbar$. 
 
More formally,the moduli space of quantizations consists of pairs consisting of:
\begin{enumerate} 
\item  A $\hbar$-dependent associative $\ast$-product
$$
\alpha \ast_{c,\hbar} \beta = \sum_{k \ge 0} \hbar^k \alpha \ast_{c,k} \beta
$$
on $\Oo_X$, which at $\hbar = 0$ are the original $\ast$-product on $X$.
\item  A $\hbar$-dependent holomorphic volume $\omega_{\hbar} = \sum \hbar^k \omega_k$ on $X$, where $\omega_0$ is the original holomorphic volume form. 
\end{enumerate}
Further, the product $\ast_{c,\hbar}$ must be $\C^\times$-invariant, where both $c$ and $\hbar$ have weight $k$.  The holomorphic volume form $\omega_{\hbar}$ must be of weight $k$ under the $\C^\times$-action on $X$. This data is taken up to gauge equivalence.

The theorem is proved in \cite{Cos16}. The proof consists of calculating the cohomology groups describing obstructions and deformations to quantization.

Let us now describe the version of the theorem that applies to $\C^2$.  We will consider quantizations that are invariant under $\mf{sl}(2)$, compatible with the $\C^\times$ action scaling $\C^2$, and also holomorphically translation invariant. The latter means that they are translation invariant but that the translations $\dpa{\zbar_i}$ act in a homotopically trivial way (that is, in a way exact for the BRST operator). 

\begin{theorem}
Let us consider quantizing the theory on $\R \times \C^2$in perturbation theory in the loop expansion parameter $\hbar$, and let us work ``uniformly in $N$'' as discussed in detail in \cite{CosLi15}.  Let us also ask that the quantization is compatible with the $\C^\times$-action on $\C^2$, is $\mf{sl}(2)$-invariant, and is holomorphically translation invariant.  Also, let us allow negative powers of $c$ of the form $c^{-r} \hbar^n$ where $r \le n$ (heuristically this means that $c$ and $\hbar$ are both small and $\hbar$ is much less than $c$).  

Then, there are no obstructions (i.e.\ anomalies) to producing a consistent quantization.  At each order in $\hbar$, we are free to add $2$ independent terms to the action functional, leading to an ambiguity in the quantization.   Explicitly, the deformations of the action functional we can incorporate at $k$ loops are 
\begin{align*} 
S &\mapsto S + \hbar^{k}c^{-k}\tfrac{1}{2} \int \alpha \d \alpha + \tfrac{1}{3} \int \alpha \ast_c \alpha \ast_c \alpha\\ 
S &\mapsto S + \hbar^k c^{-k} c \dpa{c} \int \alpha \ast_c \alpha \ast_c \alpha.
\end{align*}
These two deformations of the action can be absorbed by a change of coordinates in $c$ and $\hbar$.
\end{theorem}
It is important to note that when we work on $\C^2$, we only need to involve the parameters $c$, $\hbar$ and $c^{-1} \hbar$.   If we let $\mu$ denote $c^{-1} \hbar$, then we can think of any quantization of the classical theory as depending on three parameters $c,\hbar$ and $\mu$ where $c \mu = \hbar$.  Further, there is essentially one quantization, up to changes of coordinates on these parameters of the form
\begin{align*} 
 \hbar & \mapsto \hbar f(\mu)\\
  c & \mapsto c g(\mu) \\
  \mu &\mapsto \mu f(\mu)/g(\mu)
\end{align*}
where $f,g$ are invertible series in the parameter $\mu$ whose leading coefficient is $1$. In other words, we have a family of theories over the formal scheme which is the formal spectrum of the ring $\C[[c,\hbar,\mu]]/(c \mu = \hbar)$ and this family of theories is unique up to change of coordinates on the formal scheme.  The issue of specifying a particular quantization is then the issue of choosing coordinates on this formal scheme. 

\section{Koszul duality and a two-variable quantum algebra}
\label{sec:Koszul duality}

In this section we will analyze the Koszul dual to the algebra of operators of our $5d$ gauge theory on $\R \times \C^2$.  We will show that this Koszul dual is a quantization of the universal enveloping algebra of the Lie algebra $\op{Diff}(\C) \otimes \mf{gl}_N$ where $\op{Diff}(\C)$ is the algebra  differential operators on the plane. This section uses the techniques developed in \cite{Cos13}.  Roughly speaking, this analysis tells us that the quantum universal enveloping algebra of $\op{Diff}(\C) \otimes \mf{gl}_N$ plays the same role in this theory as the quantum group plays in Chern-Simons theory and the Yangian plays in the theory analyzed in \cite{Cos13}. 

In this section we will use the language of factorization algebras as developed in \cite{CosGwi11}.

Let us recall the strategy as developed in \cite{Cos13}.  Let $D \subset \C^2$ be a small polydisc\footnote{A polydisc in $\C^2$ is the product of two discs in $\C$} centered at $0$.  Then, the space of observables $\Obs(I \times D)$ (where $ I \subset \R$ is an interval) is part of a factorization algebra on $\R \times \C^2$. In particular, by using the factorization product in the $\R$-direction, we find that $\Obs(I \times D)$ can be viewed as part of a factorization algebra on $\R$. 

Our theory is topological in the $\R$-direction. One manifestation of this is the following:
\begin{lemma} 
 For any inclusion $I \subset I'$ of intervals in $\R$, the map
 $$
\Obs(I \times D) \to \Obs(I' \times D)
 $$
 is a quasi-isomorphism. 
\end{lemma}
Factorization algebras with this feature are called locally constant. A locally constant factorization algebra on $\R^n$ gives rise, by a result of Lurie \cite{Lur12}, to an $E_n$ algebra.  Applying this result to our situation, we find that $\Obs(I \times D)$ is an $E_1$ algebra, or homotopy-associative algebra. The product map is easy to describe: if $I_1,I_2$ are disjoint intervals, with $I_1$ situated to the left of $I_2$, such that both are contained in $I_3$, then the structure of a factorization algebra gives a map
$$
\Obs (I_1 \times D) \otimes \Obs(I_2 \times D) \to \Obs(I_3 \times D).
$$
Since the quasi-isomorpism class of $\Obs(I \times D)$ doesn't depend in the interval $I$, this map can be viewed as giving a product (in the homotopical sense) on $\Obs(I \times D)$.  The $E_1$ structure gives us the higher homotopies making this product associative up to coherent homotopy.

It turns out that we need to be a little careful about an extra structure present on $\Obs(I \times D)$. This is not just a cochain complex, but a complete filtered cochain complex: there are subspaces $F^i \Obs(I \times D) \subset \Obs(I \times D)$  such that $\Obs(I \times D)$ is the limit of the quotients $\Obs(I \times D)/F^i \Obs(I\times D)$.  The factorization structure is compatible with this filtration, so that $\Obs(I \times D)$ is a filtered $E_1$ algebra.

Further, since we have a family of theories which depends on three formal parameters $c,\mu,\hbar$ such that $c \mu = \hbar$, we find that $\Obs(I\times D)$ is a filtered $E_1$ algebra in the category of modules over the filtered ring $\C[[\mu,c,\hbar]]/ (\mu c - \hbar)$, where the ring is filtered by giving $c$ and $\mu$ filtered degree $1$ and $\hbar$ filtered degree $2$.  

When we set  $\hbar = \mu = 0$, then $\Obs(I \times D)$ becomes the classical observables of the theory, which we denote by $\Obs^{cl}$. The classical observables are functions are, in the BV formalism, functions on the space of fields, with a differential which is $\d = \{S,-\}$ where $\{-,-\}$ is the BV odd Poisson bracket and $S$ is the classical action functional. 

Recall that we can describe the space of fields of our theory and the classical action functional in terms of an auxiliary sheaf of algebras $\mc{A}$ on $\R \times \C^2$.  On an open of the form $I \times D$ where $D$ is a polydisc, we have
$$
\mc{A}(I \times D) = \Omega^\ast(I) \what{\otimes} \Omega^{0,\ast}(D)= \cinfty(I \times D)[\d t, \d \zbar_1, \d \zbar_2]
$$
where the variables $\d t$, $\d \zbar_i$ are of degree $1$.  The algebra $\mc{A}(I \times D)$ has a differential $\d_{\mc{A}}$ which is a sum of the de Rham differential on $\Omega^\ast(I)$ and the Dolbeault differential on $\Omega^{0,\ast}(D)$. The product on $\mc{A}$ is the obvious wedge product, but we have seen how to equip $\mc{A}[[c]]$ with a non-commutative Moyal product $\ast_c$. 

In terms of the algebra $\mc{A}$, the space of fields of the theory is $\mc{A} \otimes \mf{gl}_N[1]$, and the action functional is the Chern-Simons type functional
$$
S = \tfrac{1}{2} \int \op{Tr}( \alpha \ast_c \d_{\mc{A}} \alpha )+ \tfrac{1}{3} \int \op{Tr}(\alpha \ast_c \alpha \ast_c \alpha)
$$
where $\int$ is a certain integration map on the algebra $\mc{A}$. 

As explained in detail in \cite{CosGwi11}, the complex of classical observables can be represented as the Chevalley-Eilenberg Lie algebra cochain complex of the dg Lie algebra $\mc{A}(I \times D)\otimes \mf{gl}_N[[c]]$, where the Lie bracket is the commutator associated to the Moyal product and we take Chevalley-Eilenberg cochains over the base ring $\C[[c]]$. Since we have a quasi-isomorphism 
$$
\mc{A}(I \times D) \simeq \Oo(D)
$$
where $\Oo(D)$ indicates the algebra of holmorphic functions on the polydisc $D$.   We let $\Oo_c(D)$ denote the algebra $\Oo(D)[[c]]$ with the non-commutative Moyal product. It follows that (if one takes sufficient care with topological vector spaces) one has a quasi-isomorphism
$$
\Obs^{cl}(I \times D) \simeq C^\ast (\Oo_c(D)\otimes \mf{gl}_N). 
$$
between classical observables on $I \times D$ and the Chevalley-Eilenberg cochains, over the ring $\C[[c]]$, of $\Oo_c(D) \otimes \mf{gl}_N$.  

The physics interpretation of this statement is as follows.  Every solution to the equations of motion of the theory on $I \times D$ is gauge-equivalent to the trivial solution. However, the trivial solution has a very large algebra of gauge symmetries, whose Lie algebra is $\mf{gl}_N \otimes \Oo_c(D)$.  Therefore, the algebra of observables of the theory is entirely built from the ghost field and its $z_i$ derivatives.  This gives the expression as a Chevalley cochain complex of $\Oo_c(D) \otimes \mf{gl}_N$.

The Koszul dual of the Lie algebra cochains $C^\ast(\mf{g})$ of any Lie algebra $\mf{g}$ is the universal enveloping algebra $U(\mf{g})$ (although, as explained in \cite{Cos13}, one needs to work with filtered algebras for this to hold). We would therefore expect that the Koszul dual of the algebra $\Obs^{cl}(I \times D)$ is the universal enveloping algebra $U(\Oo_c(D) \otimes \mf{gl}_N)$. However, because $\Oo(D)$ is a topological vector space, there are a number of subtleties with this statement.  To avoid this, we will focus on what should be called observables on the formal disc $\what{D}$. This will bring us into the world of pure algebra, without any functional-analysis complications. 

One way to try to define the observables $\Obs(I \times \what{D})$ on the product of the formal disc with an interval is to take the limit of the spaces $\Obs(I \times D)$ over a nested sequence of polydiscs $D$ centered at the origin. This turns out to be a tricky operation in our context, so we will take another approach.  

Note that $\Obs(I \times D)$ has a $\C^\times$-action arising from the rotation of the polydisc $D$ around the origin. The parameters $c$ and $\hbar$ have weight $2$ under this $\C^\times$ action, and $\mu$ has weight $0$.  We will let 
$$
\Obs^{(k)}(I \times D) \subset \Obs(I \times D)
$$
denote the observables which are of weight $k$ under this $\C^\times$ action.  We then define
$$
\Obs(I \times \what{D}) = \oplus_k \Obs^{(k)} (I \times D)
$$
where the direct sum is taken in the world of complete filtered vector spaces. Explicitly, this means that
$$
\Obs(I \times \what{D}) = \lim_i \left( \oplus_k \Obs^{(k)}(I \times D) / F^i \Obs^{(k)}(I \times D) \right). 
$$
If we do the same procedure at the classical level, we find indeed that
$$
\Obs^{cl}(I \times \what{D}) \simeq C^\ast( \Oo_c(\what{D})  
$$
where $\Oo_c(\what{D}) = \C[[z_i,c]]$ with the non-commutative Moyal product. This justifies the interpretation of this construction as giving the observables on the formal disc. 

Once we do this, we find Koszul duality (as developed in \cite{Cos13}) works perfectly, and that the Koszul dual of $\Obs^{cl}(I \times D)$ is the universal enveloping algebra of $\Oo_c(\what{D})\otimes \mf{gl}_N$. 

We would like to construct the Koszul dual of quantum observables. Before we do this, we need to discuss one extra step.  Koszul duality is an operation performed on \emph{augmented} algebras, that is, algebras equipped with a rank $1$ module. Given an augmented algebra $A$ over a ring $R$, equipped with an augmentation homomorphism $A \to R$, the Koszul dual is $\R \op{Hom}_A(R,R)$.  

To apply this construction to our situation, we need to know that quantum observables are augmented. (Classical observables are easily seen to be augmented: the augmentation is given by evaluating a function of the fields at the field $\alpha = 0$). We can use a homological argument to show that quantum observables have a unique augmentation. 
\begin{lemma}
There is a unique filtered augmentation of the filtered $E_1$ algebra $\Obs(I \times \what{D})$ over the ring $\C[[c,\mu,\hbar]]/(c \mu - \hbar)$.   
\end{lemma}
A filtered augmentation is a map of filtered $E_1$ algebras over $\C[[c,\mu,\hbar]]$
$$
\Obs(I \times \what{D}) \to \C[[c,\mu,\hbar]].
$$
\begin{proof}
There is a unique filtered augmentation when $\hbar = \mu = 0$.  Let $M_0 = \C[[c]]$ denote the augmentation module.  Suppose we have an augmentation modulo the ideal generated by $\mu^n$ (which contains $\hbar^n$).  Let $M_n = \C[[c,\mu,\hbar]]/(\mu c - \hbar, \mu^n)$ denote the augmetation module.  An augmentation exists at the next order in $\hbar$ and $\mu$ if and only if $M_n$ lifts to a module $M_{n+1}$. The obstruction to having such a lifting is
$$
\op{Ext}^2_{\Obs(I \times D^2)/\mu^n} (M_n,M_n)
$$
where we take the $\op{Ext}$-groups over the algebra of observables modulo $\mu^n$, in the world of filtered cochain complexes. If the obstruction vanishes, then the space of lifts to the next order is a torsor for $\op{Ext}^1$.

We can compute these filtered $\op{Ext}$-groups by a spectral sequence associated to the filtration by powers of $\mu$ and $\hbar$. This reduces us to the case of $ M_0=\C[[c]]$, in which case we are computing the $\op{Ext}$-groups of the module $M_0$ over the ring of classical observables.  We have already seen that these $\op{Ext}$-groups are the cohomology of the Koszul dual, and that the Koszul dual of 
$$
\Obs^{cl}(I \times \what{D}) \simeq C^\ast (\Oo_c(\what{D})\otimes \mf{gl}_N)
$$
is the universal enveloping algebra $U(\Oo_c(\what{D}) \otimes \mf{gl}_N)$, which is concentrated in degree $0$. Thus, $\op{Ext}^1$ and $\op{Ext}^2$ vanish and there is a unique augmentation. 
\end{proof}
\begin{remark}
The same argument allows us to construct an augmentation for the algebra when we take our surface  to be a resolution of an ADE singularity instead of $\C^2$. In this case, we should work with the algebra of observables on a formal neighbourhood of the exceptional locus. 
\end{remark}
This lemma shows that we now have enough to define the Koszul dual of the algebra of quantum observables. 
\begin{definition} 
 Let $U_{\hbar}^{QFT}(\Oo_c(\what{D}) \otimes \mf{gl}_N)$ be the cohomology of the Koszul dual of the algebra of quantum observables $\Obs(I \times \what{D})$. This is a flat family of algebras over the ring $\C[[c,\mu,\hbar]]/(\mu c -\hbar)$. 
\end{definition}

Note that this algebra has a $\C^\times$-action arising from the $\C^\times$ action on the two-dimensional formal disc, which also scales the parameters $c$ and $\hbar$ each with weight $2$.  We can use this $\C^\times$-action to define a restricted version of this algebra which is a quantization of $U(\Oo_c(\C^2) \otimes \mf{gl}_N)$, where $\Oo_c(\C^2)$ is the algebra of polynomials on $\C^2$ with the Moyal product. Of course, $\Oo_c(\C^2) = \op{Diff}_c(\C)$ is the algebra of differential operators on the plane with $c$ playing the role of Planck's constant. 

The restricted algebra is defined as follows.  Working modulo $\mu^n$, let $U^{(k)}_{\hbar}(\Oo_c(\what{D}) \otimes \mf{gl}_N)/ \mu^n$ denote the $k$th eigenspace of the $\C^\times$-action.  Modulo $\mu^n$, the restricted algebra is the direct sum of these spaces:
$$
U^{QFT}_{\hbar}(\op{Diff}_c(\C) \otimes \mf{gl}_N) / \mu^n = \oplus_k U^{QFT,(k)}_{\hbar} (\Oo_c(\C)\otimes \mf{gl}_N) / \mu^n
$$
We then set $U_{\hbar}^{QFT}(\op{Diff}_c(\C)\otimes \mf{gl}_N)$ to be the inverse limit of these truncated algebras. The truncated algebras (modulo $\mu^n$) are algebras over $\C[c,\mu,\hbar]/(c \mu - \hbar, \mu^n)$. This is a consequence of the fact that $c$ is of weight $2$ under the $\C^\times$ action so only polynomials in $c$ can appear. When we take the inverse limit, we find a family of algebras over $\C[c][[\mu,\hbar]]/(\mu c - \hbar)$. Thus, the variable $c$ is treated as a polynomial but the other variables are formal. 

Now, the $\C^\times$ action scales the variables $c$ and $\hbar$.  This means that we can use the $\C^\times$ action to set $c = 1$ (note that we can not have $c = 0$ unless $\hbar = 0$ as well). If we set $c = 1$, then $\mu = \hbar$, so we have an algebra $U^{QFT}_{\hbar}(\op{Diff}(\C)\otimes \mf{gl}_N)$ over $\C[[\hbar]]$ deforming the algebra $U(\op{Diff}(\C) \otimes \mf{gl}_N)$.

\begin{theorem}
The family of algebras $U^{QFT}_{\hbar}(\op{Diff}(\C)\otimes \mf{gl}_N)$ constructed in the previous section from our $5$-dimensional gauge theory is a non-trivial deformation of $U(\op{Diff}(\C) \otimes \mf{gl}_N)$. 
\end{theorem}
This is proved by an explicit calculation of the OPE in our $5$-dimensional gauge theory, together with some more abstract results concerning the Hochschild cohomology of the classical algebra.  The relevant OPE is calculated in the next section.

\section{Calculation of an operator product}
\label{section_ope}
 In this section we will perform a one loop calculation of the operator product expansion of certain local operators in the five- dimensional gauge theory. The operators we consider are of ghost number one, and are linear functions of the ghost field. Upon performing Koszul duality, this calculation will give the form of the first order deformation of the Koszul dual algebra.

From a physics perspective it may seem unusual to calculate the OPE of operators of non-zero ghost number. There is, however, another interpretation of this calculation. We will show that the OPE calculation we perform describes the obstruction to certain classical line operators existing at the quantum level.

\subsection {}
Let us start by discussing how calculations of OPEs of local operators relate to the algebra structure of the Koszul dual algebra.  To do this, we first have to recall some details about how Koszul duality responds to deformations of algebras. We will analyze this in the simplest example of Koszul duality, that between the symmetric algebra on a vector space $V$ and the exterior algebra on the dual $V^\ast$.   

Given any dg algebra $A$, let $\op{Def}^\ast(A)$ denote the cohomology of the cochain complex controlling first order deformations of $A$.  Thus, $\op{Def}^0(A)$ is the space of first-order deformations of $A$, modulo equivalence; $\op{Def}^1(A)$ contains obstructions to a first-order deformation extending to a second order; and $\op{Def}^{-1}(A)$ is the space of derivations of $A$. Then, there is an exact sequence 
$$
\dots \to \op{Def}^0(A) \to HH^2(A) \to H^2(A) \to \dots
$$ 
where $HH^2(A)$ is the second Hochschild cohomology group, and $H^2(A)$ is the second cohomology group of $A$. There is a natural cochain map from the Hochschild cochains of $A$ to $A$ itself. In the case that $A$ is situated in degree $0$, this exact sequence tells us that $\op{Def}^0(A)$ is isomorphic to $HH^2(A)$, which is the familiar statement that $HH^2(A)$ describes first-order deformations of an ordinary (not dg) algebra.

In the particular case that $A = \Sym^\ast V$, the Hochschild-Kostant-Rosenberg theorem tells us that the Hochschild cohomology of $A$ is $\Sym^\ast V \otimes \wedge^\ast V^\ast$.  In this case, the map from the Hochschild cohomlogy of $A$ to $A$ is just the natural projection onto $\Sym^\ast V \otimes \wedge^0 V^\ast$.   This is of course surjective, so we find that deformations of $A$ are controlled by the second cohomology of the complex  $\Sym^\ast V \otimes \wedge^{> 0} V^\ast$ (with zero differential). Here we put $V^\ast$ in degree $1$ so that $\wedge^i V^\ast$ is in degree $i$.  The first cohomology of this complex describes infinitesimal symmetries of $\Sym^\ast V$, the third cohomology contains obstructions to a first-order deformation extending to second order, etc. 

To sum up, we find that the space of first-order deformations of $\Sym^\ast V$ is given by $\wedge^2 V^\ast \otimes \Sym^\ast V$. Elements of this space give rise to Poisson brackets on $\Sym^\ast V$, so it is not surprising that they correspond to first-order deformations.

Koszul duality is an operation that one applies to \emph{augmented} algebras.    Not all first-order deformations will preserve the augmentation: only those preserving the maximal ideal of the algebra will.  Applying this to the example of $\Sym^\ast V$, we find that the complex describing deformations of $\Sym^\ast V$ as an augmented algebra is $\Sym^{>0} V \otimes \wedge^{>0} V^\ast$.  For instance, a constant-coefficient Poisson bracket on $\Sym^\ast V$,  which comes from an element of $\wedge^{2} V^\ast$, does not give rise to a deformation preserving the augmentation. This is because it does not preserve the maximal ideal $\Sym^{>0} V \subset \Sym^\ast V$.

Let's now discuss the complex controlling deformations of the Koszul dual algebra $\wedge^\ast V^\ast$. This algebra should be treated as a dg algebra, with zero differential and grading where $\wedge^i V^\ast$ is in degree $i$.  The HKR theorem applies in this situation too, and tells us that the Hochschild cohomology of this algebra is $\wedge^\ast V^\ast \otimes \what{\Sym}^\ast V$.  Here $\what{\Sym}^\ast$ indicates the completed symmetric algebra, defined as the product of the spaces $\Sym^i V$.  If one only looks at deformations which shift the natural grading on $\wedge^\ast V^\ast$ by a finite amount, one finds instead that the relevant version of Hochschild cohomology is $\wedge^\ast V^\ast \otimes \Sym^\ast V$.  (We will gloss over such issues with completions in what follows).  

Deformations are not described by the Hocshchild cohomology, but by the kernel of the natural map from Hochschild cohomology to $\wedge^\ast V^\ast$. Thus, we find that deformations of the exterior algebra $\wedge^\ast V^\ast$ are given by $\wedge^\ast V^\ast \otimes \Sym^{> 0} V$.  Finally, we also want to restrict attention to those deformations which preserve the augmentation.  This restricts is to $\wedge^{> 0} V^\ast \otimes\Sym^{>0} V$.

Thus, on both sides of Koszul duality, we find that the complex describing augmented deformations is the same: it is $\wedge^{> 0} V^\ast \otimes \Sym^{> 0} V$.  The operation of Koszul duality sets up an isomorphism between the complexes of augmented deformations of each algebra.  We will use the following useful (and easy) fact: \emph{the isomorphism}
$$
\op{Def}_+(\wedge^\ast V^\ast) =   \wedge^{> 0} V^\ast \otimes \Sym^{> 0} V \iso \wedge^{> 0} V^\ast \otimes \Sym^{> 0} V  = \op{Def}_+(\Sym^\ast V) 
$$
\emph{between the complexes of augmented deformations of $\wedge^\ast V^\ast$ and $\Sym^\ast V$ is the identity map}.

This allows us to track through deformations of the algebras on both sides. Let's do some examples.

As a first example, suppose we perform a first-order deformation of the differential on $\wedge^\ast V^\ast$.  Such  a deformation is given by a derivation, which is determined by what it does on the space $V^\ast$ of generators. It is therefore given by a linear map $V^\ast \to \wedge^2 V\ast$, that is, an element of $V \otimes \wedge^2 V^\ast \subset \Sym^{>0} V \otimes \wedge^{>0} V^\ast$.  

On the Koszul dual side, an element of $\wedge^2 V^\ast \otimes V$ gives a linear-coefficient Poisson bracket on $\Sym^\ast V$.  So we see that adding a differential to $\wedge^\ast V^\ast$ has the effect of making $\Sym^\ast V$ non-commutative with a linear Poisson bracket. (Since, so far, this is a first-order deformation, there is no requirement that this linear-coefficient bracket satisfies the Jacobi identity). 

Suppose that we extend the first-order deformation of the differential on $\wedge^\ast V^\ast$ to all orders. The resulting differential squares to zero if and only if the corresponding element of $V \otimes \wedge^2 V^\ast$, viewed as a linear map $\wedge^2 V \to V$, satisfies the Jacobi identity.  Thus, we find that non-infinitesimal deformations of the differential on $\wedge^\ast V^\ast$ all arise by treating $\wedge^\ast V^\ast$ as the Chevalley-Eilenberg chain complex for some chosen Lie bracket on $V$.

On the Koszul dual side, a linear-coefficient bracket on $\Sym^\ast V$ gives an all-order deformation of the algebra if and only if it satisfies the Jacobi identity, and so defines a Lie bracket on $V$. In this case, the deformation of $\Sym^\ast  V$ is into the universal enveloping algebra of this Lie algebra.

We have thus sketched, in this example, the standard fact that the Chevalley-Eilenberg cochain complex of a Lie algebra is Koszul dual to the universal enveloping algebra of the same Lie algebra.

Let's consider another deformation, this time by an element of $\Sym^2 V \otimes \wedge^2 V^\ast$. From the point of view of the exterior algebra $\wedge^\ast V^\ast$, this deformation makes the algebra non-commutative.  Indeed, we can think of $\Sym^2 V \otimes \wedge^\ast V^\ast$ as the space of bi-vectors on the odd super-manifold $\pi V$.   We are therefore deforming the product on $\wedge^\ast V^\ast$ by using a quadratic-coefficient Poisson tensor. 

From the point of view of $\Sym^\ast V$, an element of $\Sym^2 V \otimes \wedge^2 V^\ast$ is again a quadratic-coefficient Poisson tensor, and gives rise to a non-commutative deformation of this algebra.

To see how this works concretely, let's choose a basis of $V$ and $V^\ast$ so that we can write $\Sym^\ast V$ as a polynomial algebra $\C[x_i]$ in commuting variables $x_i$, and $\wedge^\ast V^\ast$ as a polynomial algebra $\C[\eps^i]$ in anti-commuting variables $\eps^i$.  Let us denote the tensor in $\Sym^2 V \otimes \wedge^2 V^\ast$ by $A$ with entries $A^{ij}_{kl}$, where it is symmetric in the raised indices and anti-symmetric in the lowered indices.  Then the first-order deformation of $\C[x_i]$ coming from $A$ is given by setting
$$
x_i \ast x_j = x_i x_j + \hbar \tfrac{1}{2} A_{ij}^{kl} x_k x_l + O(\hbar^2).  
$$ 
Dually, the first-order deformation of $\C[\eps^i]$ is given by setting
$$
\eps^i \ast \eps^j = \eps^i \eps^j + \hbar \tfrac{1}{2}A^{ij}_{kl}\eps^k \eps^l.  
$$
This, the upper and lower indices in the tensor $A$ play dual roles in the deformations of $\C[\eps^i]$ and of $\C[x_i]$.   

We can also view these deformations as deforming the relations in the algebra, by saying that
\begin{align*} 
[x_i, x_j] &= \hbar A_{ij}^{kl} x_k x_l \\
[\eps^i, \eps^j] &= \hbar A^{ij}_{kl} \eps^k \eps^l  
\end{align*}
(where $[-,-]$ refers to the graded commutator).  Note that the upper and lower indices in the tensor $A^{ij}_{kl}$ play opposite roles in this expression. 

\subsection{}
Next, let's discuss the deformations coming from a tensor in $\Sym^3 V \otimes \wedge^2 V^\ast$.   From the point of view of the algebra $\Sym^\ast V$, these elements give rise to a deformation whereby the commutator of the generators is a cubic polynomial in the generators.    In a basis, we have
$$
[x_i,x_j] = \hbar A^{klm}_{ij} x_k x_l x_m. 
$$
From the point of view of the exterior algebra $\wedge^\ast V^\ast$, these tensors give rise to a non-trivial $A_\infty$ structure in which the operation $m_3$ is turned on.  We have$$
m_3(\eps^i, \eps^j, \eps^k) = \hbar \tfrac{1}{3!} A^{ijk}_{lm} \eps^l \eps^m.  
$$
In a similar way, deformations of $\C[x_i]$ in which the generators commute to give a polynomial of degree $d$ correspond to $A_\infty$ deformations of $\C[\eps^i]$ in which the $A_\infty$ operation $m_d$ is turned on.  In what follows, however, we will not need to consider these $A_\infty$ deformations.  

\subsection{}
After these generalities, let us return to our quantum field theory.  We have seen that there is a commutative dg algebra of classical observables on $I \times \what{D}$, where $\what{D}$ is a formal disc in $\C^2$.  We have also seen that $\Obs^{cl}(I \times \what{D})$ is canonically quasi-isomorphic to the Chevalley-Eilenberg cochains $C^\ast(\mf{gl}_N \otimes \Oo_c(\what{D}) ) $ where $\Oo_c(\what{D})$ is the algebra $\C[[z_1,z_2]]$ where the variables $z_i$ commute according to $[z_1,z_2] = c$. 

We are interested in the Koszul dual of this algebra and of its first-order quantization.  From our discussion above, we know that the Koszul dual of $C^\ast(\mf{gl}_N[[z_1,z_2]])$ is the universal enveloping algebra $U(\mf{gl}_N[[z_1,z_2]])$ (where as always the $z_i$ commute to give $c$).    

What we will calculate is the first-order deformation of the product on the algebra $\Obs^{cl}(I \times \what{D}) \simeq C^\ast(\mf{gl}_N[[z_1,z_2]])$.   We will in fact analyze the failure of this product to be (graded) commutative. That is, we will compute the commutator of the generators of this algebra, which are elements in the linear dual of $\mf{gl}_N[[z_1,z_2]]$.  For grading reasons, the commutator of two generators is given by a quadratic expression in the generators. 

Let us write $V = \mf{gl}_N[[z_1,z_2]]$, so that as an algebra we can identify $C^\ast(\mf{gl}_N[[z_1,z_2]])$ with $\wedge^\ast V^\ast$. Then the Hochschild cochain which gives rise to the deformation of the product on this algebra is an element of $\Sym^2 V \otimes \wedge^2 V^\ast$. On the Koszul dual side, it thus gives rise to a deformation of the universal enveloping algebra $U(\mf{gl}_N[[z_1,z_2]])$ in which the commutator of two generators has a term which is quadratic in the generators.

Let's now turn to the computation of the quantum correction to the commutator in the algebra $C^\ast(\mf{gl}_N[[z_1,z_2]])$.  Let us write $\partial_{z_1}^r \partial_{z_2}^s X_{\alpha \beta}$ be the generator of $C^\ast(\mf{gl}_N[[z_1,z_2]])$ which, as a linear function on $\mf{gl}_N[[z_1,z_2]]$, sends an element $M f(z_1,z_2)$ to
$$
\partial_{z_1}^r \partial_{z_2}^s X_{\alpha,\beta}(M f(z_1,z_2)) = \op{Tr} (E_{\alpha\beta} M) \left( \partial_{z_1}^{r} \partial_{z_2}^s f\right)_{z_1 = z_2 = 0}.  
$$  
Here $M$ is an element of $\mf{gl}_N$, $f \in \C[[z_1,z_2]]$, and $E_{\alpha \beta}$ is the elementary matrix. (We can identify the space of functions on the non-commutative formal disc as functions as the vector space  $\C[[z_1,z_2]]$ equipped with the Moyal product).  

As a function of the fields of the $5$-dimensional gauge theory, the operator $\partial_{z_1}^r\partial_{z_2}^sX_{\alpha,\beta}$ is a function just of the ghosts field $\chi$. As such, it is given by the same expression:
$$
\partial_{z_1}^r\partial_{z_2}^sX_{\alpha,\beta}(\chi) = \op{Tr}  E_{\alpha \beta}  \left( \partial_{z_1}^{r} \partial_{z_2}^s \chi \right)_{z_1 = z_2 = t = 0}. 
$$
In this expression, $z_1,z_2,t$ are the coordinates on the $5$-dimensional space-time $\C^2 \times \R$.

Let us denote by $\{-,-\}$ the Poisson bracket on $C^\ast(\mf{gl}_N[[z_1,z_2]])$ that arises from its deformation to first order in $\hbar$ into the algebra of quantum observables. This Poisson bracket is computed as follows.  We consider two observables of the form $\partial_{z_1}^r \partial_{z_2}^s X_{\alpha,\beta}$, and place one at $(t = 0, z_i = 0)$ and the other at $(t = \eps, z_i = 0)$. Then we consider the one-loop contribution to the operator product
$$\left(\partial_{z_1}^r \partial_{z_2}^s X_{\alpha \beta} (0)\right) \left( \partial_{z_1}^m \partial_{z_2}^n  X_{\gamma \delta}(\eps)\right).$$
We then take the limit as $\eps \to 0$.  The Poisson bracket is obtained by computing the difference between the limit as $\eps \to 0$ from above and below.  

The OPE is computed by taking the sum over all diagrams such as those in figure \ref{figure_ope}.  

Before turning to the computation, let us explain (following our discussion earlier) how the result of the Poisson bracket will translate into an expression for the deformation of the universal enveloping algebra $U(\mf{gl}_N \otimes \C[[z_1,z_2]])$.   To avoid a profusion of indices, let us consider the case $N = 1$, in which case we can write operator $X_{\alpha \beta}$ as just $X$.  

Suppose we find an expression at one loop like
\begin{equation}
\{\partial_{z_1}^p \partial_{z_2}^q X, \partial_{z_1}^k\partial_{z_2}^l X\} = \sum \hbar c^{(r+s + m + n-p-q-k-l)/2 - 1}  A^{p,q,k,l}_{r,s,m,n}  (\partial_{z_1}^r \partial_{z_2}^s X ) (\partial_{z_1}^m \partial_{z_2}^n X).
\end{equation}
Then the commutator on the universal enveloping algebra $U(\mf{gl}_1  \otimes \C[[z_1,z_2]])$ will be corrected by an expression of the form
\begin{equation}
[z_1^r z_2^s, z_1^m z_2^n] = \sum \hbar  c^{(r+s+m+n-p-q-k-l)/2  -1} A^{p,q,k,l}_{r,s,m,n} 
(z_1^p  z_2^q) (z_1^k z_2^l)\frac{m! n! r! s!  }{p! q! k! l!}  
\end{equation}

Let us now turn to the explicit computation for the commutator. 
\begin{proposition}
We have 
\begin{align*} 
 2^5 \pi^2  \{X_{\alpha,\beta}, X_{\gamma,\delta} \} = \hbar \partial_{z_2}X_{\gamma,\beta} \partial_{z_1}X_{\alpha,\delta}  
- \hbar \partial_{z_1}X_{\gamma,\beta} \partial_{z_2}X_{\alpha,\delta}
\\
+ \hbar \delta_{\alpha \delta} \sum_{\mu} \left(\partial_{z_1}X_{\mu,\beta}\partial_{z_2}X_{\gamma,\mu} -  \partial_{z_2}X_{\mu,\beta} \partial_{z_1}X_{\gamma,\mu}\right)\\
+ \hbar \delta_{\gamma \beta} \sum_\mu \left(\partial_{z_1}X_{\alpha \mu} \partial_{z_2}X_{\mu \delta} -  \partial_{z_2}X_{\alpha \mu} \partial_{z_1}X_{\mu \delta} \right) \\
+  O(c) 
\end{align*}
Here $O(c)$ indicates terms which depend on $c$ and so must involve higher derivatives of the operators $X_{\dot \dot}$. 
\label{proposition_commutator_QFT}
\end{proposition} 
\begin{remark}
It is very possible that I dropped a sign in the determination of the constant $2^5 \pi^2$.  
\end{remark}
\begin{figure}
\begin{tikzpicture}[scale=0.6]
\node[draw,circle](N1) at (-6,0) {$X_{\alpha\beta}$};
\node[draw,circle](N2) at (6,0) {$X_{\gamma\delta}$};
\node[draw, shape=circle](N3) at (-2,0) {$I$}; 
\node[draw, shape=circle](N4) at (2,0) {$I$};
\draw[dbl-<-] (N1) -- (N3) node[midway, above]{$\alpha$} node[midway,below]{$\beta$} ;
\draw[double, double equal sign distance] (N3) -- (N4) node[midway, above]{$\alpha$} node[midway,below]{$\gamma$} ;
\draw[dbl->-] (N4) -- (N2)  node[midway, above]{$\delta$} node[midway,below]{$\gamma$}; 
\draw[dbl-<-]   (N3) -- (-2,-2);
\draw[dbl-<-]  (N4) -- (2,2);  
\end{tikzpicture}
\caption{This is one of the diagrams which contribute to the OPE calculation of proposition \ref{proposition_commutator_QFT}.  We use double line notation.  The arrows indicate whether we place an operator or an external ghost field at the end of the line. The vertices labelled by $I$ come from the interaction term in the Lagrangian for our theory.\label{figure_ope} }   
\end{figure}
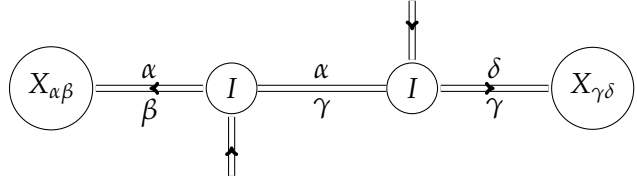
\begin{proof}
For $t \in \R$ we let $X_{\alpha \beta}(t)$ denote the same operator as described above, but evaluated at $(t,0,0) \in \R \times \C^2$.  To compute the product in the algebra of operators between $X_{\alpha \beta}$ and $X_{\gamma \delta}$, we consider placing $X_{\alpha \beta}$ at $t = 0$ and $X_{\gamma \delta}$ at $t = \eps$. That is, we consider the operator $X_{\alpha \beta} (0) X_{\gamma \delta}(\eps)$.  We then take the limit as $\eps \to 0$.  The commutator is obtained by computing the difference between the limit as $\eps \to 0$ from above and below. 

To compute this, we consider the sum over all connected graphs with two special vertices, labelled by $X_{\alpha \beta}(0)$ and $X_{\gamma \delta}(\eps)$ respectively; a number of ordinary trivalent vertices, labelled by the interaction $I$; and some number of external lines, which we label by a ghost field $\chi$.  Once we take the discontinuity at $0$ of this as a function of $\eps$, we will find some function of the ghost field $\chi$ which is a polynomial in $\partial_{z_1}^r \partial_{z_2}^s \chi_{\mu \eta}(t = 0, z_1 = 0, z_2 = 0)$. Thus, the resulting operator will again be expressed in terms of the operators $\partial_{z_1}^r \partial_{z_2}^s X_{\mu \eta}$.

Let us consider what diagrams can possibly appear at one loop. First, observe that for reasons of cohomological degree, there must be precisely two external lines.  Next, note that each interaction contributes a factor of $\hbar^{-1}$ and each internal line contributes a factor of $\hbar$. To be a one-loop diagram, the only possibility is that the graph is a tree with two internal vertices labelled by $I$, two external lines labelled by the ghost field $\chi$, and two special vertices labelled by $X_{\alpha \beta}$ and $X_{\gamma \delta}$.  The only such tree is depicted in figure \ref{figure_ope}. There, we have drawn the tree using the double-line notation (i.e. it's drawn as a planar tree) to assist in keeping track of the $\mf{gl}_N$ indices.  There is only one such tree, but it has four different structures of planar tree, depending on whether the two external lines point up or down in figure \ref{figure_ope}.   Each such tree will contribute the same analytic factor to the OPE, but the $\mf{gl}_N$ factors will be slightly different, accounting for the various terms in the sum that appears in the statement of proposition \ref{proposition_commutator_QFT}. 

In what follows we will focus on computing the analytic factor in the OPE.  In the sum over diagrams, we label each edge of the graph by a propagator for the theory. Let us write down the propagator. First, we need to introduce some notation. 

Let
\begin{align*} 
D &= \d \zbar_1 \partial_{\zbar_1} + \d \zbar_2 \partial_{\zbar_2} + \d t \partial_t \\
D^\ast &= -\dpa{\d \zbar_1} \partial_{z_1} - \dpa{\d \zbar_2} \partial_{z_2} - \dpa{\d t} \partial_t. 
\end{align*}
Here $\dpa{\d \zbar_1}$ means contraction with the vector field $\partial_{\zbar_1}$, i.e.\ it is the operator that removes a $\d \zbar_1$.  These operators act on $\Omega^{0,\ast}(\C^2) \what{\otimes} \Omega^\ast(\R)$ and so on the  fields of the theory. The operator $D$ is the linearized BRST operator of the theory.  

Note that
$$
[D,D^\ast] =-\partial_{z_1} \partial_{\zbar_1} - \partial_{z_2} \partial_{\zbar_2} - \partial_t^2 = \Lap 
$$
is the Laplacian on $\C^2 \times \R$. 

The propagator of the theory is the kernel for the operator $D^\ast \Lap^{-1}$. One can calculate that this is
\begin{align*} 
 -\frac{5}{48 \pi^2} D^\ast \left( \frac{\d \zbar_1 \d \zbar_2 \d t}{r^3}  \right)\\
&= \frac{15}{96 \pi^2}\frac{\zbar_1 \d \zbar_2 \d t - \d \zbar_1 \zbar_2 \d t + \d \zbar_1 \d \zbar_2 t}{r^5}   
\end{align*}
(although the constant appearing here will play no role for us).

In the OPE we are calculating, adding a quantity supported away from the diagonal to the propagator will have no effect, because it will only change the result by something that is continuous as a function of $\eps$ and we are interested in the discontinuity at $\eps = 0$.  We can therefore replace the propagator by a propagator with an infra-red cutoff. The details of the IR cutoff are irrelevant: for example, we could obtain an IR cutoff by replacing the function $r^{-3}$ by a function which is $r^{-3}$ for $r$ small and $0$ for $r$ large. 

We let $\tr^{-1}_{IR}$ denote the parametrix for the Laplacian whose kernel is the Green's function with such an infrared cutoff. 

We can write the analytic factor in the diagram we are computing as a composition of linear operators acting on the space $\Omega^\ast(\R) \what{\otimes} \Omega^{0,\ast}(\C^2)$ of fields, as follows. We represent the operator $X_{\gamma \delta}(\eps)$ that we place at $t = \eps$, $z_1 = z_2 = 0$, by the delta-current $\delta_{t = \eps,z_i = 0}$ which we view as a distributional element of $\Omega^1(\R) \what{\otimes} \Omega^{0,2}(\C^2)$. (We are implicitly dividing the $\delta$-function by $\d z_1 \d z_2$).  

We let $\chi_1,\chi_2 \in \Omega^{0}(\R \times \C^2)$ denote the two external ghost fields we insert. 

Recall that we are computing the OPE at $c = 0$, in which limit the cubic term in the Lagrangian contains no derivatives.

We find that the analytic factor in the diagram of \ref{figure_ope} is
\begin{equation}
\left[D^\ast \tr_{IR}^{-1} \chi_1 D^\ast \tr_{IR}^{-1} \chi_2 D^\ast \tr_{IR}^{-1}  (\delta_{t = \eps, z_i = 0} ) \right](0,0,0). \label{equation_analytic_factor} 
\end{equation}
Note that since $D^\ast \tr^{-1}$ decreases cohomological degree by $1$, and we start with something in $\Omega^1(\R) \what{\otimes} \Omega^{0,2}(\C^2)$, the expression inside the square brackets is  a distribution on $\R \times \C^2$ whose only singularities are at $t = \eps, z_i = 0$.  We then evaluate this distribution at $(0,0,0)$.  

Now, $(D^\ast)^2 = 0$, and $[D^\ast,\tr_{IR}^{-1}] = 0$.  The second equation follows from the fact that we can choose $\tr_{IR}^{-1}$ to commute with translation.  

We then commute $D^\ast$ past the ghost fields $\chi_2$ in equation \ref{equation_analytic_factor}.  Using the fact that $(D^\ast)^2= 0$ and $[D^\ast,\tr_{IR}^{-1}]= 0$ we can rewrite equation \ref{equation_analytic_factor} as
\begin{equation}
\left[D^\ast \tr_{IR}^{-1} [\chi_1, D^\ast] \tr_{IR}^{-1} [\chi_2, D^\ast] \tr_{IR}^{-1}  (\delta_{t = \eps, z_i = 0} ) \right](0,0,0). \label{equation_analytic_factor2} 
\end{equation}  
Any terms in our OPE which involve $\zbar_i$ or $t$ derivatives of the fields $\chi_i$ are BRST exact and can be discarded.  We can therefore assume, without loss of generality, that the fields $\chi_i$ are holomorphic functions of the  $z_i$.  Using this, we find that
\begin{equation}
[\chi_1,D^\ast] = \sum \dpa{\d \zbar_i} \frac{\partial \chi_1}{\partial z_i}. 
\end{equation}
Arranging the terms in equation \ref{equation_analytic_factor2} gives us 
\begin{equation}
\left[\eps_{ij} \partial_t \frac{\partial}{\partial d t}  \frac{\partial}{\partial \d \zbar_i}  \frac{\partial}{\partial \d \zbar_j} \tr_{IR}^{-1}      \frac{\partial \chi_1}{\partial z_i}  \tr_{IR}^{-1} \frac{\partial \chi_2}{\partial z_j} \tr_{IR}^{-1}  (\delta_{t = \eps, z_i = 0} ) \right](0,0,0). \label{equation_analytic_factor3} 
\end{equation}
Next, we note that if we scale the $z_i$ by a factor $\lambda$ we scale $\hbar$ by $\lambda^2$. It follows that the result of our operator product will have precisely $1$ derivative in each of the $z_i$.  Since the expression we have arrived at contains at least one derivative in each of the $z_i$, we find that no more  derivatives are possible.  Therefore, we can assume that the external ghost fields $\chi_i$ are each a linear combination of $z_1$ and $z_2$. Without loss of generality we will assume that $\chi_2 = z_2$ and $\chi_1 = z_1$.

We can now rewrite our expression as 
 \begin{equation}
\left[\partial_t \frac{\partial}{\partial d t}  \frac{\partial}{\partial \d \zbar_1}  \frac{\partial}{\partial \d \zbar_2} \tr_{IR}^{-3}   (\delta_{t = \eps, z_i = 0} ) \right](0,0,0). \label{equation_analytic_factor4} 
\end{equation}
Here $\tr_{IR}^{-3}$ is the cube of the operator $\tr_{IR}^{-1}$. 

Let us now calculate $\tr_{IR}^{-3}$.  To do this, we will choose a particular IR cutoff based on the heat kernel, and write
$$
\tr_{IR}^{-1} = \int_{0}^L K_u \d u 
$$
where
$$
K_u = (4 \pi u)^{-5/2} e^{-r^2 / 4 u}
$$
is the heat kernel (and $r^2 =  t^2 + \abs{z_1}^2 + \abs{z_2}^2$).  If we performed the integral of the heat kernel from $0$ to $\infty$, we would find the usual Green's function. However, as we are only integrating from $0$ to $L$, we find a regularized Green's function which has the same behaviour for small $r$ but which decays like $e^{-r^2/4 L}$ for large $r$. 

To calculate $\tr_{IR}^{-3}$, we should convolve $\tr_{IR}^{-1}$ with itself three times.  The convolution of $K_u$ with $K_s$ is $K_{u+s}$. We thus find that
\begin{align*} 
 \tr_{IR}^{-3} &= \int_{0 \le u_1,u_2,u_3 \le L}  K_{u_1 + u_2 + u_3} \d u_1 \d u_2 \d u_3\\
 &=  \int_{u = 0}^{3 L} \int_{\substack{s_1, s_2 \le L \\ s_1 + s_2 \le u} } (4 \pi u)^{-5/2} e^{-r^2 / 4 u} \d s_1 \d s_2 \d u
 \end{align*}
In the second line we have renamed the variables $u = \sum u_i$, $s_1 = u_1$, $s_2 = u_2$.  Note that, in the integral in the second line, if we only perform the integral in the range when $L \le u \le 3 L$, we are left with a smooth function of the space-time variables $t,z_1,z_2$.  Such a smooth function can not contribute to  the OPE calculation we are considering. Therefore it suffices to consider the integral over the domain  when $u \le L$. In this case, the integral  over the variables $s_i$ produces $u^2/2$.  We therefore find
$$
 \tr_{IR}^{-3} = \int_{u = 0}^{L}\frac{u^2}{2} (4 \pi u)^{-5/2} e^{-r^2 / 4 u} \d u.
$$
Recall that we are interested in the $t$-derivative of $\tr_{IR}^{-3}$ when $z_1 = z_2 = 0$ and $t = \eps$.  We will thus set $z_i = 0$ and take the $t$-derivative, giving
\begin{align*} 
 \partial_t \tr_{IR}^{-3}\mid_{z_i = 0} &= \int_{u = 0}^{L}\frac{u^2}{2} (4 \pi u)^{-5/2}\partial_t  e^{-t^2 / 4 u}  \d u\\
 & -t \int_{u = 0}^{L}\frac{u}{4} (4 \pi u)^{-5/2}  e^{-t^2 / 4 u}  \d u.
 \end{align*}
The powers of $u$ that appear in the integral are now such that the integral converges as $L \to \infty$.  Since changing $L$ will only change the answer by a smooth function  of $t$, we can take this limit.   We find
\begin{align*} 
  \partial_t \tr_{IR}^{-3}\mid_{z_i = 0} &=  
 -  t \pi^{-5/2} 4^{-7/2}   \int_{u = 0}^{\infty}u^{-3/2} e^{-t^2 / 4 u} \d u.
 \end{align*}
Changing variables $s = t^2 /4 u$ gives us 
\begin{align} 
  \partial_t \tr_{IR}^{-3}\mid_{z_i = 0} &=  
   t \pi^{-5/2} 4^{-7/2} 2 \abs{t}^{-1}    \int_{s = 0}^{\infty}s^{3/2} e^{-s}s^{-2} \d s\\
   &=(t / \abs{t}) \pi^{-5/2} 2^{-6} \Gamma (\tfrac{1}{2})  \\
   &=(t / \abs{t}) 2^{-6} \pi^{-2}.  \label{equation_analytic_factor5}  
 \end{align}
Setting $t = \eps$, we find that the result of the OPE is 
$$
\begin{cases}
2^{-6} \pi^2 & \text{ if } \eps > 0 \\
- 2^{-6} \pi^2 & \text{ if } \eps < 0. 
\end{cases}
$$
We are interested in the commutator, which is the difference of the limit  as $\eps \to 0$ from above and below. This tells us that the analytic factor in the commutator is $2^{-5} \pi^{-2}$.

\end{proof}
Recall that the $A_\infty$ structure on $C^\ast(\mf{gl}_N[[z_1,z_2]])$ contributes to the deformation of the product on the Koszul dual universal enveloping algebra.  At first order in $\hbar$, there is no possibility of a non-trivial $A_\infty$ structure. This is simply because there are no connected diagrams with three or more special vertices at which we place the local operators $X$ which appear at this order in $\hbar$; all such diagrams appear at second and higher order in $\hbar$. 

The deformation of $C^\ast(\mf{gl}_{N} \otimes \C[z_1,z_2])$ can be turned, by Koszul duality, into a deformation of the universal enveloping algebra $U(\mf{gl}_{N} \otimes \C[z_1,z_2] )$ .  The explicit formula presented above for $c = 0$ for certain commutators in $C^\ast(\mf{gl}_{N} \otimes \C[z_1,z_2] )$ can be  turned into an explicit expression for a deformation of $U(\mf{gl}_{N} \otimes \C[z_1,z_2])$.  We find that the one-loop correction to the commutator of $E_{\alpha \beta} z_1 $ with $E_{\gamma \delta} z_2$ is
\begin{multline} 
 [E_{\alpha\beta}z_1, E_{\gamma\delta} z_2] = \delta_{\beta \gamma}E_{\alpha\delta}(z_1 \ast_c z_2) - \delta_{\alpha\delta}E_{\gamma\beta}(z_2 \ast_c z_1)\\
 + \hbar 2^{-5} \pi^{-2}  E_{\alpha \delta} E_{\gamma \beta} - \hbar 2^{-5} \pi^{-2} \sum_{\mu} \delta_{\beta \gamma}  E_{\alpha \mu} E_{\mu \delta}  
- \hbar 2^{-5} \pi^{-2} \sum_{\mu} \delta_{\alpha \delta} E_{\gamma \mu} E_{\mu \beta} .
\label{equation_algebra_commutator_QFT}
\end{multline}
There is one slight subtlety: in translating between the Poisson bracket on the cochain algebra $C^\ast(\g[[z]])$ and the deformation of the product on the Koszul  dual, one has to worry about the ordering in a product like $E_{\gamma \mu} E_{\mu \beta}$.  The most scientific  way to deal with this is to take the average of the product in both orderings.   However it turns out that the difference between choosing one ordering ($E_{\gamma \mu} E_{\mu \beta})$ versus the other ($E_{\mu \beta} E_{\gamma \mu}$) can be absorbed into a  redefinition of the generators of the form
$$
E_{\alpha \beta} z_i \mapsto  E_{\alpha \beta} z_i + \hbar c^{-1} 2^{-5} \pi^{-2} E_{\alpha \beta} z_i. 
$$

Note that, although the expression for the quantum-corrected commutator in the algebra $C^\ast(\mf{gl}_N[z_1,z_2])$ involves some expressions involving $c$ which were not explicitly evaluated, the formula for the quantum correction to the commutator of the elements $E_{\alpha \beta}z_1$ and $E_{\gamma \delta} z_2$  in $U(\mf{gl}_N[z_1,z_2])$ is exact. The point is that the terms in the equation in proposition \ref{proposition_commutator_QFT} which depend on $c$ must involve higher derivatives of the $z_i$, and therefore will only contribute to the commutators of the form $[E_{\alpha \beta} z_1^k z_2^l, E_{\gamma \delta} z_1^r z_2^s]$ for $k+l+r+s >2$.   

In the appendix \ref{appendix_non_trivial_deformation} we will prove the following proposition.
\begin{proposition}
For sufficiently large $R$, this deformation is non-trivial as a first-order deformation of the family of algebras $U(\mf{gl}_{N+R \mid R}\otimes \C[z_1,z_2] )$ for all values of $c$ (where $[z_1,z_2] = c$).  
\end{proposition}
Because we also have a theorem regarding the uniqueness of the family of deformations of $U(\mf{gl}_{N+R \mid R}\otimes \C[z_1,z_2])$, this result will allow us to relate the deformation coming from field theory to an explicit deformation presented in the next section.  

\section{A combinatorial description of the algebra}
\label{section_combinatorial_algebra}
In this section we'll give a combinatorial description of a family of algebras deforming $U(\op{Diff}(\C) \otimes \mf{gl}_N)$.   We  will later relate this combinatorially defined family of  algebras with $U_{\hbar}^{QFT}(\op{Diff}(\C) \otimes\mf{gl}_N)$, using a slightly subtle uniqueness statement about deformations of $U(\op{Diff}(\C) \otimes \mf{gl}_N)$.   

We will refer to the combinatorially defined algebra either as $U_{\hbar}^{comb}(\op{Diff}_c(\C) \otimes \mf{gl}_N)$, or (more often) as $A_{N,\hbar,c}$.
 
The definition of this algebra involves combinatorial games with surfaces.  
\begin{definition}
A \emph{marked surface} is a compact oriented topological surface $\Sigma$ with boundary, together with a number of marked points on the boundary.  The set of marked points is equipped with a total ordering,  so that we have a collection $p_1,\dots,p_n$ of marked points.  The marked points are equipped with the following labels:
\begin{enumerate} 
\item Each marked point is labelled as either incoming or outgoing.  
\item Each marked point is labelled as being of type $1$ or of type $2$. 
\item Type $1$ marked points have an additional label by an element $\alpha$ in the set $\{1,\dots,N\}$.   
\end{enumerate}
These labels are subject to the following constraint.  For every marked point $p_i$ of type $1$ on $\Sigma$ which is incoming, we require that there is a marked point $p_j$ which is on the same boundary component as $p_i$, and directly to the right of $p_i$ using the orientation on the boundary of $\Sigma$, and which is of type $1$ and outoing.  Conversely, every type $1$ outgoing marked point $p_j$ has a type $1$ incoming marked point on the same boundary component and directly to the left.  In this way, type $1$ marked points come in pairs of an incoming and outgoing marked point, and these pairs are adjacent to each other on the boundary of $\Sigma$. 

It is important to note that the total order on the set $p_i$ of marked points is \emph{completely independent} of the way these marked points are arranged on the boundary of $\Sigma$.

The surface $\Sigma$ may be disconnected, but each connected component is required to have non-empty boundary.  

Two such surfaces $\Sigma$, $\Sigma'$ are equivalent if there is an orientation-preserving homeomorphism $f : \Sigma \to \Sigma'$ which takes the marked point $p_i$ on $\Sigma$ to the marked point $p_i'$ on $\Sigma'$, and is such that the various labels on $p_i$ are the same as those on $p_i'$.  
\end{definition}
Figure \ref{fig:surface} is a local picture of what such a marked surface looks like, near the boundary.

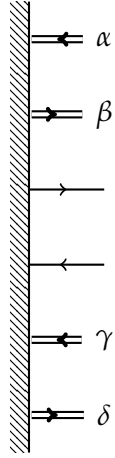
\begin{figure}

\begin{tikzpicture}[thick]
\draw(3,0)--(3,6);
\fill[pattern = north west lines ](2.75,0) rectangle (3,6);
\node(N1) at (4,5.5) {$\alpha$}; 
\node(N2) at (4,4.5) {$\beta$}; 
\draw[dbl-<-](3,5.5)--(N1);
\draw[dbl->-](3,4.5)--(N2);

\draw[->-](3,3.5)--(4,3.5);
\draw[-<-](3,2.5)--(4,2.5);
%\draw[->-](3,1.5)--(4,2.5);
 
\node(N3) at (4,1.5) {$\gamma$};
\node(N4) at (4,0.5) {$\delta$};
\draw[dbl-<-] (3,1.5)--(N3);
\draw[dbl->-](3,0.5) -- (N4); 
\end{tikzpicture}
\caption{This figure displays a typical region in the boundary of a marked surface.  The arrows correspond to the marked points, and the orientation of the arrows indicates whether the marked points are viewed as incoming or outgoing.  The double arrows are type $1$ marked points and the single arrows are type $2$ marked points.  The labels on the double arrows are elements of the set $\{1,\dots,N\}$. Note that the double arrows come in adjacent pairs.   Not shown is the total ordering on the set of marked points (arrows). \label{fig:surface} }  
\end{figure}

We can define a large algebra $\til{A}$ as the vector space spanned by the set of equivalence classes of such marked surfaces.   We will refer to the element of $A$ corresponding to a surface $\Sigma$ with marked points $p_1,\dots,p_n$ as $\Sigma(p_1,\dots,p_n)$.   The algebra structure on $A$ is defined by saying that
$$
\Sigma(p_1,\dots,p_n)\cdot \Sigma'(q_1,\dots,q_m) = \left( \Sigma \amalg \Sigma'\right) (p_1,\dots,p_n,q_1,\dots,q_m).
$$
In other words, the product of basis elements corresponding to marked surfaces is obtained by taking the disjoint union of the surfaces, with the total order on the marked points of the resulting disconnected surface defined by the order $p_1,\dots,p_n,q_1,\dots,q_m$ as indicated.

We will form a quotient $A$ of $\til{A}$ by a combinatorially-defined collection of relations, which depends on two parameters $c,\hbar$ as before.  The relations involve the operation of boundary connect sum of surfaces.  

\begin{definition} 
 Let $\Sigma(p_1,\dots,p_n)$ be a marked surface, and consider the marked points $p_{i},p_{i+1}$.  Let $I_i$ (respectively $I_{i+1}$) be a closed interval in the boundary of $\Sigma$ which contains $p_i$ (respectively $p_{i+1}$) in its interior, but contains no other marked points. Fix an orientation-reversing homeomorphism $f : I_{i} \simeq I_{i+1}$.  We can form a new surface by gluing $I_i$ to $I_{i+1}$ using $f$.  We will denote this new surface  by
 $$
\Sigma(p_1,\dots,p_i \asymp  p_{i+1},\dots,p_n). 
 $$
 This new surface has two fewer marked points, but the remaining marked points are equipped with the natural total ordering and labels they had before. 
\end{definition}

Let us now describe the relations. There are three different classes of relations, which we call \emph{permuting, cutting and moving}. We will start by discussing the permuting relations. These tell us what happens when we change the order on the set of marked points on the boundary of $\Sigma$ without changing their position on $\Sigma$. 

 Let $\Sigma(p_1,\dots,p_n)$ be a marked surface. 
\begin{enumerate} 
 \item[(P0)] Suppose that $p_i$ and $p_{i+1}$ are either
 \begin{enumerate} 
  \item Of different types, so that $p_i$ is of type $1$ and $p_{i+1}$ of type $2$, or the other way around.
  \item Both incoming or both outgoing.
  \item Both of type $1$, but labelled by different elements of the set $\{1,\dots,N\}$.  
 \end{enumerate}
 Then,
 $$
\Sigma(p_1,\dots,p_i,p_{i+1},\dots,p_n) = \Sigma(p_1,\dots,p_{i+1},p_{i},\dots,p_n). 
 $$
 \item[(P1)] Suppose the marked points $p_i$ and $p_{i+1}$ are both of type $2$, where $p_i$ is incoming and $p_{i+1}$ is outgoing. Then,
 $$
\Sigma(p_1,\dots,p_i,p_{i+1},\dots,p_n) = \Sigma(p_1,\dots,p_{i+1},p_{i},\dots,p_n) + \hbar \Sigma(p_1,\dots,p_i\asymp p_{i+1},\dots,p_n) .  
 $$
 In other words, when we switch the order of $p_i$ and $p_{i+1}$ we get an extra contribution from the surface where $p_i$ and $p_{i+1}$ are glued.  This is depicted in the following diagram:

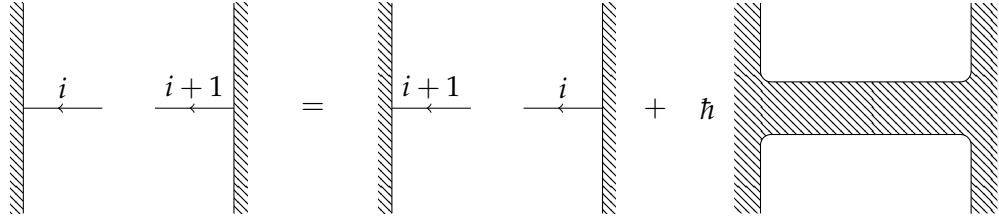
\begin{tikzpicture}[scale=0.7, shift = {(-5,0)}]

\begin{scope}[shift={(-2,0)}]
\fill[pattern=north west lines ] (2.75,0) rectangle (7.25,4);
\fill[color=white] (3,-0.1) rectangle (7,4.1); 
\draw (3,4)--(3,0);
\draw (7,4)--(7,0);
\draw[-<-](3,2) -- (4.5,2) node[midway,above]{$i$};
\draw[-<-](5.5,2) -- (7,2) node[midway,above]{$i+1$};
\end{scope}
\node at (6.5,2) {$=$}; 
\begin{scope}[shift={(5,0)}]
\fill[pattern=north west lines ] (2.75,0) rectangle (7.25,4);
\fill[color=white] (3,-0.1) rectangle (7,4.1); 
\draw (3,4)--(3,0);
\draw (7,4)--(7,0);
\draw[-<-](3,2) -- (4.5,2) node[midway,above]{$i+1$};
\draw[-<-](5.5,2) -- (7,2) node[midway,above]{$i$};
\end{scope}
\node at (13,2) {$+$};  

\begin{scope}[shift = {(12,0)}]
\node at (2.0,2) {$\hbar$};
\draw[rounded corners] (3,4) -- (3,2.5) -- (7,2.5) -- (7,4);  
 \draw[rounded corners] (3,0) -- (3,1.5) -- (7,1.5) --  (7,0);
\fill[pattern=north west lines ]  (2.5,4) -- (3,4)[rounded corners] -- (3,2.5) -- (7,2.5) [sharp corners]--(7,4)  -- (7.5,4) -- (7.5,0) -- (7,0) [rounded corners] -- (7,1.5) -- (3,1.5) [sharp corners]-- (3,0)  -- (2.5,0) -- cycle; 

\end{scope}

\end{tikzpicture}
Here, the arrows indicate incoming or outgoing type $2$ marked points, and the labels indicate the position of these marked points in the total order. 
 \item[(P2)] Suppose the marked points $p_i$ and $p_{i+1}$ are both of type $1$, where $p_i$ is incoming and $p_{i+1}$ is outgoing.  Suppose further that $p_i$ and $p_{i+1}$ are both labelled by the same element in the set $\{1,\dots,N\}$. Finally, suppose that $p_{i+1}$ is \emph{not} the marked point which is directly to the right of $p_i$ on the same boundary component.  

Then,
 $$
\Sigma(p_1,\dots,p_i,p_{i+1},\dots,p_n) = \Sigma(p_1,\dots,p_{i+1},p_{i},\dots,p_n) +  \Sigma(p_1,\dots,p_i\asymp p_{i+1},\dots,p_n) .  
 $$
This is depicted in the following diagram:
 
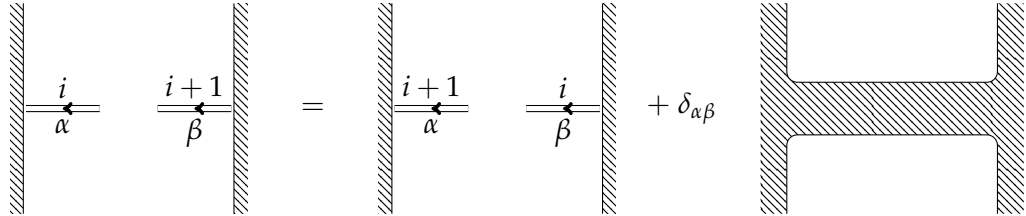
\begin{tikzpicture}[scale=0.7, shift = {(-5,0)}]

\begin{scope}[shift={(-2,0)}]
\fill[pattern=north west lines ] (2.75,0) rectangle (7.25,4);
\fill[color=white] (3,-0.1) rectangle (7,4.1); 
\draw (3,4)--(3,0);
\draw (7,4)--(7,0);
\draw[dbl-<-](3,2) -- (4.5,2) node[midway,above]{$i$} node[midway, below]{$\alpha$};
\draw[dbl-<-](5.5,2) -- (7,2) node[midway,above]{$i+1$} node[midway, below]{$\beta$}; 
\end{scope}
\node at (6.5,2) {$=$}; 
\begin{scope}[shift={(5,0)}]
\fill[pattern=north west lines ] (2.75,0) rectangle (7.25,4);
\fill[color=white] (3,-0.1) rectangle (7,4.1); 
\draw (3,4)--(3,0);
\draw (7,4)--(7,0);
\draw[dbl-<-](3,2) -- (4.5,2) node[midway,above]{$i+1$} node[midway, below]{$\alpha$};
\draw[dbl-<-](5.5,2) -- (7,2) node[midway,above]{$i$} node[midway, below]{$\beta$};
\end{scope}
\node at (13.5,2) {$+\  \delta_{\alpha \beta}$};  

\begin{scope}[shift = {(12.5,0)}]

\draw[rounded corners] (3,4) -- (3,2.5) -- (7,2.5) -- (7,4);  
 \draw[rounded corners] (3,0) -- (3,1.5) -- (7,1.5) --  (7,0);
\fill[pattern=north west lines ]  (2.5,4) -- (3,4)[rounded corners] -- (3,2.5) -- (7,2.5) [sharp corners]--(7,4)  -- (7.5,4) -- (7.5,0) -- (7,0) [rounded corners] -- (7,1.5) -- (3,1.5) [sharp corners]-- (3,0)  -- (2.5,0) -- cycle; 

\end{scope}
\end{tikzpicture}
As before, the Roman labels on the arrows indicate the position in the total order on the set of marked points (or arrows), and the Greek labels are the elements of $\{1,\dots,N\}$ present on every marked point. 
 \item[(P3)]  Next, suppose that $p_i$, $p_{i+1}$ are type $1$ marked points where $p_i$ is incoming, $p_{i+1}$ is outgoing, and $p_{i+1}$ is directly next to $p_i$ on the same boundary component.  Then,
$$
\Sigma(p_1,\dots,p_i,p_{i+1},\dots,p_n) = \Sigma(p_1,\dots,p_{i+1},p_{i},\dots,p_n) +  \Sigma(p_1,\dots,\what{p_i},\what{ p_{i+1}},\dots,p_n) .  
$$
This is illustrated as follows: 

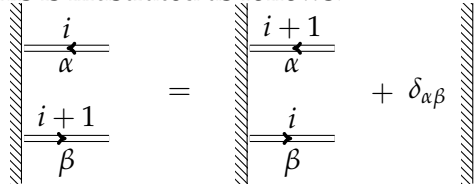
\begin{tikzpicture}[scale=0.6]
\begin{scope}
\fill[pattern=north west lines ] (0.75,0) rectangle (1,4);
\draw (1,4)--(1,0);
\draw[dbl-<-](1,3) -- (3,3) node[midway,above]{$i$}  node[midway, below]{$\alpha$};
\draw[dbl->-](1,1) -- (3,1) node[midway,above]{$i+1$} node[midway,below]{$\beta$};
\end{scope}
\node at (4.5,2) {$=$};
\begin{scope}[shift={(5,0)}]

\fill[pattern=north west lines ] (0.75,0) rectangle (1,4);
\draw (1,4)--(1,0);
\draw[dbl-<-](1,3) -- (3,3) node[midway,above]{$i+1$} node[midway, below]{$\alpha$};
\draw[dbl->-](1,1) -- (3,1) node[midway,above]{$i$} node[midway,below]{$\beta$};
\end{scope}

\node at (9,2) {$+$}; 
\begin{scope}[shift={(10,0)}]

\node at (0,2) {$\delta_{\alpha \beta}$};
\fill[pattern=north west lines ] (0.75,0) rectangle (1,4);
\draw (1,4)--(1,0);
\end{scope}
\end{tikzpicture}

\end{enumerate}
The relations (P0), (P1), (P2) and (P3) allow us to arbitrarily permute the total order on the set of marked points, at the price of introducing some additional topological complexity to the surface. 

Let's now discuss the cutting relations.
\begin{enumerate}
 \item[(C0)] Consider a simple closed curve in $\Sigma$ and let $\Sigma'$ be the surface obtained by cutting $\Sigma$ along the curve and then filling in each of the two new boundary components with a disc. If there are any components of the resulting surface with empty boundary, we remove them, resulting in a surface $\Sigma'$ each of whose components has non-empty boundary.  Then,
$$
\Sigma(p_1,\dots,p_n)  = \Sigma'(p_1,\dots,p_n).  
$$
\item[(C1)]  Let $\Sigma(p_1,\dots,p_n)$ be a marked surface.  Suppose that $p_i,p_j,p_k,p_l$ are type $1$ marked points, where $p_i,p_k$ are incoming and $p_j,p_l$ are outgoing.  Suppose that $p_j$ is directly to the right of $p_i$ on the same boundary component, and that similarly $p_l$ is directly to the right of $p_k$.  Suppose that we choose an arc on $\Sigma$ which starts between $p_i$ and $p_k$ and ends between $p_k$ and $p_l$.  Let $\Sigma'$ be the surface obtained by cutting $\Sigma$ along this arc.  Then, 
$$
\Sigma(p_1,\dots,p_n) = \Sigma'(p_1,\dots,p_n). 
$$ 
This is indicated in the following diagram:

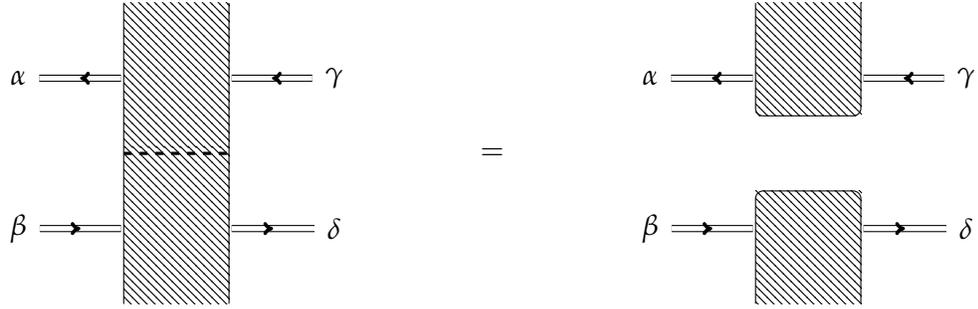
\begin{tikzpicture}[xscale=0.7]
\begin{scope}
\fill[pattern= north west lines ] (3,0) rectangle (5,4); 
\node(A) at (1,3) {$\alpha$};
\node(B) at (1,1) {$\beta$};
\node(C) at (7,3) {$\gamma$};
\node(D) at (7,1) {$\delta$};
\draw[dbl->-] (3,3) -- (A);
\draw[dbl->-] (B) -- (3,1);
\draw[dbl->-] (C) -- (5,3);
\draw[dbl->-] (5,1) -- (D);
\draw (3,0) -- (3,4);
\draw (5,0) -- (5,4);
\draw[very thick, dashed] (3,2) -- (5,2);
\end{scope}

\node at (10,2) {$=$};
\begin{scope}[ shift = {(12,0)}]
%\fill[pattern= north west lines] (3,0) rectangle (5,4); 
\node(A) at (1,3) {$\alpha$};
\node(B) at (1,1) {$\beta$};
\node(C) at (7,3) {$\gamma$};
\node(D) at (7,1) {$\delta$};
\draw[dbl->-] (3,3) -- (A);
\draw[dbl->-] (B) -- (3,1);
\draw[dbl->-] (C) -- (5,3);
\draw[dbl->-] (5,1) -- (D);

\fill[pattern=north west lines ]  (3,4) -- (3,2.5) -- (5,2.5) -- (5,4) -- cycle;
\fill[pattern=north west lines ]  (3,0) -- (3,1.5) -- (5,1.5) -- (5,0) -- cycle;

\draw[rounded corners = 1mm]  (3,4) -- (3,2.5) -- (5,2.5) -- (5,4);
\draw[rounded corners = 1mm] (3,0) -- (3,1.5) -- (5,1.5) -- (5,0);
\end{scope}
\end{tikzpicture}

\end{enumerate}
The final relation we call the ``moving'' relation, because it involves moving an incoming marked point past an outgoing marked point on the boundary of a surface. 
\begin{enumerate}
\item[(M)]
Consider a marked surface $\Sigma(p_1,\dots,p_n)$.    Suppose that $p_n$ is immediately to the right of $p_{n-1}$ on the same boundary component. 

 We want to consider various ways we can label the marked points $p_{n-1},p_{n}$, while leaving them in the same position on the boundary of $\Sigma$.  We let $\Sigma(p_1,\dots,p_{n-1}^{in}, p_{n}^{out})$ be the configuration where $p_{n-1}$ is incoming and $p_{n}$ is outgoing,  and both are of type $2$. Let $\Sigma(p_1,\dots,p_{n-1}^{out}, p_{n}^{in})$ be the configuration where $p_{n-1}$ is outgoing and $p_{n}$ is incoming, but again they are both of type $2$.  Finally, for each $\alpha,\beta \in \{1,\dots,N\}$ let $\Sigma(p_1,\dots,p_{n-2},p_{n-1}^{\alpha}, p_{n}^{\beta})$ be the same surface but where $p_{n-1}$ and $p_{n}$ are both type $1$ marked points labelled by $\alpha$ and $\beta$ respectively, and where $p_{n-1}$ is incoming and $p_{n}$ is outgoing.  Then, we have the relation
\begin{multline} 
 \Sigma(p_1,\dots,p_{n-2},p^{in}_{n-1},p_{n}^{out}) -  \Sigma(p_1,\dots,p_{n-2},p^{out}_{n-1},p_{n}^{in}) \\+ \hbar \sum_{\alpha = 1}^{N} \Sigma(p_1,\dots,p_{n-2},p_{n-1}^{\alpha}, p_{n}^{\alpha}) -\hbar \tfrac{1}{2} \Sigma(p_1,\dots,p_{n-2},p_{n-1}\asymp p_{n})  = c \Sigma (p_1,\dots,p_{n-2},\what{p_{n-1}},\what{p_{n}}). \label{equation_moment_map_universal} 
\end{multline}
Here $c$ is the other parameter of our algebra. 

This relation is best understood in pictorial terms:

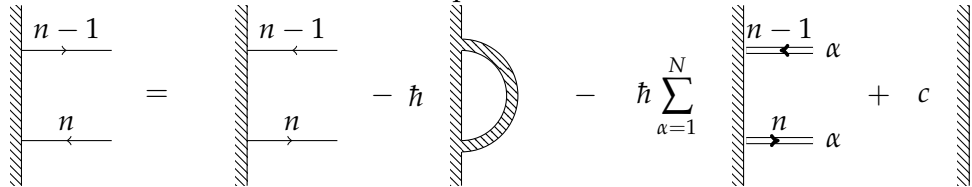
\begin{tikzpicture}[scale=0.6]
\begin{scope}
\fill[pattern=north west lines ] (0.75,0) rectangle (1,4);
\draw (1,4)--(1,0);
\draw[->-](1,3) -- (3,3) node[midway,above]{$n-1$};
\draw[->-](3,1) -- (1,1) node[midway,above]{$n$};
\end{scope}
\node at (4,2) {$=$};
\begin{scope}[shift={(5,0)}]
\fill[pattern=north west lines ] (0.75,0) rectangle (1,4);
\draw (1,4)--(1,0);
\draw[->-](3,3) -- (1,3) node[midway,above]{$n-1$};
\draw[->-](1,1) -- (3,1) node[midway,above]{$n$};
\end{scope}

\begin{scope}[shift={(9.75,0)}]

\filldraw[pattern=north west lines , even odd rule] (1,2) circle (1) (1,2) circle (1.25);
\fill[color = white] (-1,0) rectangle (1,4);
\fill[pattern=north west lines ] (0.75,0) rectangle (1,4);
\draw (1,3) -- (1,1);
\draw (1,4) -- (1,3.25);
\draw (1,0) -- (1,0.75);
\node at (0,2) {$\hbar$}; 
\end{scope}
\node at (9,2) {$-$};
\node at (13.5,2) {$-$}; 

\begin{scope}[shift={(16,0)}]
\node at (-0.7,2) {$\hbar \displaystyle \sum_{\alpha = 1}^{N} $} ;  
\fill[pattern=north west lines ] (0.75,0) rectangle (1,4);
\draw (1,4)--(1,0);
\node(N1) at (3,3) {$\alpha$};
\node(N2) at (3,1) {$\alpha$};
\draw[dbl->-](N1) -- (1,3)  node[midway,above]{$n-1$}; 
\draw[dbl->-] (1,1) -- (N2) node[midway,above]{$n$};
\end{scope}
\node at (20,2) {$+$} ; 
\begin{scope}[shift={(21,0)}]
\node at (0,2) {$c$};
\fill[pattern=north west lines ] (0.75,0) rectangle (1,4);
\draw (1,4)--(1,0);
\end{scope}
\end{tikzpicture}

\end{enumerate}

\begin{remark}
For relation (M), we are only imposing it when it involves the last two marked $p_{n-1},p_n$ in the  total order on the set of marked points.  
\end{remark}

Let us record some statements about the algebra $A_{N,\hbar,c}$.
\begin{lemma} 
 For any non-zero constant $\lambda$, there is an isomorphism
$$
A_{N,\hbar,c} \iso A_{N, \lambda \hbar,\lambda c}. 
$$
\end{lemma}
\begin{proof} 
This isomorphism is given by multiplying a marked surface $\Sigma(p_1,\dots,p_n)$ by $\lambda$ raised to the power of the number of incoming type $2$ marked points on $\Sigma$.  Doing this takes the relations with parameters $(\hbar,c)$ to the relations with parameters $(\lambda \hbar, \lambda c)$.  
\end{proof}

In section \ref{appendix_flatness} we prove the following result. 
\begin{theorem}
If we treat $\hbar$ as a formal parameter, then for all $c \neq 0$, $A_{N,\hbar,c}$ is a flat family of algebras over $\C[[\hbar]]$. 
\end{theorem}
The proof of this is a bit tricky: we will rely on a relationship between $A_{N,\hbar,c}$ and deformation quantizations of certain quiver varieties.

Another fundamental feature of this family of algebras is the following. 
\begin{proposition}
 For $c \neq 0$,  $A_{N,\hbar=0,c}$ is naturally isomorphic to the universal enveloping algebra of the Lie algebra $\op{Diff}_{c}(\C)\otimes \mf{gl}_N$.
\label{proposition_classical_limit_universal_enveloping}
\end{proposition}
\begin{proof}
When $\hbar = 0$, the relations (P0) and (P1) concerning permutations of the ordering on the marked points become very simple: one has
$$
\Sigma(p_1,\dots,p_i,p_{i+1},\dots,p_n) = \Sigma(p_1,\dots,p_{i+1},p_{i},\dots,p_n).
$$ 
has long as at least one of the $p_i,p_{i+1}$ is of type $2$.  We still have the more complicated relation (P2), 
$$
\Sigma(p_1,\dots,p_i^{\alpha},p_{i+1}^{\beta},\dots,p_n) = \Sigma(p_1,\dots,p^{\beta}_{i+1},p^{\alpha}_{i},\dots,p_n) + \delta_{\alpha,\beta} \Sigma(p_1,\dots,p_i \asymp p_{i+1}, \dots,p_n)
$$
if $p_i$ is of type $1$ and incoming and labelled by $\alpha \in \{1,\dots,N\}$, and $p_{i+1}$ is of type $1$ and outgoing and labelled by $\beta$.

Consider relation (C0) which states that we can cut a surface along a simple closed curve.  By repeatedly applying this relation, we find that the generators of $A_{N,\hbar=0,c }$ can be taken to be disjoint unions of discs with marked points. Relation (C1) allows us to cut these disks further, so that we can take the generators are  discs with at most two marked points of type $1$.   

By permuting the total ordering on the set of marked points on the boundary of a disc using relations (P0) and (P1), we can always consider discs $D(p_1,\dots,p_n)$ where $p_i$ is immediately to the right of $p_{i-1}$ on the boundary of $D$.  There are two types of such discs: those with no type $1$ marked points, and those with two type $1$ marked points. If we have two type $1$ marked points, we can assume that these are $p_1$ and $p_n$, where $p_1$ is outgoing and $p_n$ is incoming.

In what follows, we will always assume that the ordering on our marked points is compatible with the cyclic order on the boundary of a disc.  Let us introduce some new notation to refer to the possible classes of discs which are our generators.  If a disc has an incoming type $2$ marked point in position $i$, we can indicate it by placing the symbol $\shortdownarrow$ like $D(p_1,\dots, \shortdownarrow, p_{i+1},\dots,p_n)$, where the labels $p_i$ indicate that they can be any kind of marked point.  Similarly, $\shortuparrow$ indicates an outgoing type $2$ marked point. We denote an incoming type $1$ marked point labelled by $\alpha \in \{1,\dots,N\}$ by $ {\alpha}\Downarrow$, and an outgoing type $1$ marked point by $ {\alpha}\Uparrow$.     

Then, note that relation (M) when $\hbar = 0$ becomes
\begin{multline*}
D(p_1,\dots,p_{i-1},\shortdownarrow,\shortuparrow,p_{i+2},\dots,p_n) - D(p_1,\dots,p_{i-1}, \shortuparrow, \shortdownarrow, p_{i+2},\dots,p_n)  \\ 
= -c D(p_1,\dots,p_{i-1},p_{i+2},\dots,p_n).   \tag{M'} 
\end{multline*}
We can use this relation to move any type $2$ incoming marked points to the right of type $2$ outgoing marked points. Thus, we can take the generators of our algebra to be of two types:
\begin{align*} 
 D( {\alpha}\Uparrow,\uparrow^k,\downarrow^l, {\beta}\Downarrow)&\\
D( \uparrow^k,\downarrow^l)&
\end{align*}
where $\uparrow^k$ indicates that we take $k$ copies of this marked point (and similarly for $\downarrow^l$). 

As the next step in our argument, we will show that 
$$
D(\uparrow^k, \downarrow^l) = 0
$$
in $A_{N,\hbar = 0,c}$. 

We can use equation (M') to move an incoming marked point past an outgoing marked point.   We find that
$$
 D(\shortdownarrow^k, \shortuparrow^l) = D(\shortdownarrow^{k-1},\shortuparrow,\shortdownarrow, \shortuparrow^{l-1}) -  c D(\shortdownarrow^{k-1}, \shortuparrow^{l-1}). 
$$ 
By repeatedly applying this relation we find that
$$
 D(\shortdownarrow^k, \shortuparrow^l) -  D(\shortdownarrow^{k-1}, \shortuparrow^{l},\shortdownarrow) = - l c D(\shortdownarrow^{k-1}, \shortuparrow^{l-1}). 
$$
But cyclic symmetry of the disc tells us that
$$
D (\shortdownarrow^k, \shortuparrow^l) = D (\shortdownarrow^{k-1}, \shortuparrow^l, \shortdownarrow).
$$
Therefore 
$$
 l c D(\shortdownarrow^{k-1}, \shortuparrow^{l-1}) = 0 
$$
and so, since $c \neq 0$, $ D(\shortdownarrow^{k-1}, \shortuparrow^{l-1}) = 0$ as desired.  

Our next goal is to define a map of associative algebras from $U( \mf{gl}_N \otimes \op{Diff}_c(\C))$ to $A_{N,\hbar = 0,c}$.  We will start by defining a linear map from $\mf{gl}_N \otimes \op{Diff}_{c}(\C)$ to $A_{N,\hbar = 0,c}$.  This linear map sends 
$$
E_{\alpha \beta} z^{k_1}\partial^{l_1} \dots z^{k_n} \partial^{l_n}  \mapsto D( {\alpha}\Uparrow,\downarrow^{k_1},\uparrow^{l_1}, \dots, \downarrow^{k_n} , \uparrow^{l_n},  {\beta}\Downarrow).
$$
Here $E_{\alpha \beta}$ is the elementary matrix. 

Let us check that this linear map is well-defined.  For this, note that the relation
$$
E_{\alpha \beta} f(z,\partial) \partial z g(z,\partial) = E_{\alpha \beta} f(z,\partial) z \partial g(z,\partial) + c E_{\alpha\beta} f(z,\partial) g(z,\partial)
$$
that holds in $\mf{gl}_N \otimes \op{Diff}_{c} (\C)$ precisely matches the relation of equation (M'). 

We claim that this linear map is a map of Lie algebras, and so extends to a map of associative algebras
$$
U(\op{Diff}_c(\C) \otimes \mf{gl}_N) \to A_{N,\hbar = 0,c}. 
$$ 
To check this, we need to check that the commutation relations hold.  We find
\begin{multline*} 
 \left[ D( {\alpha}\Uparrow,\uparrow^k,\downarrow^l, {\beta}\Downarrow), D( {\gamma}\Uparrow,\uparrow^m,\downarrow^n, {\delta}\Downarrow)\right] \\
 = \delta_{\beta \gamma} D( {\alpha}\Uparrow, \uparrow^k,\downarrow^l,\uparrow^m,\downarrow^n,  {\delta}\Downarrow)  - \delta_{\delta \alpha } D( {\gamma}\Uparrow, \uparrow^m,\downarrow^n,\uparrow^k,\downarrow^l,  {\beta}\Downarrow)  + \delta_{\alpha \delta} \delta_{\beta \gamma} D(\uparrow^k, \downarrow^l, \uparrow^m, \downarrow^n ) .   
\end{multline*}
Since the disk in the last term has no type $1$ marked points, it is zero, so we can remove it from the result of the commutator.  

If we do so, we find the result matches the commutation relation
$$
\left[ E_{\alpha \beta} z^k \partial^l , E_{\gamma \delta} z^m \partial^n \right] = \delta_{\beta \gamma} E_{\alpha \delta} z^k \partial^l z^m \partial^n - \delta_{\delta\alpha} E_{\gamma \beta} z^m \partial^n z^k \partial^l
$$
that holds in $\mf{gl}_N \otimes \op{Diff}(\C)$. 

We therefore find a map of associative algebras
$$
U(\op{Diff}(\C) \otimes \mf{gl}_N) \to A_{N,\hbar = 0,c}. 
$$
This map is surjective, since the image contains the generators $D(\alpha \Uparrow, \uparrow^k, \downarrow^l, \beta \Downarrow)$ of $A_{N,\hbar=0,c}$.  Injectivity follows from the analysis of flatness of the family of algebras over $\C[[\hbar]]$ presented in section \ref{appendix_flatness}.  
\end{proof}

Next, let us write down the first-order deformation of the algebra $U(\op{Diff}(\C) \otimes \mf{gl}_N)$ we obtain by working modulo $\hbar^2$. 
\begin{proposition}
Let $E_{\alpha \beta}$ denote the elementary matrix. The first-order deformation of the product is defined by deforming the commutator according to the formula 
\begin{align*}
[E_{\alpha \beta} \partial^m z^n, E_{\gamma \delta} \partial^r z^s] =& \delta_{\beta \gamma} E_{\alpha \delta} (\partial^m z^n \partial^r z^s) - \delta_{\delta\alpha} E_{\gamma \beta} (\partial^r z^s \partial^m z^n) \\
 &+ \hbar \sum_{i = 1}^{m} \sum_{j = 1}^s \left(E_{\alpha \delta} \partial^{i-1} z^{s-j}\right)\left(E_{\gamma\beta} \partial^r z^{j-1} \partial^{m-i} z^n \right)\\& - \hbar \sum_{i =1}^{n} \sum_{j = 1}^r \left(E_{\alpha \delta} \partial^m z^{i-1} \partial^{r-j} z^s \right) \left( E_{\gamma \beta} \partial^{j-l} z^{n-i} \right). 
\end{align*}
Here $\hbar$ is the deformation parameter.  
\label{proposition_combinatorial_one_loop}
\end{proposition}
\begin{proof}
 The calculation is given diagramatically. According to our rules, the coefficient of $\hbar$ in the commutator 
$$
 \left[ D( {\alpha}\Uparrow,\uparrow^m,\downarrow^n, {\beta}\Downarrow), D( {\gamma}\Uparrow,\uparrow^r,\downarrow^s, {\delta}\Downarrow)\right] 
$$
is given by a sum over two types of diagram. There are diagrams where we join a single type $2$ marked point on one disc with one on the other, with an appropriate sign; and there are diagrams which in addition join a type $1$ marked point on one disc with one on the other.

The first type of diagram is given explicitly by the sum 

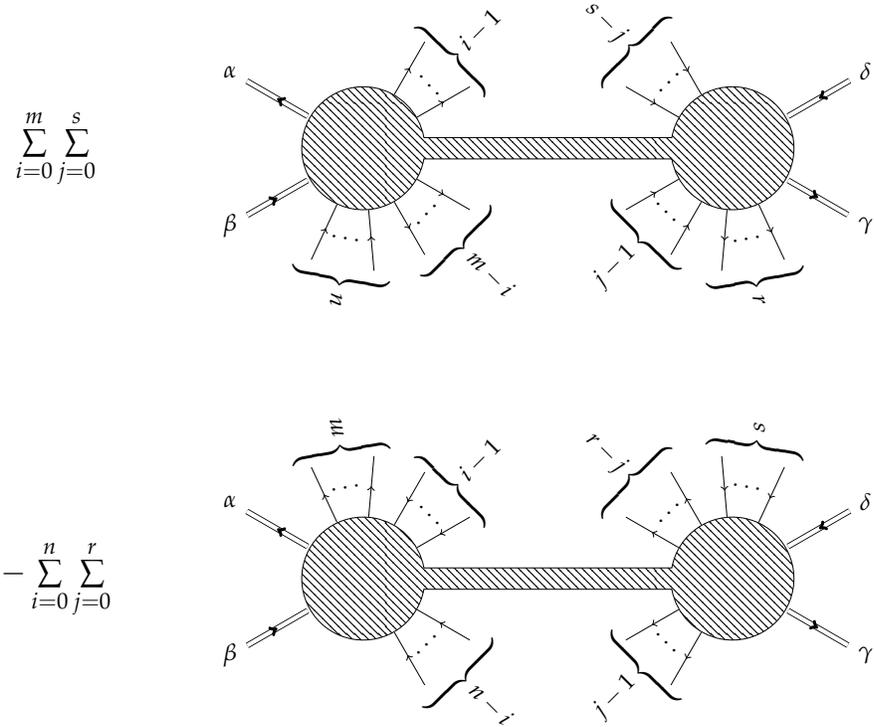
\begin{center}
\resizebox{12cm}{!}{ \begin{tikzpicture}%[transform canvas = {scale=0.6}]

\begin{scope}
\node[scale=1.3] at (180:5){$\sum\limits_{i = 0}^m \sum\limits_{j = 0}^s$};

\node[circle,  minimum size=2cm](N1) at (0,0){}; 
\node[circle, minimum size=2cm](N2) at (6,0){};
\coordinate (A) at (10:1);
\coordinate (B) at ( $ (N2) + (170:1)$ );
\coordinate (D) at (-10:1); 
\coordinate (C) at ( $(N2) + (190:1)$ );
\filldraw[ pattern=north west lines ] (A) -- (B)  arc(170:-170:1) (C) -- (D) arc (350:10:1) (A);
%\draw[color=red] (c) -- (d);
%\draw[color=red]  (c) arc (-170:170:1) (b);
%\draw[color=red] (a) arc (10:350:1) (d);

\node(a) at (150:2.5){$\alpha$};
\node(b) at (210:2.5){$\beta$};
\node(d) at ( $(N2) + (30:2.5)$ ) {$\delta$};
\node(g) at ( $(N2) + (-30:2.5)$ ) {$\gamma$};

\draw[dbl->-] (N1) -- (a);
\draw[dbl-<-] (N1) --(b);
\draw[dbl-<-] (N2) -- (d);
\draw[dbl->-] (N2) -- (g);

%Now the type 2 arrows, first outgoing 

\draw[->-] (N1) -- (60:2);

\node[rotate=45] at (45:2.3){$\left. \begin{array}{c}\\    \\ \\ \end{array} \right\} i-1 $}; 

\node[rotate=135]  at (45:1.5){$\cdots$}; 
\draw[->-] (N1) -- (30:2);

%midpoint

\draw[->-] (N1) -- (-30:2); 
\node[rotate=45]  at (-45:1.5){$\cdots$};

\draw[->-] (N1) -- (-60:2);
\node[rotate=-45] at (-45:2.5){$\left. \begin{array}{c}  \\ \\ \\ \end{array} \right\} m-i $};

\draw[-<-] (N1) -- (-85:2); 
\node[rotate=-10]  at (-100:1.5){$\cdots$};
\draw[-<-] (N1) -- (-115:2);

\node[rotate=-100] at (-100:2.1){$\left. \begin{array}{c}  \\ \\ \\ \end{array} \right\} n $};

%Now a scope  for the second node. The rotate option is a bit funny we have to add additional rotations to the text within each node.
\begin{scope}[shift={(6,0)}, rotate=180]

\draw[-<-] (N2) -- (60:2);

\node[rotate=45] at (45:2.3){$j-1 \left\{ \begin{array}{c}\\    \\ \\ \end{array} \right.  $}; 

\node[rotate=315]  at (45:1.5){$\cdots$}; 
\draw[-<-] (N2) -- (30:2);

\draw[-<-] (N2) -- (-30:2); 
\node[rotate=225]  at (-45:1.5){$\cdots$};

\draw[-<-] (N2) -- (-60:2);%  node[midway, above,rotate=-70]{$m+1$};
\node[rotate=-45] at (-45:2.5){$s- j \left\{ \begin{array}{c}  \\ \\ \\ \end{array} \right. $};

\draw[->-] (N2) -- (85:2); 
\node[rotate=10]  at (100:1.5){$\cdots$};
\draw[->-] (N2) -- (115:2);

\node[rotate=-80] at (100:2.1){$\left. \begin{array}{c}  \\ \\ \\ \end{array} \right\} r $};

\end{scope}

\end{scope}

\begin{scope}[shift={(0,-7)}]
\node[scale=1.3] at (180:5){$- \sum\limits_{i = 0}^n \sum\limits_{j = 0}^r$}; 
\node[circle,  minimum size=2cm](N1) at (0,0){}; 
\node[circle, minimum size=2cm](N2) at (6,0){};
\coordinate (A) at (10:1);
\coordinate (B) at ( $ (N2) + (170:1)$ );
\coordinate (D) at (-10:1); 
\coordinate (C) at ( $(N2) + (190:1)$ );
\filldraw[ pattern=north west lines ] (A) -- (B)  arc(170:-170:1) (C) -- (D) arc (350:10:1) (A);
%\draw[color=red] (c) -- (d);
%\draw[color=red]  (c) arc (-170:170:1) (b);
%\draw[color=red] (a) arc (10:350:1) (d);

\node(a) at (150:2.5){$\alpha$};
\node(b) at (210:2.5){$\beta$};
\node(d) at ( $(N2) + (30:2.5)$ ) {$\delta$};
\node(g) at ( $(N2) + (-30:2.5)$ ) {$\gamma$};

\draw[dbl->-] (N1) -- (a);
\draw[dbl-<-] (N1) --(b);
\draw[dbl-<-] (N2) -- (d);
\draw[dbl->-] (N2) -- (g);

%Now the type 2 arrows, first outgoing 

\draw[->-] (N1) -- (85:2); 
\node[rotate=10]  at (100:1.5){$\cdots$};
\draw[->-] (N1) -- (115:2);

\node[rotate=100] at (100:2.1){$\left. \begin{array}{c}  \\ \\ \\ \end{array} \right\} m $};

%next incoming

\draw[-<-] (N1) -- (60:2);

\node[rotate=45] at (45:2.3){$\left. \begin{array}{c}\\    \\ \\ \end{array} \right\} i-1 $}; 

\node[rotate=135]  at (45:1.5){$\cdots$}; 
\draw[-<-] (N1) -- (30:2);

%midpoint

\draw[-<-] (N1) -- (-30:2); 
\node[rotate=45]  at (-45:1.5){$\cdots$};

\draw[-<-] (N1) -- (-60:2);
\node[rotate=-45] at (-45:2.5){$\left. \begin{array}{c}  \\ \\ \\ \end{array} \right\} n-i $};

%Now a scope  for the second node. The rotate option is a bit funny we have to add additional rotations to the text within each node.
\begin{scope}[shift={(6,0)}, rotate=180]

\draw[-<-] (N2) -- (-85:2); 
\node[rotate=-10]  at (-100:1.5){$\cdots$};
\draw[-<-] (N2) -- (-115:2);

\node[rotate=80] at (-100:2.1){$\left. \begin{array}{c}  \\ \\ \\ \end{array} \right\} s $};

%second group
\draw[->-] (N2) -- (60:2);

\node[rotate=45] at (45:2.3){$j-1 \left\{ \begin{array}{c}\\    \\ \\ \end{array} \right.  $}; 

\node[rotate=315]  at (45:1.5){$\cdots$}; 
\draw[->-] (N2) -- (30:2);

%third group
\draw[->-] (N2) -- (-30:2); 
\node[rotate=225]  at (-45:1.5){$\cdots$};

\draw[->-] (N2) -- (-60:2);
\node[rotate=-45] at (-45:2.5){$r- j \left\{ \begin{array}{c}  \\ \\ \\ \end{array} \right. $};

\end{scope}
\end{scope}
\end{tikzpicture} }

\end{center}

This results in an expression as a sum over discs each with $4$ type $1$ marked points.  We can apply relation (C1) to cut each disc into two discs,  as shown in the following diagram (where we cut along the dotted line):

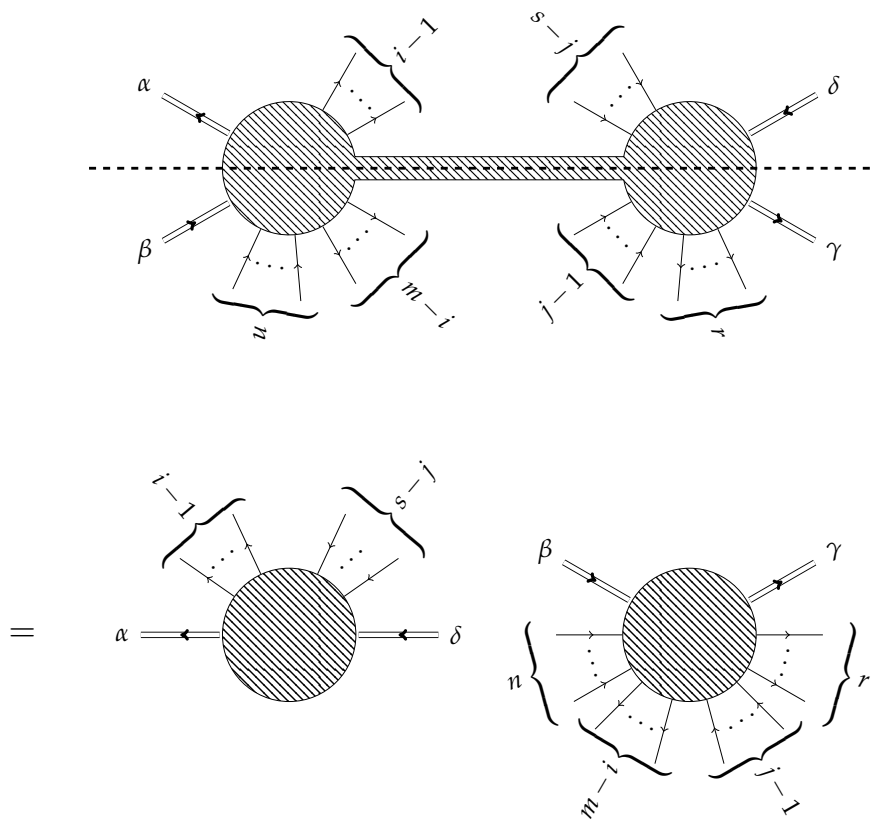
\begin{center}

\resizebox{12cm}{!} {\begin{tikzpicture}%[transform canvas={scale=0.6}]
%This picture indicates cutting a disc in two to get two discs
\node[circle,  minimum size=2cm](N1) at (0,0){}; 
\node[circle, minimum size=2cm](N2) at (6,0){};
\coordinate (A) at (10:1);
\coordinate (B) at ( $ (N2) + (170:1)$ );
\coordinate (D) at (-10:1); 
\coordinate (C) at ( $(N2) + (190:1)$ );
\filldraw[ pattern=north west lines ] (A) -- (B)  arc(170:-170:1) (C) -- (D) arc (350:10:1) (A);

\node(a) at (150:2.5){$\alpha$};
\node(b) at (210:2.5){$\beta$};
\node(d) at ( $(N2) + (30:2.5)$ ) {$\delta$};
\node(g) at ( $(N2) + (-30:2.5)$ ) {$\gamma$};

\draw[dbl->-] (N1) -- (a);
\draw[dbl-<-] (N1) --(b);
\draw[dbl-<-] (N2) -- (d);
\draw[dbl->-] (N2) -- (g);

%Now the type 2 arrows, first outgoing 

\draw[->-] (N1) -- (60:2);

\node[rotate=45] at (45:2.3){$\left. \begin{array}{c}\\    \\ \\ \end{array} \right\} i-1 $}; 

\node[rotate=135]  at (45:1.5){$\cdots$}; 
\draw[->-] (N1) -- (30:2);

%midpoint

\draw[->-] (N1) -- (-30:2); 
\node[rotate=45]  at (-45:1.5){$\cdots$};

\draw[->-] (N1) -- (-60:2);%  node[midway, above,rotate=-70]{$m+1$};
\node[rotate=-45] at (-45:2.5){$\left. \begin{array}{c}  \\ \\ \\ \end{array} \right\} m-i $};

\draw[-<-] (N1) -- (-85:2); 
\node[rotate=-10]  at (-100:1.5){$\cdots$};
\draw[-<-] (N1) -- (-115:2);

\node[rotate=-100] at (-100:2.1){$\left. \begin{array}{c}  \\ \\ \\ \end{array} \right\} n $};

%Now a scope  for the second node. The rotate option is a bit funny we have to add additional rotations to the text within each node.
\begin{scope}[shift={(6,0)}, rotate=180]

\draw[-<-] (N2) -- (60:2);

\node[rotate=45] at (45:2.3){$j-1 \left\{ \begin{array}{c}\\    \\ \\ \end{array} \right.  $}; 

\node[rotate=315]  at (45:1.5){$\cdots$}; 
\draw[-<-] (N2) -- (30:2);

\draw[-<-] (N2) -- (-30:2); 
\node[rotate=225]  at (-45:1.5){$\cdots$};
\draw[-<-] (N2) -- (-60:2);%  node[midway, above,rotate=-70]{$m+1$};
\node[rotate=-45] at (-45:2.5){$s- j \left\{ \begin{array}{c}  \\ \\ \\ \end{array} \right. $};

\draw[->-] (N2) -- (85:2); 
\node[rotate=10]  at (100:1.5){$\cdots$};
\draw[->-] (N2) -- (115:2); 
\node[rotate=-80] at (100:2.1){$\left. \begin{array}{c}  \\ \\ \\ \end{array} \right\} r $};

\end{scope}
%A line cutting the first disc in two
\draw[very thick,dashed] (-3,0) -- (9,0);

\begin{scope}[shift={(0,-7)} , text opacity = 1]
%This scope contains two disjoint discs, shifted down
\node[scale=1.3] at (-4,0) {$=$}; 
\node[preaction={fill }, draw, circle,  minimum size=2cm, pattern=north west lines](N1) at (0,0){}; 
\node(a) at (180:2.5){$\alpha$};

\node(d) at (0:2.5) {$\delta$};

\draw[dbl->-] (N1) -- (a);
\draw[dbl-<-] (N1) --(d);

%Now the type 2 arrows, first outgoing 

\draw[->-] (N1) -- (115:2);

\node[rotate=-50] at (130:2.3){$i-1\left\{ \begin{array}{c}\\    \\ \\ \end{array} \right. $}; 

\node[rotate=40]  at (130:1.5){$\cdots$}; 
\draw[->-] (N1) -- (145:2);

%incoming group

\draw[-<-] (N1) -- (35:2); 
\node[rotate=--40]  at (50:1.5){$\cdots$};
\draw[-<-] (N1) -- (65:2);%  node[midway, above,rotate=-70]{$m+1$};
\node[rotate=50] at (50:2.5){$ \left. \begin{array}{c}  \\ \\ \\ \end{array} \right\} s-j $};

\begin{scope}[shift={(6,0)}]
%This scope is the second component disc

\node[preaction={fill }, draw, circle,  minimum size=2cm, pattern=north west lines](N1) at (0,0){}; 
\node(b) at (150:2.5){$\beta$};

\node(g) at (30:2.5) {$\gamma$};

\draw[dbl->-] (N1) -- (g);
\draw[dbl-<-] (N1) --(b);

%Now the type 2 arrows, left to right 

\draw[-<-] (N1) -- (180:2);

\node[rotate=15] at (195:2.3){$n \left\{ \begin{array}{c}\\    \\ \\ \end{array} \right.  $}; 

\node[rotate=105]  at (195:1.5){$\cdots$}; 
\draw[-<-] (N1) -- (210:2);

% second group

\draw[->-] (N1) -- (225:2);

\node[rotate=60] at (240:2.3){$m - i \left\{ \begin{array}{c}\\    \\ \\ \end{array} \right.  $}; 

\node[rotate=150]  at (240:1.5){$\cdots$}; 
\draw[->-] (N1) -- (255:2);

% third group

\draw[-<-] (N1) -- (285:2);

\node[rotate=300] at (300:2.3){$\left. \begin{array}{c}\\    \\ \\ \end{array} \right\} j-1 $}; 

\node[rotate=210]  at (300:1.5){$\cdots$}; 
\draw[-<-] (N1) -- (315:2);

%fourth group

\draw[->-] (N1) -- (330:2);

\node[rotate=345] at (345:2.3){$\left. \begin{array}{c}\\    \\ \\ \end{array} \right\} r $};

\node[rotate=75]  at (345:1.5){$\cdots$}; 
\draw[->-] (N1) -- (0:2);
\end{scope}
\end{scope}
\end{tikzpicture}  }
\end{center}
\vspace{-15pt}
Applying this relation to each disc in the sum results in the expression shown.

Recall that there is a second type of diagram that can potentially contribute, of the form
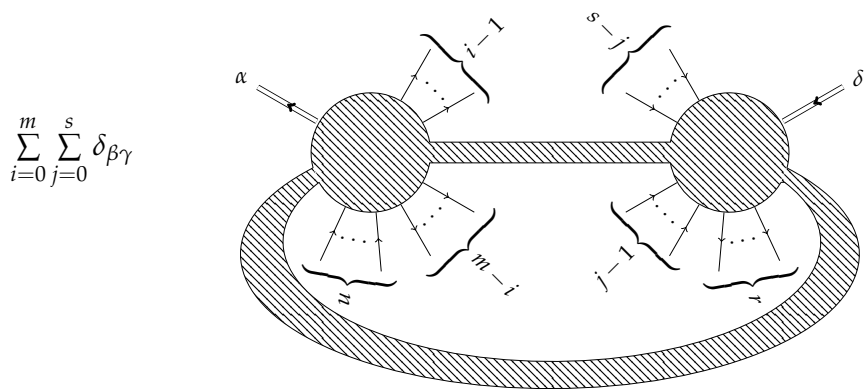
\begin{center}

\resizebox{12cm}{!} {\begin{tikzpicture}

\begin{scope}
\node[scale=1.3] at (180:5){$\sum\limits_{i = 0}^m \sum\limits_{j = 0}^s \delta_{\beta\gamma}$};

\node[circle,  minimum size=2cm](N1) at (0,0){}; 
\node[circle, minimum size=2cm](N2) at (6,0){};
\coordinate (A) at (10:1);
\coordinate (B) at ( $ (N2) + (170:1)$ );
\coordinate (D) at (-10:1); 
\coordinate (C) at ( $(N2) + (190:1)$ );
\coordinate (beta) at (210:1);
\coordinate (beta2) at (195:1);
\coordinate (gamma) at ( $(N2) + (-30:1)$ );
\coordinate (gamma2) at ( $(N2) + (-15:1)$ );
\filldraw[ pattern=north west lines ] (A) -- (B)  arc(170:-170:1)  --  (C) -- (D) arc (350:10:1) (A);
\draw[color=white, very thick](beta2) arc(195:210:1) -- (beta);
\draw[color=white, very thick](gamma) arc(-30:-15:1) -- (gamma2);

\fill [ pattern=north west lines ] (gamma2)  arc(400:140: 5.18 and 2.25) -- (beta2) -- (beta)  arc(150:390: 4.5 and 2) -- (gamma)-- (gamma2);

 \draw (gamma2)  arc(400:140: 5.18 and 2.25) -- (beta2); 
 \draw (beta)  arc(150:390: 4.5 and 2) -- (gamma);
%\draw[color=red] (c) -- (d);
%\draw[color=red]  (c) arc (-170:170:1) (b);
%\draw[color=red] (a) arc (10:350:1) (d);

\node(a) at (150:2.5){$\alpha$};
%\node(b) at (210:2.5){$\beta$};
\node(d) at ( $(N2) + (30:2.5)$ ) {$\delta$};
%\node(g) at ( $(N2) + (-30:2.5)$ ) {$\gamma$};

\draw[dbl->-] (N1) -- (a);
%\draw[dbl-<-] (N1) --(b);
\draw[dbl-<-] (N2) -- (d);
%\draw[dbl->-] (N2) -- (g);

%Now the type 2 arrows, first outgoing 

\draw[->-] (N1) -- (60:2);

\node[rotate=45] at (45:2.3){$\left. \begin{array}{c}\\    \\ \\ \end{array} \right\} i-1 $}; 

\node[rotate=135]  at (45:1.5){$\cdots$}; 
\draw[->-] (N1) -- (30:2);

%midpoint

\draw[->-] (N1) -- (-30:2); 
\node[rotate=45]  at (-45:1.5){$\cdots$};

\draw[->-] (N1) -- (-60:2);
\node[rotate=-45] at (-45:2.5){$\left. \begin{array}{c}  \\ \\ \\ \end{array} \right\} m-i $};

\draw[-<-] (N1) -- (-85:2); 
\node[rotate=-10]  at (-100:1.5){$\cdots$};
\draw[-<-] (N1) -- (-115:2);

\node[rotate=-100] at (-100:2.1){$\left. \begin{array}{c}  \\ \\ \\ \end{array} \right\} n $};

%Now a scope  for the second node. The rotate option is a bit funny we have to add additional rotations to the text within each node.
\begin{scope}[shift={(6,0)}, rotate=180]

\draw[-<-] (N2) -- (60:2);

\node[rotate=45] at (45:2.3){$j-1 \left\{ \begin{array}{c}\\    \\ \\ \end{array} \right.  $}; 

\node[rotate=315]  at (45:1.5){$\cdots$}; 
\draw[-<-] (N2) -- (30:2);

\draw[-<-] (N2) -- (-30:2); 
\node[rotate=225]  at (-45:1.5){$\cdots$};

\draw[-<-] (N2) -- (-60:2);%  node[midway, above,rotate=-70]{$m+1$};
\node[rotate=-45] at (-45:2.5){$s- j \left\{ \begin{array}{c}  \\ \\ \\ \end{array} \right. $};

\draw[->-] (N2) -- (85:2); 
\node[rotate=10]  at (100:1.5){$\cdots$};
\draw[->-] (N2) -- (115:2);

\node[rotate=-80] at (100:2.1){$\left. \begin{array}{c}  \\ \\ \\ \end{array} \right\} r $};

\end{scope}

\end{scope}

\end{tikzpicture} } 
\end{center}
There are similar diagrams where the $\alpha$ arrow is joined to the $\delta$ arrow, and a factor of $\delta_{\alpha \delta}$ appears instead of $\delta_{\alpha \gamma}$.  However, these diagrams all contribute zero. The reason is that the correspond to Riemann surfaces with a boundary component with no type $1$ marked point, i.e.\ with no double arrows.  By the cutting relation (C0), we can replace such a diagram by a disjoint of two discs, one of which has two type $1$ marked points, and the other has  none. By the argument in the proof of proposition \ref{proposition_classical_limit_universal_enveloping}, a disc with no type $1$ marked points  is zero modulo $\hbar$.   (These diagrams will make a contribution modulo $\hbar^2$, but  we do not consider that).  
\end{proof}

\subsection{A  basis for the algebra $A_{N,\hbar,c}$}
The fact that the algebra $A_{N,\hbar,c}$ is flat over $\hbar$, and at $\hbar = 0$ can be described as a universal enveloping algebra, tells us that there must exist a normal form for the labelled surfaces which span the algebra $A_{N,\hbar,c}$, with the feature that surfaces in this normal  form will give a basis for the algebra. In this subsection  we will introduce such a normal form.
\begin{definition}
\label{definition_reduced}
Consider a marked surface as defined at the beginning of this section.  We say such a surface is \emph{reduced}, if
\begin{enumerate} 
\item All the outgoing marked points (of type $1$ and type $2$) are placed, in the ordering on the set of marked points, after all the incoming marked points.  Since we can switch positions of adjacent outgoing (or incoming) marked points, the ordering on the set of outgoing (or incoming) marked points is irrelevant.
\item  Each component of the surface is a disc with precisely two type $1$ marked points.
\item In order coming from their arrangement on the boundary of each disc, outgoing type $2$ marked points are placed after incoming marked points.  
\end{enumerate} 
\end{definition} 
It is clear that the relations in the algebra allow one to make every surface into a sum of reduced surfaces. This tells us that reduced surfaces span $A_{N,\hbar,c}$.   The fact that the algebra $A_{N,\hbar,c}$ is flat over $\hbar$ implies that the reduced surfaces form a basis for $A_{N,\hbar,c}$.  Indeed, by using reduced surfaces we get an isomorphism of vector spaces
$$
\Sym^\ast \mf{gl}_N[z_1,z_2] \to A_{N,\hbar,c}.
$$
This isomorphism is defined as follows.  Given an expression like
$$
(E_{\alpha_1 \beta_1} z_1^{k_1} z_2^{l_1} ) \dots (E_{\alpha_n \beta_n} z_1^{k_n} z_2^{l_n})
$$ 
we build a reduced surface as follows.  The surface has $n$ components, each of which is a disc with two type $1$ marked points. The $i$'th disc has its type $1$ marked points (i.e.\ the ``double arrows'' in the figures) labelled by $\alpha_i$ and $\beta_i$.  It has $k_i$ outgoing type $2$ marked points and $l_i$ incoming type $2$ marked points.  The ordering on the set of marked points is now dictated by the condition that the surface be reduced.  The resulting element of $A_{N,\hbar,c}$ does not depend on the order in which these $n$ discs are placed.

The relations in the algebra $A_{N,\hbar,c}$ give us  an implicit algorithm to take the product of two reduced surfaces and  produce a sum of reduced surface.  The product of two reduced surfaces $\Sigma,\Sigma'$  is simply their disjoint union, but with the ordering on the set of marked points in which  we first list the marked points on $\Sigma$, and then those on $\Sigma'$.  The relations in the algebra allow us to change the ordering on the marked points and then to manipulate the resulting surface back into reduced form.  The statement that the algebra is flat over $\hbar$, and  that therefore the reduced surfaces form a basis, implies that no matter how we do this, we end up with the same sum of reduced surfaces.   Unfortunately, I don't know of a closed-form expression for the product of two reduced surfaces in terms  of reduced surfaces.  
  
\section{Matching the computations from gauge theory and algebra}
In this section we will see how to some computations from the $5$-dimensional gauge theory can be matched with those from the combinatorial algebra we have just defined. Recall that in the algebra $U_\hbar^{QFT} (\mf{gl}_N[z_1,z_2])$ which is Koszul dual to the algebra of observables of the $5$-dimensional gauge theory, we computed the commutator

\begin{multline} 
 [E_{\alpha\beta}z_1, E_{\gamma\delta} z_2] = \delta_{\beta \gamma}E_{\alpha\delta}(z_1 \ast_c z_2) - \delta_{\alpha\delta}E_{\gamma\beta}(z_2 \ast_c z_1)\\
 + \hbar 2^{-5}\pi^{-2}  E_{\alpha \delta} E_{\gamma \beta} - \hbar 2^{-5} \pi^{-2}  \sum_{\mu} \delta_{\beta \gamma}  E_{\alpha \mu} E_{\mu \delta}  
- \hbar 2^{-5} \pi^{-2} \sum_{\mu} \delta_{\alpha \delta} E_{\gamma \mu} E_{\mu \beta} .
\end{multline}
We want to match this with the commutator in the algebra $A_{N,\hbar,c} = U^{comb}_{\hbar}(\mf{gl}_N[z_1,z_2])$.  To do this, we must match up the generators of the algebra $U_\hbar^{QFT}(\mf{gl}_N[z_1,z_2])$ with those of $A_{N,\hbar,c}$.  We will work modulo $\hbar^2$, and we will only write down a map on those generators which are at most quadratic in the $z_i$.

The map we choose sends
\begin{align*} 
E_{\alpha \beta} & \mapsto D( {\alpha}\Uparrow, {\beta}\Downarrow) \\
E_{\alpha \beta} z_1 & \mapsto D(\alpha \Uparrow, \uparrow , \beta \Downarrow ) \\
E_{\alpha \beta} z_2 & \mapsto D(\alpha \Uparrow, \downarrow , \beta \Downarrow ) \\
E_{\alpha \beta} z_1 z_2 & \mapsto\tfrac{1}{2} \left(  D(\alpha \Uparrow, \uparrow, \downarrow , \beta \Downarrow ) +  D(\alpha \Uparrow, \downarrow, \uparrow , \beta \Downarrow ) \right) 
 \end{align*}
There is a slight sublety in this definition.  In the $5$-dimensional gauge theory, we viewed $\C[z_1,z_2]$ as being the usual space  of polynomials, with it's ordinary action of $SL(2,\C)$, but equipped with the Moyal product.  The Moyal product is an explicit formula which takes two polynomials $f,g$ of $z_1,z_2$ and produces a new polynomial $f \ast_c g$.  Thus, $z_1 z_2$ does not mean the algebra product of $z_1$ and $z_2$. Rather we have 
$$
 z_1 \ast_c z_2 = z_1 z_2 + \tfrac{1}{2} c . 
$$

Let us denote by 
$$
\rho : \C[z_1,z_2] \otimes \mf{gl}_N \mapsto A_{N,\hbar,c}
$$
this linear map. 
\begin{lemma}
In the algebra $A_{N,\hbar,c}$, working modulo $\hbar^2$, we have
\begin{multline} 
 [\rho ( E_{\alpha\beta}z_1 ) ,\rho( E_{\gamma\delta} z_2) ] = \delta_{\beta \gamma}\rho(E_{\alpha\delta}z_1 \ast_c z_2)  - \delta_{\alpha\delta}\rho ( E_{\gamma\beta}(z_2 \ast_c z_2)) \\
 + \hbar  \rho ( E_{\alpha \delta}) \rho( E_{\gamma \beta}) - \hbar \sum_{\mu} \delta_{\beta \gamma} \rho( E_{\alpha \mu}) \rho( E_{\mu \delta}) 
- \hbar \sum_{\mu} \delta_{\alpha \delta} \rho(E_{\gamma \mu})\rho( E_{\mu \beta})
\label{equation_algebra_commutator_combinatorial}
\end{multline}
\label{lemma_combinatorial_one_loop}
\end{lemma}
Let us write the parameter $\hbar$ in the algebra $U_{\hbar}^{QFT}(\mf{gl}_N[z_1,z_2])$ as $\hbar_{QFT}$,  and the  corresponding parameter in the combinatorially defined algebra as $\hbar_{comb}$. What this lemma tells us is that, if we identify $\hbar_{comb} = 2^{-5} \pi^{-2} \hbar_{QFT}$, the commutators of the elements $E_{\alpha \beta} z_i$ in the two algebras can be matched, modulo $\hbar^2$.  

\begin{proof}
We have already seen in proposition \ref{proposition_combinatorial_one_loop} that
\begin{equation}
[E_{\alpha \beta} \partial, E_{\gamma \delta}  z] = \delta_{\beta \gamma} E_{\alpha \delta} (\partial z) - \delta_{\delta\alpha} E_{\gamma \beta} ( z \partial) 
 + \hbar E_{\alpha \delta}E_{\gamma\beta}.
 \label{equation_commutator_simple}
\end{equation}
We have
$$
\tfrac{1}{2} E_{\alpha \delta} (\partial z + z \partial) = \rho ( E_{\alpha \delta} z_1 z_2 ). 
$$
Let us use relation (M) which we recall:
\begin{center}
\begin{tikzpicture}[scale=0.6]
\begin{scope}
\fill[pattern=north west lines ] (0.75,0) rectangle (1,4);
\draw (1,4)--(1,0);
\draw[->-](1,3) -- (3,3) node[midway,above]{$n-1$};
\draw[->-](3,1) -- (1,1) node[midway,above]{$n$};
\end{scope}
\node at (4,2) {$=$};
\begin{scope}[shift={(5,0)}]
\fill[pattern=north west lines ] (0.75,0) rectangle (1,4);
\draw (1,4)--(1,0);
\draw[->-](3,3) -- (1,3) node[midway,above]{$n-1$};
\draw[->-](1,1) -- (3,1) node[midway,above]{$n$};
\end{scope}

\begin{scope}[shift={(9.75,0)}]

\filldraw[pattern=north west lines , even odd rule] (1,2) circle (1) (1,2) circle (1.25);
\fill[color = white] (-1,0) rectangle (1,4);
\fill[pattern=north west lines ] (0.75,0) rectangle (1,4);
\draw (1,3) -- (1,1);
\draw (1,4) -- (1,3.25);
\draw (1,0) -- (1,0.75);
\node at (0,2) {$\hbar$}; 
\end{scope}
\node at (9,2) {$-$};
\node at (13.5,2) {$-$}; 

\begin{scope}[shift={(16,0)}]
\node at (-0.7,2) {$\hbar \displaystyle \sum_{\mu = 1}^{N} $} ;  
\fill[pattern=north west lines ] (0.75,0) rectangle (1,4);
\draw (1,4)--(1,0);
\node(N1) at (3,3) {$\mu$};
\node(N2) at (3,1) {$\mu$};
\draw[dbl->-](N1) -- (1,3)  node[midway,above]{$n-1$}; 
\draw[dbl->-] (1,1) -- (N2) node[midway,above]{$n$};
\end{scope}
\node at (20,2) {$+$} ; 
\begin{scope}[shift={(21,0)}]
\node at (0,2) {$c$};
\fill[pattern=north west lines ] (0.75,0) rectangle (1,4);
\draw (1,4)--(1,0);
\end{scope}
\end{tikzpicture}
\end{center}
We want to apply this to find a linear relation among $E_{\alpha \beta} \partial z$, $E_{\alpha \beta} z \partial$, and $\sum E_{\alpha \mu} E_{\mu \alpha}$.   (In our conventions, $\partial$ corresponds to an outgoing arrow and $z$ to an incoming arrow).  

To derive such a linear relation, we need to get rid of the second term on the right hand side of the pictorial expression above. This involves a surface  with a boundary component with no marked points.  Such  surfaces are zero modulo $\hbar$. Since this term appears with a coefficient of $\hbar$, it does not contribute modulo $\hbar^2$. 

Using relation $(M)$, we therefore find 
$$
E_{\alpha \delta} \partial z =  E_{\alpha \delta} z \partial  + c E_{\alpha \delta} - \hbar \sum E_{\alpha \mu} E_{\mu \delta}. 
$$
We can use this to derive the equality
\begin{align*} 
  \rho ( E_{\alpha \delta} z_1 z_2 ) &=  \tfrac{1}{2} E_{\alpha \delta} (\partial z + z \partial) \\
  &= E_{\alpha \delta} \partial z - \tfrac{c}{2} E_{\alpha \delta} + \hbar \sum E_{\alpha \mu} E_{\mu\delta}. 
 \end{align*}
Similarly we have
$$
\rho(E_{\gamma \beta} z_1 z_2) = E_{\gamma \beta} z \partial + \tfrac{c}{2} E_{\gamma \beta} - \hbar \sum E_{\gamma \mu} E_{\mu \delta}.  
$$
Plugging this into equation \ref{equation_commutator_simple} gives us
\begin{align*}
[E_{\alpha \beta} \partial, E_{\gamma \delta}  z] =& \delta_{\beta \gamma} \rho ( E_{\alpha \delta} z_1\ast_c z_2) - \hbar\delta_{\beta \gamma} \sum\rho( E_{\alpha \mu}) \rho(E_{\mu \delta})\\
&- \delta_{\delta \alpha} \rho(E_{\gamma \beta} z_2\ast_c z_1)  -\hbar\delta_{\alpha \delta} \sum \rho(E_{\gamma \mu})\rho( E_{\mu \beta})\\
& + \hbar \rho(E_{\alpha \delta})\rho(E_{\gamma\beta}).
\end{align*}
\end{proof}

It is an interesting, and very difficult, question as to whether one can  \emph{explicitly} match commutators computed using the quantum field theory with those computed in the combinatorial  algebra, to higher  orders in $\hbar$.  I will make no attempt to address this question in this paper. Instead, we will  rely on abstract results about the uniqueness of a deformation.  

\section{Matching combinatorial and field-theoretic algebras to higher order}
Let me now state the abstract results that we will prove later.  Recall that we have two algebras $A_{N,\hbar,c}$ and $U_{\hbar}^{QFT}(\mf{gl}_N  \otimes  \C[z_1,z_2])$, where the latter algebra is a deformation of $U(\mf{gl}_N \otimes \C[z_1,z_2])$ over the ring $\C[[\hbar]]$.  Here, we have used a rescaling symmetry of the field theory to set $c = 1$. We can do the same on the algebra $A_{N,\hbar,c}$, using the fact that $A_{N,\hbar,c} = A_{N,\lambda \hbar, \lambda c}$ for any non-zero $\lambda$.

In section \ref{section_uniqueness} we prove the following theorem.  The proof is rather difficult, but is purely algebraic.   Let $\kappa \in A_{N,\hbar,1}$ be the central element corresponding to the disc with no marked points.  There us a similar central element in $U^{QFT}_{\hbar}(\mf{gl}_N \otimes \op{Diff}(\C))$, which is a lift of the element $\op{Id}_{\mf{gl}_N} \otimes 1 \in \mf{gl}_N \otimes \op{Diff}(\C)$ to the quantum level. The existence of such a lift is not obvious, but is proved as part of the uniqueness result. 
\begin{theorem}
If $\hbar$ is a formal parameter, then there is an isomorphism of algebras 
$$
A_{N,\hbar,1} \iso U^{QFT}_{\hbar} ( \mf{gl}_N \otimes \op{Diff}(\C)). 
$$
This isomorphism may not be an isomorphism of $\C[[\hbar]]$ algebras, but it takes 
\begin{equation}
\hbar \mapsto 2^{-5} \pi^{-2} \hbar + f_2(\kappa) \hbar^2 + f_3(\kappa)\hbar^3 + \dots
\label{equation_matching_parameters}
\end{equation}
where $\lambda$ is a non-zero constant, and $f_i(\kappa)$ is a polynomial in $\kappa$ of degree at most $i-1$.

This isomorphism is canonical up to conjugation by an inner isomorphism.   
\end{theorem}
In other words, the two algebras are isomorphic but we may need to choose a different set of generators of the algebra $\C[\kappa][[\hbar]]$ of central elements.

In this way, we have found an explicit  computation  of an important object of our field theory: the  algebra Koszul dual to the algebra of observables.   This gives us an implicit  understanding of the algebra of observables, by applying  Koszul duality to the combinatorially defined algebra $A_{N,\hbar,1}$.  

The proof  of this theorem has two parts. First, we perform an abstract analysis of the possible deformations of $U(\mf{gl}_N \otimes \C[z_1,z_2])$, or more precisely for the deformations of the algebras $U(\mf{gl}_{N+R \mid R} \otimes \C[z_1,z_2])$ for all $R$ in a compatible way.   We find that the space  of deformations is a free module for $\C[\kappa]$.  

Next, we check that, working modulo $\hbar^2$, the algebras $U_{\hbar}^{QFT} (\mf{gl}_N \otimes \C[z_1,z_2])$ and $A_{N,\hbar,1}$ both give a non-trivial deformation, and that in fact the deformation in each  case  provides a basis for the space  of deformations as a $\C[\kappa]$-module.     In order to verify this, we show that we can detect the cohomology class  of any first-order deformation of $U(\mf{gl}_{N+R \mid R} \otimes\C[z_1,z_2])$ by calculating the commutator of the elements $E_{\alpha \beta} z_i$.  Since  we did this computation for each of our two algebras,  we find that each  algebra gives a non-trivial cohomology class modulo $\hbar^2$.  We can further find that the cohomology classes given by the two algebras are the same up to a factor of $2^{-5} \pi^{-2}$.     

The fact that there exists an isomorphism at all orders in $\hbar$ as stated in the theorem follows from the fact that two first order the cohomology class of the deformation associated to each algebra  forms a basis for the space  of deformations as  a  $\C[\kappa]$-module.   The restrictions on the powers of $\kappa$ that appear in equation \ref{equation_matching_parameters} follow from the fact that in each algebra, when we work to order $i$ in $\hbar$, the commutator of two generators $E_{\alpha \beta} z_1^k z_2^l$ can only involve at most the $i-1$'st power of the central element $\kappa$. (In each algebra this  can  be verified by a diagrammatic analysis).
\section{Deformation quantization of a Nakajima quiver variety and the algebra $A_{N,\hbar,c}$.}
\label{section_adhm}
In this section we will construct a homomorphism from the algebra  $A_{N,\hbar,c}$ to a deformation-quantization of the algebra of holomorphic functions on the moduli of instantons on a non-commutative $\C^2$ of rank $N$ and charge $K$. We will show that this map becomes an isomorphism as $K \to \infty$. As we explained in the introduction, we will interpret this as an instance of the AdS/CFT correspondence. 

Recall that in the construction of Nakajima quiver varieties one considers a quiver with two types of nodes, which are conventionally drawn as circles and squares.  Circular nodes are ``gauged'', meaning that when we construct the variety we perform symplectic reduction with respect to a copy of $GL(K)$ associated to the node (where $K$ is the integer labelling the node).  Square nodes are not gauged, and lead to a $GL(K)$ symmetry of the quiver variety. 

The ADHM  quiver is the following quiver:
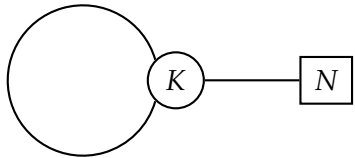
\begin{center}
\begin{tikzpicture}
  \node[circle,draw,thick](NL) at (0,0){$K$};
  \node[rectangle,draw,thick, inner sep=5](NR) at (2,0){$N$};
  \draw[thick](NL.east)--(NR.west);
  
  \draw[thick](NL.north west) arc (16:344:1) ;  

\end{tikzpicture}
\end{center}

Let 
$$
V_{N,K} = \mf{gl}(K) \oplus \op{Hom}(\C^N, \C^K).
$$
This has an obvious action of both $GL(K)$ and $GL(N)$.  The Nakajima quiver variety is 
$$
\mc{M}_{N,K} = \left( T^\ast V_{N,K}\right) \sslash GL(K) 
$$
where we perform the algebraic symplectic reduction.  The variety $\mc{M}_{N,K}$ is the moduli of framed torsion-free sheaves on $\C^2$ of rank $N$ and second Chern class $K$.   There is a deformation of this where we set the moment map to be a multiple $c$ of the identity, giving a variety $\mc{M}^{c}_{N,K}$.  This is the moduli of framed torsion-free sheaves on a non-commutative deformation of $\C^2$, again of rank $N$ and second Chern class $K$.  We will construct a deformation quantization $\Oo_{\hbar}(\mc{M}^{c}_{N,K})$ of the algebra $\Oo(\mc{M}_{N,K}^{c})$ of polynomial functions on $\mc{M}^{c}_{N,K}$, by quantum hamiltonian reduction. 

Let us first introduce some notation to describe differential operators on $V_{N,K}$. We can choose a basis $\phi_{ij}$ (for $i,j = 1,\dots, K$) of linear functions on $\mf{gl}(K)$, and $\rho_{i\alpha}$ (for $i = 1,\dots K$ and $\alpha = 1,\dots,N$) of linear functions on $\op{Hom}(\C^N,\C^K)$.  The algebra of differential operators on $V_{N,K}$ is generated by $\phi_{ij}$, $\dpa{\phi_{ij}}$, $\rho_{i\alpha}$ and $\dpa{\rho_{i\alpha}}$. We will introduce a parameter $\hbar$, which will be the quantum-mechanical Planck's constant, so that the commutation relations are 
\begin{align*} 
\left[\dpa{\phi_{ij}}, \phi_{mn} \right] &= \hbar \delta_{im} \delta_{jn} \\
\left[\dpa{\rho_{i\alpha}}, \rho_{j\beta} \right] &= \hbar \delta_{ij} \delta_{\alpha \beta}.
\end{align*}
The algebra defined by these generators with these commutation relations is $\op{Diff}_{\hbar}(V_{N,K})$, the algebra of differential operators on $V_{N,K}$.
 
The quantization of the moment map is the map 
\begin{align*} 
 \mu : \mf{gl}(K) \to & \op{Diff}_{\hbar} (V_{N,K}) \\
   \mu(E_{ij})   
=&  \sum_{k}   \dpa{\phi_{jk}}   \phi_{ik}-  \sum_{k}  \dpa{\phi_{ki}}\phi_{kj} + \sum_{\alpha}   \dpa{\rho_{j\alpha}} \rho_{i\alpha} 
\end{align*}
 In general, there is some arbitrariness in how we lift the classical moment map to the quantum moment map: we could always add a multiple of the identity. Different choices will be related by a change of coordinates of the parameters $c$ and $\hbar$.

We will define the quantum Hamiltonian reduction of the algebra $\op{Diff}_{\hbar}(V_{N,K})$ by $GL(K)$. We can include a parameter in this defintion, which is the quantum version of forming the symplectic reduction where we set the moment map to be not zero but a multiple of the identity in $\mf{gl}(K)$.  The quantum Hamiltonian reduction, with moment map set to be $c \op{Id}$,  of the algebra $\op{Diff}_{\hbar}(V_{N,K})$ is defined in two steps. First, we form the left ideal in $\op{Diff}_{\hbar}(V_{N,K})$ generated by 
$$
\mu(E_{ij})=\delta_{ij}c.
$$
or more explicitly
\begin{equation}
\sum_{k}     \phi_{ik}  \dpa{\phi_{jk}} -  \sum_{k}\phi_{kj}   \dpa{\phi_{ki}} + \sum_{\alpha}   \dpa{\rho_{j\alpha}} \rho_{i\alpha} = c \delta_{ij}.   \label{equation_relation} 
\end{equation}
Call this left ideal $I_{c}$.  Then, the quantum Hamiltonian reduction is
$$
\Oo_{\hbar}( \mc{M}_{N,K}^{c}) = \left( \op{Diff}(V_{N,K}) / I_{c} \right)^{GL(K)}.  
$$

Note that there is an algebra map
$$
\op{Diff}(V_{N,K})^{GL(K)} \to \Oo_{\hbar}(\mc{M}_{N,K}^{c}).  
$$
Let $D_{ij}$ be a collection of elements of $\op{Diff}(V_{N,K})$ which transform under $GL(K)$ in the adjoint representation. Then, this map has a kernel which includes elements of the form
\begin{equation} 
 D_{ij}\mu ( E_{ji})  -c D_{ij} \delta_{ji} 
\end{equation}
where $E_{ij} \in \mf{gl}(K)$ is the elementary matrix.

For all non-zero $\lambda$, there is an isomorphism
$$
\Oo_{\hbar}(\mc{M}_{N,K}^{c}) \iso \Oo_{\lambda \hbar}( \mc{M}_{N,K}^{\lambda c}) 
$$
obtained by multiplying all the generators $\phi_{ij}$ and $\rho_{i\alpha}$ by $\lambda$. 

\subsection{}
As an example, let's see what this algebra look like when $K = 1$. We will invert $\hbar$ in what follows.  In what follows since there is only one $GL(K)$ index we will surpress it. The algebra is generated  by the elements
$$
E_{\alpha \beta} = \frac{1}{\hbar} \rho_{\alpha} \dpa{ \rho_{ \beta} } 
$$
which satisfy the usual commutation relations of $\mf{gl}(N)$,
$$
[E_{\alpha \beta}, E_{\gamma \delta}] = \delta_{\beta \gamma} E_{\alpha \delta}- \delta_{\delta \alpha} E_{\gamma \beta}. 
$$
The moment map relation tells us that 
$$
\sum  \dpa{\rho_\alpha} \rho_\alpha= c \hbar^{-1},    
$$
so that 
$$
\sum  E_{\alpha \alpha} + N  =     c \hbar^{-1}.  
$$
In addition, there is a  relation
$$
 E_{\alpha \beta} E_{\gamma \delta} = E_{\alpha \delta} E_{\gamma\beta} + \delta_{\beta \gamma} E_{\alpha \delta} - \delta_{\gamma \delta} E_{\alpha \beta}.  
$$
These relations tell us that the algebra generated by the $E_{\alpha \beta}$ is a ring of twisted differential operators on $\mbb{P}^{n-1}$.

Further, the operators $\phi, \dpa{\phi}$ are completely independent of the operators $\rho_\alpha, \partial_{\rho_\beta}$, and do not satisfy any moment map constraint.  The operators $\phi, \dpa{\phi}$ thus form a copy of $\op{Diff}(\C)$, the algebra of differential operators on $\C$.  It follows that we find a a tensor product of differential operators on $\C$ with twisted differential operators on $\mbb{P}^{n-1}$.  The particular twist depends on $c \hbar^{-1}$.

Let us analyze what twist we find in the simplest case, $N=2$.   In this case, the algebra generated by the $E_{\alpha \beta}$ forms a quotient of the universal enveloping algebra of $\mf{gl}_2$ by the relations
\begin{align*} 
 E_{12} E_{21} &= E_{11} E_{22} + E_{11}\\ 
E_{11} + E_{22} +  2 &= \tfrac{c}{\hbar}.  
\end{align*}
We can use the second relation to remove the identity matrix from $\mf{gl}_2$, and view the algebra as a quotient of the universal enveloping algebra of $\mf{sl}_2$. If we do this, the first relation becomes 
$$
E_{12}E_{21} = - \tfrac{1}{4} (E_{11} - E_{22} + \tfrac{c}{\hbar}  - 2  ) (E_{11} - E_{22} + 2 - \tfrac{c}{\hbar} )  + \tfrac{1}{2}( E_{11} - E_{22}  + \tfrac{c}{\hbar} - 2).  
$$ 
Using the standard notation $E = E_{12}$, $F = E_{21}$, $H = E_{11} - E_{22}$, we find the relation becomes 
$$
E F + \tfrac{1}{4} H H - \tfrac{1}{2} H = \tfrac{1}{4} (\tfrac{c}{\hbar}-2)^2 + \tfrac{1}{2}(\tfrac{c}{\hbar}-2) 
$$
or equivalently,
\begin{equation}
\tfrac{1}{2} (E F + F E) + \tfrac{1}{4} H H =   \tfrac{1}{4} (\tfrac{c}{\hbar}-2)^2 +  \tfrac{1}{2} (\tfrac{c}{\hbar}-2) .
 \label{equation_relation_sl2} 
\end{equation}
Thus, we find the central quotient of the universal enveloping algebra where the quadratic Casimir is set to $\tfrac{1}{4} (\tfrac{c}{\hbar}-2)^2 + \tfrac{1}{2} (\tfrac{c}{\hbar}-2)$.  
\subsection{}
Consider the combinatorially-defined algebra $A_{N,\hbar,c}$ defined in section \ref{section_combinatorial_algebra}. 
\begin{proposition}
Whenever $\hbar \neq 0$, there is an algebra homomorphism
$$
A_{N,\hbar,c} \to \Oo_{\hbar } (\mc{M}_{N,K}^{c}).  
$$
\label{proposition_homomorphism_adhm}
\end{proposition}

\begin{proof}

Recall that we define the larger algebra $\til{A}$ to be generated by surfaces with markings $\Sigma(p_1,\dots,p_n)$.  We will explain how every such surface $\Sigma(p_1,\dots,p_n)$ gives rise to a $GL(K)$-invariant differential operator on $V_{N,K}$. Then we will check that this map descends to the quotient $A_{N,\hbar,c}$ of $\til{A}$.

Given a marked surface $\Sigma(p_1,\dots,p_n)$, we define a \emph{$K$-marking} of $\Sigma(p_1,\dots,p_n)$ to be the following data.  Consider any pair $p_r,p_s$ of marked points which are on the same boundary component, such that
\begin{enumerate} 
 \item $p_s$ is directly to the right of $p_r$ on the same boundary component of $\Sigma$.
\item If both $p_r$ and $p_s$ are of type $1$, then $p_r$ is outgoing and $p_s$ is incoming.  
\end{enumerate} Then, for any such pair, we label the interval on the boundary of $\Sigma$ between $p_r$ and $p_s$ by an element of the set $\{1,\dots,K\}$.  Further, for every boundary component of $\Sigma$ which has no marked points, we label this boundary component by an element of $\{1,\dots,K\}$. 

In other words, a $K$-marking on $\Sigma$ consists of a labelling of each connected component of the complement of the marked points on $\partial \Sigma$ by an element of $\{1,\dots,K\}$, except that we don't label the intervals between two marked points which are both of type $1$, with an incoming marked point to the left of an outgoing one.  

Given a $K$-marking on $\Sigma(p_1,\dots,p_n)$, then for any marked point $p_r$ we let $l(p_r)$ (repectively, $r(p_r)$) be the element of $\{1,\dots,K\}$ which is the labels on the interval to the left (respectively, the right) of $p_r$.  Note that if $p_r$ is of type $1$ and incoming, then only $l(p_r)$ is defined, and if it is type $1$ and outgoing, only $r(p_r)$ is defined. 

Given a surface $\Sigma(p_1,\dots,p_n)$ with a $K$-marking, we will construct a differential operator on $V_{N,K}$ for each marked point $p_s$.  We let
$$
D(p_s) = \begin{cases} \dpa{\phi_{l(p_s) r(p_s)} } & \text{ if } p_r \text{ is of type } 2 \text{ and incoming.}\\
 \phi_{r(p_s) l(p_s)}  & \text{ if } p_s \text{ is of type } 2 \text{ and outgoing.}\\
\frac{1}{\hbar} \dpa{\rho_{l(p_s) \alpha}} & \text{ if } p_s \text{ is of type } 1 \text{ and incoming, and labelled by } \alpha \in \{1,\dots,N\}. \\
 \rho_{r(p_s)\alpha } & \text{ if } p_s \text{ is of type } 1 \text{ and outgoing,  and labelled by } \alpha \in \{1,\dots,N\}. \\
\end{cases}
$$
Then, given a $K$-marking which we denote by $\Gamma$ on $\Sigma(p_1,\dots,p_n)$ we define
$$
D(\Sigma(p_1,\dots,p_n), \Gamma) = D(p_1) D(p_2) \dots D(p_n).
$$
Finally, we define
$$
D(\Sigma(p_1,\dots,p_n) ) = \sum_{K \text{ markings } \Gamma}D(\Sigma(p_1,\dots,p_n), \Gamma). 
$$
Note that this is a $GL(K)$-invariant differential operator on $V_{N,K}$.

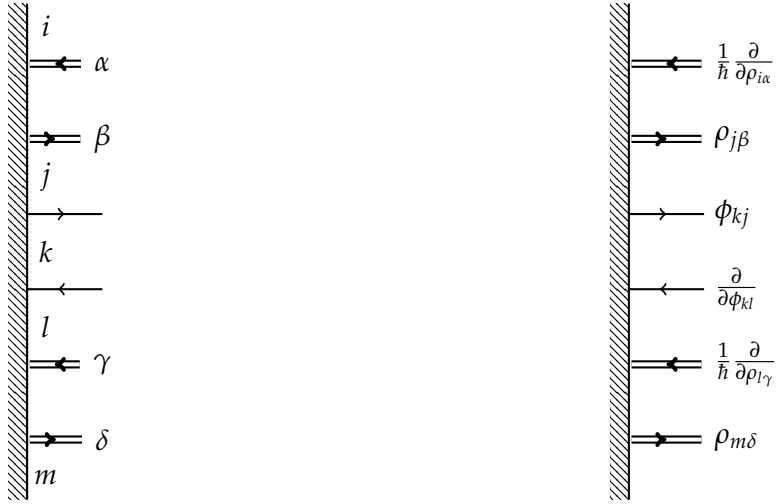
\begin{figure}
\begin{tikzpicture}[thick]
\begin{scope}[shift={(-2,0)}] 
\draw(3,-0.3)--(3,6.3);
\fill[pattern = north west lines ](2.75,-0.3) rectangle (3,6.3);
\node(N1) at (4,5.5) {$\alpha$}; 
\node(N2) at (4,4.5) {$\beta$};
\node(L0) at (3.25,6) {$i$}; 
\draw[dbl-<-](3,5.5)--(N1);
\draw[dbl->-](3,4.5)--(N2);

\node(L1) at (3.25,4) {$j$}; 
\draw[->-](3,3.5)--(4,3.5);
\node(L2) at (3.25,3){$k$};
\draw[-<-](3,2.5)--(4,2.5);
\node(L3) at (3.25,2) {$l$};
%\draw[->-](3,1.5)--(4,2.5);
 
\node(N3) at (4,1.5) {$\gamma$};
\node(N4) at (4,0.5) {$\delta$};
\draw[dbl-<-] (3,1.5)--(N3);
\draw[dbl->-](3,0.5) -- (N4);
\node(L4) at (3.25, 0) {$m$}; 
\end{scope}

\begin{scope}[shift={(6,0)}] 
\draw(3,-0.3)--(3,6.3);
\fill[pattern = north west lines ](2.75,-0.3) rectangle (3,6.3);
\node[anchor=west](N1) at (4,5.5) {$\tfrac{1}{\hbar} \dpa{\rho_{i\alpha}}$}; 
\node[anchor=west](N2)  at (4,4.5)  {$\rho_{j \beta}$};
%\node(L0) at (3.25,6) {$i$}; 
\draw[dbl-<-](3,5.5)--(N1); 
\draw[dbl->-](3,4.5)--(N2);

\node[anchor=west](N3) at (4,3.5) {$\phi_{kj}$};
%\node(L1) at (3.25,4) {$j$}; 
\draw[->-](3,3.5)--(N3);

%\node(L2) at (3.25,3){$k$};
\node[anchor=west](N4) at (4,2.5) {$\dpa{\phi_{kl}} $};
\draw[-<-](3,2.5)--(N4);
%\node(L3) at (3.25,2) {$l$};
%\draw[->-](3,1.5)--(4,2.5);
 
\node[anchor=west](N5) at (4,1.5) {$\tfrac{1}{\hbar}\dpa{\rho_{l \gamma}}$};
\node[anchor=west](N6) at (4,0.5) {$\rho_{m\delta}$};
\draw[dbl-<-] (3,1.5)--(N5);
\draw[dbl->-](3,0.5) -- (N6);
%\node(L4) at (3.25, 0) {$m$}; 
\end{scope} 

\end{tikzpicture}
\caption{The picture on the left is of a region of the boundary of a marked surface equipped with a $K$-marking.  The labels $i$, $j$, $k$, $l$, $m$ on boundary segments are elements of the set $\{1,\dots,K\}$ and give the $K$-marking. On the right, the same surface is drawn where each arrow is labelled by the corresponding differential operator.  By taking the product of these differential operators using the order on the set of arrows we get an element of $\op{Diff}(V_{N,K})$.}  
\end{figure}
Let us check that the relations defining the quotient $A_{N,\hbar,c}$ hold.  The relations (P0), (P1), (P2), (P3) involving switching the ordering of the marked points $p_i$ and $p_{i+1}$ follow immediately from the commutation relations in the algebra $\op{Diff}(V_{N,K})$.  The cutting relations (C0) and (C1) follow immediately from the fact that the differential operator associated to $\Sigma$ only depends on the way the marked points are arranged on the various boundary components of $\Sigma$. In this way we see that  all relations except relation the moving relation (M) hold in the algebra $\op{Diff}(V_{N,K})^{GL(K)}$.  

Only for relation (M) do we need to descend to the quantum Hamiltonian reduction $\Oo_{\hbar}(\mc{M}_{N,K}^{c})$.  Let's check that relation (M) matches what we find from quantum Hamiltonian reduction.  To see this, let's write relation (M) including the $K$-markings.
\begin{center}
\begin{tikzpicture}[scale=0.6]
\begin{scope}
\fill[pattern=north west lines ] (0.75,0) rectangle (1,4);
\node at (-0.5,2) {$ \displaystyle \sum_{k = 1}^{K} $} ;  
\draw (1,4)--(1,0);
\node at (1.35,3.5) {$i$};
\draw[->-](1,3) -- (3,3);% node[midway,above]{$n-1$};
\node at (1.35,2) {$k$};
\draw[->-](3,1) -- (1,1);% node[midway,above]{$n$};
\node at (1.35,0.5) {$j$};
\end{scope}
\node at (4,2) {$=$};
\begin{scope}[shift={(5.5,0)}]
\node at (-0.5,2) {$ \displaystyle \sum_{k = 1}^{K} $} ;  
\fill[pattern=north west lines ] (0.75,0) rectangle (1,4);
\draw (1,4)--(1,0);
\node at (1.35,3.5) {$i$};
\draw[->-](3,3) -- (1,3);% node[midway,above]{$n-1$};
\node at (1.35,2) {$k$};
\draw[->-](1,1) -- (3,1);% node[midway,above]{$n$};
\node at (1.35,0.5) {$j$};
\end{scope}

\begin{scope}[shift={(10.5,0)}]

\filldraw[pattern=north west lines , even odd rule] (1,2) circle (1) (1,2) circle (1.25);
\fill[color = white] (-1,0) rectangle (1,4);
\fill[pattern=north west lines ] (0.75,0) rectangle (1,4);
\draw (1,3) -- (1,1);
\draw (1,4) -- (1,3.25);
\draw (1,0) -- (1,0.75);
\node at (-0.5,2) {$\hbar\delta_{ij}$}; 
\end{scope}
\node at (9,2) {$-$};
\node at (13.5,2) {$-$}; 

\begin{scope}[shift={(16,0)}]
\node at (-0.5,2) {$\hbar \displaystyle \sum_{\mu = 1}^{N} $} ;  
\fill[pattern=north west lines ] (0.75,0) rectangle (1,4);
\draw (1,4)--(1,0);
\node at (1.35,3.5) {$i$};
\node(N1) at (3,3) {$\mu$};
\node(N2) at (3,1) {$\mu$};
\draw[dbl->-](N1) -- (1,3) ;% node[midway,above]{$n-1$}; 
\draw[dbl->-] (1,1) -- (N2);% node[midway,above]{$n$};
\node at (1.35,0.5) {$j$};
\end{scope}
\node at (20,2) {$+$} ; 
\begin{scope}[shift={(21,0)}]

\node at (0,2) {$ \delta_{ij}c$};
\fill[pattern=north west lines ] (0.75,0) rectangle (1,4);
\draw (1,4)--(1,0);
\end{scope}
\end{tikzpicture}
\end{center}
Let's translate this picture into algebra using our rules.  In the above diagram, the convention is that marked points are ordered from top to bottom.  Thus the equation we find is
$$
\sum_k \phi_{ki} \dpa{\phi_{kj}} = \sum_k\dpa{\phi_{ik}}  \phi_{jk}  
- \hbar \delta_{ij}K  -\hbar \sum_{k,\mu}\tfrac{1}\hbar \dpa{\rho_{i\mu}} \rho_{j\mu} + \delta_{ij} c.
$$
We can rewrite this as
$$
\sum_k \phi_{ki} \dpa{\phi_{kj}} -  \sum_k \phi_{jk}  \dpa{\phi_{ik}}  + \sum_{k,\mu} \dpa{\rho_{i\mu}} \rho_{j\mu} = \delta_{ij} c.
$$
This is the same as the moment map equation \ref{equation_relation}.  
\end{proof} 

\subsection{Modules}
One nice corollary of this result is that we can obtain a family of modules for the algebra $A_{N,\hbar,c}$, by constructing modules for the quantizations $\Oo_\hbar(\mc{M}_{N,K}^c)$ of quantized functions on instanton moduli space, and then using the homomorphism
$$
A_{N,\hbar,c} \to \Oo_\hbar(\mc{M}_{N,K}^{c}). 
$$
A particularly nice class of modules are obtained when we set $K = 1$, so that $\mc{M}_{N,K}^{c}$ is the product of $\C^2$ with a twisted cotangent bundle of $\mbb{P}^{N-1}$.  The quantized algebra $\Oo_\hbar(\mc{M}_{N,K}^{c})$ is then a tensor product of $\op{Diff}(\C)$ with an algebra of twisted differential operators on $\mbb{P}^{N-1}$.  One thus gets a module for $A_{N,\hbar,c}$ from the tensor product of any $D$-module on $\C$ with a twisted $D$-module on $\mbb{P}^{N-1}$.

For generic values of $c/\hbar$, we will find a non-integral twist of differential operators on $\mbb{P}^{N-1}$, so there will not be any finite-dimensional modules for this algebra.  However, there are always infinite-dimensional modules: the algebra $\Oo(\C^{N-1})$ of polynomial functions on the affine subspace $\C^{N-1} \subset \mbb{P}^{N-1}$ is always a module for TDOs on $\mbb{P}^{N-1}$. This is simply because there is a restriction map from TDOs on $\mbb{P}^{N-1}$ to those on $\C^{N-1}$, and twisted differential operators on $\C^{N-1}$ are isomorphic to ordinary differential operators.

This tells us that we can find a module for $A_{N,\hbar,c}$ for any value of $\hbar$ and $c$ (with $\hbar \neq 0$)  which is a tensor product of $\Oo(\C^{N-1})$ with an arbitrary $D$-module on $\C$. 

For some special values of $c/\hbar$, we will find smaller modules.  Let us analyze this in the case $N=2$.  As we have seen, the algebra $\Oo_\hbar(\mc{M}^c_{2,1})$ is a tensor product of $\op{Diff}(\C)$ with the  central quotient of the universal enveloping algebra of $\mf{sl}_2$ in which the quadratic Casimir $\tfrac{1}{2} (E F + F E) + \tfrac{1}{4} H H$ is identified with $\tfrac{1}{4} (\tfrac{c}{\hbar}-2)^2 + \tfrac{1}{2} (\tfrac{c}{\hbar}-2)$.  On the finite-dimensional representation of highest weight spin $n$ the quadratic Casimir takes value $\frac{1}{4} n^2  - \tfrac{1}{2} n$. 

This tells us that as long as $2 - \tfrac{c}{\hbar}$ is a non-negative integer $n$, there is a representation of $\Oo_\hbar(\mc{M}^c_{2,1})$ which is a tensor product of a $D$-module on $\C$ (such as $\Oo_\C$) with a vector space of dimension $n$. 

Because there is a homomorphism
$$
A_{2,\hbar,c} \to \Oo_\hbar(\mc{M}^c_{2,1}),
$$
the same statement holds for the algebra $A_{2,\hbar,c}$.

\subsection{Taking the large $K$ limit}
In this subsection, we will do the following. First, we will define a uniform-in-$K$ version of the algebra $\Oo_\hbar(\mc{M}^c_{N,K}$.  This uniform-in-$K$ version of the algebra has the feature that it has a surjective algebra homomorphism to $\Oo_\hbar(\mc{M}^c_{N,K})$ for each $K$.  These maps become close to being an isomorphism as $K \to \infty$. We will then show that the uniform-in-$K$ version is isomorphic to the combinatorially defined algebra $A_{N,\hbar,c}$.  

Let us start by defining the uniform-in-$K$ version of the algebra of $GL(K)$-invariant polynomials on $T^\ast V_{N,K}$.  To do this, let us fix identifications $\C^{K \oplus K'} = \C^K \oplus \C^{K'}$.  This gives us a map $\mf{gl}_K \into \mf{gl}_{K + K'}$ by viewing $\mf{gl}_K$ as the block diagonal matrices in the upper left quadrant.  We similarly get a map $\mf{gl}_{K + K'} \to \mf{gl}_K$ by taking a $(K + K') \times (K + K')$ matrix and considering the $K \times K$ submatrix in the upper left quadrant. 

Recall that 
$$
V_{N,K} = \mf{gl}_K \oplus \op{Hom}(\C^N, \C^K).  
$$
We can use these maps to construct maps
 \begin{align*} 
V_{N,K} & \to V_{N,K+ K' } \\
V_{N,K+K'} & \to V_{N,K}
\end{align*}
and so a map
$$
V^\ast_{N,K} \oplus V_{N,K} \to V^\ast_{N,K+K'} \oplus V_{N,K+K'}.
$$
This map is $GL(K)$-invariant.  If $\Oo(V_{N,K})$ indicates the algebra of polynomial functions on $V_{N,K}$, we get an induced map 
$$
\iota^{K+K'}_K : \Oo ( T^\ast V_{N,K+K'} ) \to \Oo( T^\ast V_{N,K}). 
$$
\begin{definition} 
 A sequence $f_K \in \Oo(T^\ast V_{N,K})$ of polynomial functions on $T^\ast V_{N,K}$ is \emph{admissible of weight} $0$ if 
 \begin{enumerate} 
  \item Each $f_K$ is $GL(K)$ invariant.
  \item $\iota^{L}_K f_L = f_K$ for all $L > K$.
 \end{enumerate}
 Such a sequence is \emph{admissible of weight} $r$ if the sequence $K^{-r} f_K$ is admissible of weight $0$. A sequence $\{f_K\}$ is \emph{admissible} if it is a finite sum of sequences which are admissible of some weight. Note that if $\{f_K\}$ and $\{g_K\}$ are admissible sequences, then so is $\{f_K g_K\}$, so that the space of admissible sequences is an algebra. 

 We let $\Oo(T^\ast V_{N,\bullet})^{GL(\bullet)}$ be the algebra of admissible sequences of functions on $T^\ast V_{N,K}$. 
\end{definition}

Let us now define a non-commutative product on $\Oo(T^\ast V_{N,\bullet})^{GL(\bullet)}$. Recall that for any vector space $W$, we can define a Moyal product $\ast_\hbar$ on $\Oo(W^\ast \oplus W)$ for any complex value of the constant $c$, by a standard formula.   This is invariant under $GL(W)$.  In particular, we have such a Moyal product  on $\Oo( T^\ast V_{N,K})$ which is $GL(K)$-invariant.   We will refer to $\Oo(T^\ast V_{N,K})$ with this product as $\Oo_{\hbar} (T^\ast V_{N,K})$. 
\begin{lemma}
If $\{f_K\}$, $\{g_K\}$ are admissible sequences of functions on $T^\ast V_{N,K}$, then so is
$$
f_K \ast_{\hbar} g_K. 
$$
\end{lemma}
\begin{proof}
This follows from an analysis of the diagrammatic expression for the Moyal product. 
\end{proof}
In this way,  we find a non-commutative Moyal product on $\Oo(T^\ast V_{N,\bullet})^{GL(\bullet)}$.  We will refer to the non-commutative algebra as $\Oo_{\hbar} (T^\ast V_{N,\bullet})^{GL(\bullet)}$.   This defines the uniform-in-$K$ version of the algebra of $GL(K)$-invariant differential operators on $V_{N,K}$.  To get the uniform-in-$K$ version of the quantum Hamiltonian reduction, we need to quotient by the  ideal defined by the quantum moment map.  As discussed above, the quantum moment map  is a map
$$
\mu : \mf{gl}(K) \to \Oo_{\hbar}(T^\ast V_{N,K}).
$$
We are interested in a shift of the quantum moment map, because we want to consider the quantum version of setting the moment map to be a multiple of the identity.  We let 
$$
\mu_{c}(E) = \mu(E) + c \op{Tr} (E) \op{Id}.
$$
The quantum moment map, together with the Moyal product, gives a $GL(K)$-equivariant map of left $\Oo_{\hbar}(T^\ast V_{N,K})$-modules
$$
\mu_{c} : \Oo_{\hbar}(T^\ast V_{N,K}) \otimes \mf{gl}_K \to \Oo_{\hbar}(T^\ast V_{N,K})
$$
Passing to $GL(K)$ invariants, we get a map
$$
\mu_{c} : \left[  \Oo_{\hbar}(T^\ast V_{N,K}) \otimes \mf{gl}_K \right]^{GL(K)} \to \left[  \Oo_{\hbar}(T^\ast V_{N,K})\right]^{GL(K)}.  
$$
The image of this map is a two-sided ideal where the right hand side is equipped with the Moyal product. The quantum Hamiltonian reduction is defined to be the quotient of the right hand side by this idea.

Now it is clear how to perform quantum Hamiltonian reduction uniformly in $K$.  We can define the notion of an admissible sequence of elements of $\Oo(T^\ast V_{N,K}) \otimes \mf{gl}_K$ in the same way as our definition of admissible sequences of elements of $\Oo(T^\ast V_{N,K})$. We call this space $\left[\Oo_{\hbar}(T^\ast V_{N,\bullet})\otimes \mf{gl}_{\bullet}\right]^{GL(\bullet)}$.

Then, the quantum moment map at finite $K$ gives rise to a map
$$
\mu_{c} : \left[\Oo_{\hbar}(T^\ast V_{N,\bullet}) \otimes \mf{gl}_{\bullet}\right]^{GL(\bullet)}\to \Oo_{\hbar}(T^\ast V_{N,\bullet})^{GL(\bullet)}.
$$
The image of this map is a two-sided ideal, since it is at finite $K$.  We define the uniform-in-$K$ quantum Hamiltonian reduction to be the quotient by this ideal, and use the notation 
$$
\Oo_{\hbar}( \mc{M}^c_{N,\bullet}) 
$$
recalling that $\mc{M}^c_{N,K}$ is the symplectic reduction of $T^\ast V_{N,K}$ by $GL(K)$ with moment map set to $c$ times the identity. 
\begin{remark}
One could also define this large-$K$ version of the algebra using the notion of Deligne category, see \cite{Eti14}.
\end{remark}

\begin{proposition}
The maps we have constructed
$$
\phi_{K} : A_{N,\hbar,c} \to  \Oo_{\hbar}(\mc{M}^c_{N,K}) 
$$
lift to a map 
$$
\phi_{\bullet} : A_{N,\hbar,c} \to \Oo_{\hbar}(\mc{M}^c_{N,\bullet}). 
$$
This second map is an isomorphism of associative algebras. 
\end{proposition}
\begin{proof}
We will prove this by introducing a larger algebra $A'_{N,\hbar,c}$ which maps to $\Oo_{\hbar}( T^\ast V_{N,\bullet})^{GL(\bullet)}$, and then seeing that this map descends to give the map we need.

Let $A'_{N,\hbar,c}$ be defined in the same way as we defined $A_{N,\hbar,c}$, as a quotient of a large combinatorial algebra $\til{A}$ built from surfaces; except that we only impose the relations (P0), (P1), (P2), (P3),(C0) and (C1) in the definition of $A'_{N,\hbar,c}$. We do not impose the relation (M), which is related to the quantum moment map. 

 In the proof of \ref{proposition_homomorphism_adhm} we constructed a map $\til{A} \to \Oo_{\hbar}(T^\ast V_{N,K})^{GL(K)}$ for all $K$.  As we saw in the proof, the kernel of this map contains all the relations defining $A'_{N,\hbar,c}$ so it descends to a map 
$$\phi'_{K} : A'_{N,\hbar,c}\to \Oo_{\hbar}(T^\ast V_{N,K}) ^{GL(K)}.$$
It is evident that for all $\alpha\in A'_{N,\hbar, c}$ the sequence $\phi'_{K}(\alpha)$ is an admissible sequence.  Thus, $\phi'_K$ lifts to a map 
$$
\phi'_{\bullet} : A'_{N,\hbar,c}\to \Oo_{\hbar}(T^\ast V_{N,\bullet}) ^{GL(\bullet)}.
$$ 
As we saw in the proof of \ref{proposition_homomorphism_adhm},  the map $\phi'_K$ sends the ideal in $A'_{N,\hbar,c}$ generated by relation (M) to the ideal in $\Oo_{\hbar}(T^\ast V_{N,K})^{GL(K)}$ generated by the moment map. It follows that $\phi'_{\bullet}$ does the same, so that it gives the desired map
$$
\phi_{\bullet} : A_{N,\hbar,c} \to \Oo_{\hbar}(\mc{M}^c_{N,\bullet}). 
$$

Next, we need to check that this map is an isomorphism.  We will first check that the map
$$
A'_{N,\hbar,c} \to  \Oo_{\hbar}(T^\ast V_{N,\bullet}) ^{GL(\bullet)}.
$$
is an isomorphism.  This is an exercise in invariant theory. Then,  one checks that the two-sided ideal in $\Oo_{\hbar}(T^\ast V_{N,\bullet}) ^{GL(\bullet)}$ whose quotient is $\Oo_{\hbar}(\mc{M}^c_{N,\bullet})$ matches the two-sided ideal in $A'_{N,\hbar,c}$ generated by the relation (M).     

\end{proof}
In this way, we can think of our combinatorially-defined algebra $A_{N,\hbar,c}$ as being something like the large $K$ limit of the algebras $\Oo_{\hbar}(\mc{M}^c_{N,K})$.

\section { Categorified Donaldson-Thomas invariants and AdS / CFT. }
\label{section_DT}
In this section,  I derive a conjecture stating that the algebra $U_{\delta}(\Oo_{\eps}(\C^2)\otimes \mf{gl}_N)$ acts on certain categorified Donaldson-Thomas invariants.  Some aspects of this story were discussed in the introduction in section \ref{subsection_DT_intro}.  Here I will focus on the version related to Pandharipande-Thomas, rather than Donaldson-Thomas, invariants. 

Let 
$$
X_N = \til{\C^2 / \Z_N}
$$
be the resolution of the $A_{N-1}$ surface singularity.  The Hilbert scheme $ X_N^{[k]}$ of $k$ points on $X_N$ is a smooth quasi-projective variety with an algebraic symplectic form.  

We will be interested in quasi-maps from a curve $C$ to $X_N^{[k]}$.  The theory of such quasi-maps is studied in detail by Okounkov in \cite{Oko15}.

The notion of quasi-map depends on the presentation of $X_N^{[k]}$ as a quiver variety.  A quiver variety is always an algebraic symplectic reduction of some complex symplectic vector space $V$ (built from the edges of the quiver) by an algebraic group $G$ (which is a product of general linear groups associated to the gauge nodes of the quiver).  To define quasi-maps to $X_N^{[k]}$ we will not need a detailed description of $V$ and $G$. 

The algebraic symplectic quotient is defined as follows.  There is a moment map $\mu : V \to \mf{g}^\vee$. Inside $\mu^{-1}(0)$ is a subvariety $\mu^{-1}(0)_{stable}$ of stable points. (The stability condition depends on the choice of a vector $\theta \in \Z^I$ where $I$ is the set of nodes of the quiver). The symplectic reduction is
$$
V \sslash_{stable} G = \mu^{-1}(0)_{stable} / G. 
$$
We assume we are in a situation where semistable implies stable, and where the $G$-action is free.

If we don't remove the unstable points, and just form the quotient $\mu^{-1}(0) / G$, we find a stack. In fact, it's better to treat this as a derived stack.  The derived stack is obtained by taking the homotopy fibre at $0$ of $\mu$ and then forming the quotient by $G$. Let us denote this derived stack by $V\sslash G$.  The Nakajima quiver vareity $V \sslash_{stable} G$ is an open substack of $V \sslash G$. 

Given a curve $C$, then a quasi-map from $C$ to the quiver variety $V \sslash_{stable} G$ is by definition a map from $C$ to $V \sslash G$ which generically on $C$ lands in $V \sslash_{stable} G$.  We can thus make sense of quasi-maps to $X_N^{[k]}$.  

Fix a point $\alpha \in X_N^{[k]}$.  Let us consider the space of quasi-maps from $\mbb{P}^1$ to $X_N^{[k]}$ which at $\infty$ go to $\alpha$, and which are of degree $d$. Call this scheme $\mc{M}_{k,d,N,\alpha}$.
\begin{claim}
\begin{enumerate} 
 \item The quasi-projective scheme $\mc{M}_{k,d,N,\alpha}$ has a symmetric perfect obstruction theory.
 \item The space $\mc{M}_{k,d,N,\alpha}$ is the underlying classical scheme of a derived scheme which has a $-1$ shifted symplectic structure.  The derived scheme is the derived space of maps from $\mbb{P}^1$ to the derived stack $V \sslash G$ (using the notation above) which at $\infty$ go to $\alpha$, are of degree $d$, and generically map to the stable locus.  
\end{enumerate}
\end{claim}
Since $X_N^{[k]}$ parametrizes rank $1$ torsion free sheaves on $X_N$ which are framed at $\infty$,  we should imagine of $\mc{M}_{k,d,N,\alpha}$ as parametrizing a class of rank $1$ torsion free sheaves on $\mbb{P}^1 \times X_N$ which are framed at $\infty$ and with some singularities at those points on $\mbb{P}^1$ where the quasi-map has singularities. 

In the case $N=1$  this is made precise in \cite{Oko15}. It shown there that the moduli space of such quasi-maps is isomorphic to the Pandharipande-Thomas moduli space on $C \times \C^2$.  This is the moduli space of sheaves $\mc{F}$ on $C \times \C^2$, with a section 
$$
s : \Oo_{C \times \C^2} \to \mc{F} 
$$
such that:
\begin{enumerate} 
\item $\mc{F}$ is one dimensional.
\item $\mc{F}$ has no zero-dimensional subsheaves.
\item The cokernel of $s$ is zero-dimensional.  
\end{enumerate} 
It is natural to expect that there is a similar description of the space of quasi-maps from $C$ to $X_N^{[k]}$, in terms of Pandharipande-Thomas moduli spaces on $C \times X_N$.    

Since $\mc{M}_{k,d,N,\alpha}$ has a symmetric perfect obstruction theory, the work of Behrend \cite{Beh09} gives a constructible function $f_{B}$ on $\mc{M}_{k,d,N,\alpha}$ with the feature that the Euler characteristic twisted by $f_{B}$ is the integral against the virtual fundamental class of $\mc{M}_{k,d,N,\alpha}$.   It is expected  \cite{KonSoi10, KieLi12, BusJoyMei13, JoySon08} that this constructible function can be lifted to a perverse sheaf we will call $\mc{P}$ on $\mc{M}_{k,d,N,\alpha}$. We can thus consider the graded vector space
$$
H^\ast (\mc{M}_{k,d,N,\alpha}, \mc{P} ) 
$$
given by the cohomology with coeffients in this perverse sheaf.   This graded vector space plays the role of the categorified Donaldson-Thomas invariants when we replace the Donaldson-Thomas moduli space of ideal sheaves by the space of quasi-maps to the Hilbert scheme of points on a surface. 

We are also interested in the equivariant version of this, with two equivariant parameters.  Let us give $X_N$ an action of two copies of $\C^\times$, as follows. One copy of $\C^\times$ we denote by $\C^\times_\eps$, and this acts via the action on $\C^2$ sending 
$$(z_1,z_2) \to (\lambda z_1, \lambda^{-1} z_2).$$
The other copy of $\C^\times$ we denote by $\C^\times_{\delta}$, and this acts via the action on $\C^2$ sending
$$
(z_1,z_2) \to (\lambda z_1, \lambda z_2).
$$
The torus $\C^\times_\eps \times \C^\times_\delta$ acts on the Hilbert schemes $X_N^{[k]}$, and $\C^\times_\delta$ scales the holomorphic symplectic form on these Hilbert schemes with weight $2$. 

The space $\mc{M}_{k,d,N,\alpha}$ of quasi-maps from $\mbb{P}^1$ to $X_{N}^{[k]}$ also acquires an action of $\C^\times_\eps \times \C^\times_{\delta}$, where $\C^\times_{\eps}$ acts just via the action  on $X_{N}^{[k]}$, but where the action of $\C^\times_{\delta}$ also involves a rotation of $\mbb{P}^1$.  This can be contrived so that the action of $\C^\times_{\eps} \times \C^\times_{\delta}$ preserves the shifted symplectic structure on the derivedspace of quasi-maps.  

Since this torus preserves the shifted symplectic structure, the perverse sheaf $\mc{P}$ is equivariant.  We can therefore take equivariant cohomology of $\mc{M}_{k,d,N,\alpha}$ with coefficients in $\mc{P}$, with equivariant parameters $\eps,\delta$.   

Recall that we have defined combinatorially a deformation $A_{N,\hbar,c}$ of $U(\op{Diff}_c(\C) \otimes \mf{gl}(N))$.  
 \begin{conjecture}
 The algebra $A_{N,\delta,\eps}$ acts on 
$$\oplus_d H^\ast_{\eps,\delta}(\mc{M}_{k,d,N,\alpha}, \mc{P})$$ for all values of $k$ and $\alpha$.  
 \end{conjecture}
Let me present one strand of evidence for this conjecture.  The symplectic algebraic variety $X_N^{[k]}$ is symplectic dual to the ADHM moduli space of instantons of charge $k$ and rank $N$ on $\C^2$.  Under symplectic duality, the algebra of monopole operators \cite{BulDimGai15,BraFinNak16} becomes a deformation quantization of the symplectic dual. The algebra of monopole poerators is supposed to act on the cohomology of the space of vortices.  Quasi-maps from $\mbb{P}^1$ to $X_N^{[k]}$ with fixed behaviour at $\infty$ is an algebro-geometric version of the moduli of vortices in $X_N^{[k]}$.  The perverse sheaf of vanishing cycles is missing from most physics discussion of the cohomology of the moduli of vortices; however, a close analysis of topological twists of $3d$ $N=4$ theories (which I  will not present here) tells us that we should expect the Hilbert space to be the cohomology of the moduli of vortices with coefficients in this sheaf.  

Putting all this together, we find that a deformation quantization of the moduli of instantons of rank $N$ and charge $k$ should act on $\oplus_d H^\ast ( \mc{M}_{k,d,N,\alpha}, \mc{P})$.   As we have seen, the algebra $A_{N,\delta,\eps}$ maps to the deformation quantization of moduli of instantons of rank $N$ and charge $k$, for all $k$.  From these symplectic duality considerations, one justifies the conjecture.

More generally, one can make the following conjecture.  Let $X_M$ denote, as above, the resolution of the $A_{M-1}$ singularity.  The $5$-dimensional quantum field theory can be defined on $X_M$.  There are more parameters in the definition, however, because $X_M$ admits more than one deformation-quantization.  There is an $M-1$ dimensional space of inequivalent deformation-quantizations, with one parameter associated to each $2$-cycle in $X_M$. Let $\Oo_{c,\lambda_1,\dots,\lambda_{M-1}}(X_M)$ denote this family of deformation-quantizations of $\Oo(X_M)$ (where $c$ is the quantization parameter).   
    
The Koszul dual of the algebra of observables of the $5$-dimensional field theory will be a deformation of $U(\mf{gl}_N \otimes \Oo_{c,\lambda_1,\dots,\lambda_{M_1}} (X_M))$.  Let us call this algebra $U_{\hbar}^{QFT} ( \mf{gl}_N \otimes \Oo_{c,\lambda_1,\dots,\lambda_{M_1}} (X_M))$.  
\begin{conjecture}
The algebra $U_{\hbar}^{QFT} ( \mf{gl}_N \otimes \Oo_{c,\lambda_1,\dots,\lambda_{M_1}} (X_M))$ act on the equivariant cohomology space of quasi-maps from from $\mbb{P}^1$ to the moduli of rank $M$ instantons on the $X_N$, the resolution of the $A_{N-1}$ singularity.  The parameters $\hbar$ and $c$ should be equivariant parameters related to the rotations of $\C \times X_N$ which preserve the Calabi-Yau structure (where $\C$ is the source curve for the quasi-map).  The parameters $\lambda_i$ should be equivariant parameters which rotate the framing of the instanton on $X_N$ at $\infty$.    
\end{conjecture}
I also expect that one can define a combinatorial version of this algebra, and prove a uniqueness theorem for the quantization of $U ( \mf{gl}_N \otimes \Oo_{c,\lambda_1,\dots,\lambda_{M_1}} (X_M))$, exactly as in this paper for the $M=0$ case.  

\section{ Flatness of the family of algebras $A_{N,\hbar,c}$} 
\label{appendix_flatness}
In section \ref{section_combinatorial_algebra} we defined, combinatorially, a family of algebras $A_{N,\hbar,c}$ which we claimed deforms the universal enveloping algebra $U(\Oo_{c}(\C^2) \otimes \mf{gl}_N)$.  In that section, we stated that the family of algebras is flat over $\C[[\hbar]]$ (when we complete the parameter $\hbar$).  In this section we prove this result. 

In fact, for later use, we need to prove a generalization.  One can modify the definition of the algebra $A_{N,\hbar,c}$ to accomodate the case when the Lie algebra $\mf{gl}(N)$ is replaced by a super Lie algebra $\mf{gl}(N\mid M)$.  To do this, one allows the type $1$ marked points to be labelled by an element of the set $\{1,\dots,N, \br{1},\dots,\br{M} \}$, where the barred numbers indicate the fermionic indices.  The relations in the algebra are exactly the same, except that we pick up some extra signs in the relations (P0), (P2), (P3). Recall that, in the definition of this combinatorial algebra, the generators are given by surfaces with decorated marked points on the boundary and a total order on the set of such marked points.  The relations (P0)- (P3) tell us how to permute the order on these marked points. When defining the super version of this algebra, we get an extra Koszul sign when we permute type $1$ marked points with fermionic labels.   Let us denote the super version of this algebra by $A_{N \mid M, \hbar,c}$.  

The algebra $A_{N,\hbar,c}$ is the uniform-in-$K$ version of the algebra $\Oo_{\hbar}(\mc{M}^{c}_{N,K})$ deforming the algebra of polynomial functions on the quiver variety $\mc{M}^{c}_{N,K}$. Recall that this is the quiver variety associated to the ADHM quiver and that $c$ indicates we take the symplectic quotient by setting the moment map to $c$ times the identity.  We can define a super-quiver variety $\mc{M}_{N \mid M, K}^{c}$ by the super analog of symplectic reduction, and define an algebra $\Oo_{\hbar}( \mc{M}_{N \mid M, K}^{c})$ by performing, as before, quantum Hamiltonian reduction.  

The first result we will prove is that the family of algebras $\Oo_{\hbar}(\mc{M}_{N \mid M,K}^{c})$ is flat over $\C[[\hbar]]$,   
\begin{lemma} 
The family of algebras $\Oo_{\hbar}(\mc{M}^{c}_{N \mid M,K})$, when $\hbar$ is treated as a formal parameter, is flat over $\C[[\hbar]]$, for all values of $c$. 
\label{lemma_flat_adhm} 
\end{lemma}
\begin{proof} 
Let  
$$
V_{N\mid M,K} = \mf{gl}(K) \oplus \op{Hom}(\C^{N \mid M},\C^K). 
$$
We will refer to an element of $V_{N\mid M,K}$ as a pair $(B,I)$ where $B \in \mf{gl}(K)$ and $I \in \op{Hom}(\C^{N \mid M}, \C^K)$. 

The super scheme $\mc{M}^{c}_{N \mid M,K}$ is the symplectic reduction of $T^\ast V_{N \mid M,K}$ by $GL(K)$.  
Let 
$$
U \subset V_{N \mid M,K}
$$
be the open subset where
\begin{enumerate} 
\item $B$ is regular semi-simple, that is diagonizable with distinct eigenvalues.
\item If $I_1,\dots,I_N, I_{\br{1}}, \dots, I_{\br{M}} \in \C^K$ denote the components of $I$, we ask that $\C^K$ is spanned by the vectors $B^i I_1$ for $i \in \Z_{\ge 0}$.  This amounts to saying that, if we expand $I_1$ as a sum of eigenvectors of $B$, the coefficient of each eigenvector is non-zero. 
\end{enumerate}
 We can use the action of $GL(K)$ to diagonalize $B$, with diagonal entries $\lambda_1,\dots,\lambda_K$ being distinct.  The remaining symmetry is the semi-direct product of copy of the maximal torus $(\C^\times)^K \subset GL(K)$ with the symmetric group $S_K$. We can use this maximal torus to scale $I_1$ so that it is the vector with $1$ in each entry.  Once we do this, the only symmetry we are left with is the symmetric group $S_K$.   

This shows that $GL(K)$ acts freely on $U$, and that  the quotient of $U$ by $GL(K)$ is isomorphic to the quotient
$$
U / GL(K) \iso \left( \left\{\lambda_1,\dots,\lambda_K \in \C^\times \mid \lambda_i \neq \lambda_j \right\} \times \C^{(N-1)K \mid MK} \right)/S_K. 
$$ 
This quotient is an affine super variety.

Note also that we can identify 
$$
(T^\ast U) \sslash GL(K) = T^\ast ( U / GL(K)).
$$
The same holds at the quantum level:
$$
\Oo_{\hbar} ( T^\ast U \sslash GL(K)) = \op{Diff}_{\hbar}( U / GL(K) )  
$$
where the left hand side is defined by quantum Hamiltonian reduction.  

An analogous statement holds when we introduce the moment-map parameter $c$. There is a line bundle over $U / GL(K)$ whose principal $\C^\times$-bundle is $U / SL(K)$. We can associate a one-parameter family of twisted differential operators $\op{Diff}^{c}_{\hbar}( U / GL(K))$ to this line bundle.  There is a natural identification between 
$$
\Oo_{\hbar}( T^\ast U \sslash^c GL(K)) \iso \op{Diff}^c_{\hbar}( U / GL(K)),
$$
where the left hand side is defined by quantum Hamiltonian reduction where we specialize the moment map to a particular value.  It turns out that this line bundle is trivial, so that $\op{Diff}^{c}_{\hbar}(U / GL(K))$ is isomorphic to $\op{Diff}_{\hbar}(U / GL(K))$. 
 
It follows from the description in terms of twisted differential operators that $\Oo_{\hbar}( T^\ast U \sslash^c GL(K))$ is a flat family over $\C[\hbar]$.

There's an injective restriction map 
$$
r : \Oo( \mc{M}_{N \mid M,K}^{c} ) \to \Oo ( T^\ast U \sslash^c GL(K) ).  
$$ 
This map quantizes to a map
\begin{equation*}
r_{\hbar} : \Oo_{\hbar}( \mc{M}_{N \mid M,K}^{c} ) \to \Oo_{\hbar} ( T^\ast U \sslash^c GL(K) ) = \op{Diff}_{\hbar} ( U / GL(K) ) \tag{$\dagger$}. 
\end{equation*}
To see this, note that the restriction map 
$$
\op{Diff}_{\hbar}(V_{N \mid M,K}) \to \op{Diff}_{\hbar}(U)
$$
is $GL(K)$-equivariant, and the diagram
$$
 \xymatrix{ \op{Diff}_{\hbar}(V_{N \mid M,K} ) \otimes \mf{gl}_K \ar[r]^{\mu_c} & \op{Diff}_{\hbar}(V_{N \mid M,K} )  \\  
\op{Diff}_{\hbar}(U) \otimes \mf{gl}_K \ar[r]^{\mu_c}  & \op{Diff}_{\hbar}(U)   } 
$$ 
commutes, where the horizontal arrows are given by the quantum moment map shifted by $c$ times the identity. Passing to $GL(K)$ invariants and then taking the cokernel of the horizontal arrows gives the desired map.  

The map $r_{\hbar}$ in the displayed equation ($\dagger$) above is injective when we specialize to $\hbar = 0$. Further, the range $\op{Diff}_{\hbar} ( U / GL(K))$ of this map is flat over $\C[\hbar]$.  It follows that $r_{\hbar}$ is injective on some Zariski open subset of $\mbb{A}^1_{\hbar}$ containing $0$.  

Finally, this implies that $\Oo_{\hbar}(\mc{M}_{N \mid M,K}^{c})$ is flat on this Zariski open subset of $\mbb{A}^1_{\hbar}$, and in particular flat in the formal neighbourhood of $0$.  

\end{proof} 

Next we will prove the main result of this section. 
\begin{theorem*} 
The family of algebras $A_{N\mid M,\hbar,c}$ is flat over $\C[[\hbar]]$.  
\end{theorem*}
\begin{proof} 

Define  a new  algebra $B_{N \mid M,\hbar,c}$, defined just like $A_{N \mid M,\hbar,c}$, except that in relation (P2) we introduce a $\hbar$ next to the glued surface; and in relation (M) we remove the $\hbar$ that appears  in the term with the type $1$ marked points.  Then, there is a map $B_{N \mid M, \hbar,c} \to A_{N\mid M, \hbar,c}$, which on the level of generators, multiplies every surface by $\hbar$ to the number of incoming type $1$ marked points.  This map preserves all the relations.  This map is an isomorphism when $\hbar \neq 0$.  

In section \ref{section_adhm} we showed that there is a map of algebras
$$
A_{N , \hbar,  c} \to \Oo_{\hbar}(\mc{M}_{N,K}^{c} ) 
$$
defined when $\hbar \neq 0$ and for all values of $K$.  This map gives an isomorphism
$$
A_{N , \hbar,  c} \iso \Oo_{\hbar}(\mc{M}_{N,\bullet}^{c} ) 
$$
where $\bullet$ indicates that we work uniformly in $K$.  

The same statement, with the same proof, holds in the super case, so that we have an isomorphism 
\begin{equation*}
A_{N\mid M , \hbar,  c} \iso \Oo_{\hbar}(\mc{M}_{N\mid M,\bullet}^{c} ) \tag{$\dagger$} 
\end{equation*}
when $\hbar \neq 0$.  

The map $B_{N\mid M, \hbar, c} \to A_{N \mid M, \hbar, c}$ induces a similar map
\begin{equation*}
B_{N \mid M, \hbar, c} \iso  \Oo_{\hbar}(\mc{M}_{N \mid M, \bullet}^{c} ) \tag{$\ddagger$} \end{equation*} 
when $\hbar \neq 0$.  One can check from the definition of the map ($\dagger$) above that the map ($\ddagger$) extends to a map even when we don't invert $\hbar$.  The proof in section \ref{section_adhm} that the map ($\dagger$) is an isomorphism shows that the map ($\ddagger$) is an isomorphism for all values of $\hbar$, that is, it defines an isomorphism of algebras over $\C[\hbar]$.

We have seen in the previous lemma that the spaces $\Oo_{\hbar}(\mc{M}_{N\mid M, K}^c)$ are flat over $\C[[\hbar]]$. We can deduce from this that $B_{N \mid M, \hbar,c}$ is flat over $\C[[\hbar]]$.  Indeed, there is no element of $B_{N \mid M,\hbar,c}$ which is in the kernel of all the  maps $B_{N \mid M , \hbar,c} \to \Oo_{\hbar}(\mc{M}^{c}_{N \mid M,K})$.  However, any element which is annihilated by some power of $\hbar$ must be in the kernel of all these maps, because $\Oo_\hbar(\mc{M}^c_{N\mid M,K})$ is flat over $\C[[\hbar]]$.   Thus, $B_{N \mid M,\hbar,c}$ is torsion-free over $\C[[\hbar]]$ and therefore  flat.

We need to use this to deduce that $A_{N\mid M, K, c}$ is flat over $\C[[\hbar]]$. Let us choose a PBW isomorphism  
$$\Sym^\ast \left( \Oo_{c}(\C^2)  \otimes \mf{gl}(N \mid M) \right) \iso U ( \Oo_{c}(\C^2) \otimes\mf{gl}_N).$$ 
This gives a basis of $U ( \Oo_{c}(\C^2) \otimes\mf{gl}_N)$ where basis elements are of the form 
$$
(E_{\alpha_1 \beta_1} z_1^{k_1} z_2^{m_1} )\dots (E_{\alpha_n \beta_n} z_1^{k_n} z_2^{m_n} ) 
$$ 
using some normal-ordering prescription to specify the order on the product. Here, the indices $\alpha_i$ and $\beta_j$ lie in the set $\{1,\dots,N,\br{1},\dots,\br{M}\}$.  More formally a normal ordering prescription is a choice of total order on the set of $4$-tuples 
$$(\alpha,\beta,k_1,k_2) \in \{1,\dots,N,\br{1},\dots,\br{M} \}^2 \times \Z_{> 0}^2.$$  
One then multiplies elements according to their order in this total order.  

We can lift these basis elements to $A_{N\mid M,\hbar,c}$ by representing them by products of discs $D(\alpha \Downarrow, \shortuparrow^k, \shortdownarrow^m, \beta \Uparrow)$ as in the proof of \ref{proposition_classical_limit_universal_enveloping}.  The products of these discs are taken using the same normal ordering prescription.  

Evidently, these products of discs span $A_{N\mid M,\hbar,c}$ over $\C[[\hbar]]$.  Our task is to show that they form a basis over $\C[[\hbar]]$. This is equivalent to showing that no finite linear combination of our putative basis given by products of discs is annihilated by any power of $\hbar$.  

We say a linear combination of putative basis elements is of degree $n$ if it is a sum of terms each of which is a product of $\le n$ discs. If such a linear combination of degree $n$ is annihilated by some power of $\hbar$, then the same holds after multiplying by $\hbar^n$.  After multiplying by $\hbar^n$, we find something in the image of the map $B_{N\mid M,\hbar,c} \to A_{N\mid M,\hbar,c}$.  Since $B_{N\mid M,\hbar,c}$ is flat over $\C[[\hbar]]$, it suffices to show that if we take an element of our putative basis, multiply it by $\hbar^n$, lift it to an element of $B_{N\mid M,\hbar,c}$, and then set $\hbar = 0$, we get a basis for $B_{\hbar = 0,c}$.  

We can rephrase this as follows.  There is an algebra map
$$
\Sym^\ast \left( \Oo_{c}(\C^2) \otimes \mf{gl}(N \mid M) \right) \to B_{N \mid M, \hbar = 0,c}
$$
which, at the level of generators, sends 
$$
E_{\alpha \beta} z_1^k z_2^m \to D(\alpha \Downarrow, \shortuparrow^k, \shortdownarrow^m, \beta \Uparrow) 
$$
where the disc on the right hand side is viewed as one of the generators of $B_{N \mid M,\hbar,c}$.  This algebra map is surjective.  We need to show that it is an isomorphism.  This, of course, is the same as saying that these elements of $B_{N \mid M,\hbar = 0,c}$ are algebraically independent.

We will check algebraic independence as follows. There is a map of algebras
$$
B_{N \mid M, \hbar = 0,c} \to \Oo ( \mc{M}^{c}_{N \mid M,K} )
$$ 
for all values of $K$.  It suffices to show that, under this map, the image of the  elements $D(\alpha \Downarrow, \shortuparrow^k, \shortdownarrow_m,\beta \Uparrow)$ do not satisfy any polynomial relation of degree $\le K$. 

Let us introduce some notation to describe the space $\mc{M}^{c}_{N \mid M, K}$.  This space is, of course, the symplectic quotient of 
$$
T^\ast \left(\Hom(\C^{N\mid M}, \C^K) \oplus \mf{gl}(K) \right)
$$ 
by $GL(K)$.  Let us describe the elements of this cotangent vector space by
\begin{align*} 
I &: \C^{N \mid M} \to \C^{K}\\
J &: \C^{K} \to \C^{N \mid M} \\
B_1,B_2 &: \C^K \to \C^K.   
\end{align*}

The space $\mc{M}^{c}_{N \mid M,K}$ consists of quadruples $I,J,B_1,B_2$ satisfying the moment map equation
$$
IJ + [B_1,B_2] = c \op{Id} \in \mf{gl}(K)
$$
modulo the action of $GL(K)$.  

Let $E_{\alpha \beta} \in \mf{gl}(N \mid M)$ denote the elementary matrix corresponding to the indices $\alpha,\beta$. The map from $B_{N \mid M,\hbar = 0, c}$ to functions on $\mc{M}^{c}_{N \mid M,K}$ sends 
$$D(\alpha \Downarrow, \shortuparrow^k, \shortdownarrow_m,\beta \Uparrow) \to \op{Tr}_{\C^K} \left( I E_{\alpha \beta} J B_1^k B_2^m    \right).$$ 
Thus, to finish the proof, it suffices to show that these functions on $\mc{M}^{c}_{N \mid M,K}$ do not satisfy any polynomial equation of degree $\le K$. We will do this in the next proposition. 
\end{proof}

\begin{proposition}
Let $I,J,B_1,B_2$ be as in the previous paragraph, and let $E_{\alpha \beta} \in \mf{gl}(N \mid M)$ be an elementary matrix.  The functions
$$
\op{Tr}_{\C^K} \left( I E_{\alpha \beta} J B_1^k B_2^m    \right) 
$$
on $\mc{M}^{c}_{N\mid M, K}$ do not satisfy any polynomial relation of degree $\le K$. 
\end{proposition} 
Before we prove this, we need a lemma.
\begin{lemma} 
Let $R$ be any super-commutative algebra. Let
$$
\Sym^K R = (R^{\otimes k})^{S_K}
$$ 
be the $K$'th symmetric power algebra of $R$. For any element $r \in R$, let $r_i \in R^{\otimes K}$ indicate $r$ placed in position $i$, and let 
$$
P(r) = \sum_{i = 1}^{K} r_i \in \Sym^K R. 
$$ 

Let $R_0 \subset R$ be a complement to the subspace spanned by $1 \in R$.  Choose a basis $r_{\alpha}$ of $R_0$. Then, the elements $P(r_\alpha) \in \Sym^K R$ do not satisfy any polynomial relation of degree $\le K$.  
\end{lemma}
Note that in the case $R = \C[x]$, the elements $P(x^n)$ are the power sums from the theory of symmetric functions.
\begin{proof} 
If $r_1,\dots,r_K$ are elements of $R$, let $r_1 \circ \dots \circ r_K \in \Sym^K R$ be the symmetrization of $r_1 \otimes \dots \otimes r_K$.  For $l < K$, there is a map
\begin{align*} 
 \Sym^l R & \mapsto \Sym^K R \\
r_1 \circ \dots \circ r_l & \mapsto r_1 \circ \dots \circ r_l \circ 1^{\circ (K-l) }. 
\end{align*}
This map is injective.  Let $F^l \Sym^K R$ be the image.  This defines a filtration on $\Sym^K R$, since $F^{l-1} \Sym^K R \subset F^l \Sym^K R$.  The associated graded with respect to this filtration can be described by the isomorphism
$$
F^l \Sym^K R / F^{l-1} \Sym^K R \iso \Sym^l R_0, 
$$
where, as above, $R_0$ is a complement to the subspace of $R$ spanned by the identity.  

In fact, if we restrict the map $\Sym^l R \to \Sym^K R$ to $\Sym^l R_0$, we find a splitting of this filtration and an isomorphism
$$
\Sym^K R \iso \oplus_{l = 0}^{K} \Sym^l R_0. 
$$

Let $r_1,\dots,r_l$ be elements of $R_0$, and let $P(r_i) \in \Sym^K R$ be the elements described in the statement of the lemma.  Then, for $l \le K$ 
$$
P(r_1) \dots P(r_l) = r_1 \circ \dots \circ r_l \circ 1^{\circ (K-l) } \text{ mod } F^{l-1} \Sym^K R. 
$$
From this it follows that the monomials
$$
P(r_{\alpha_1})^{m_1} \dots P(r_{\alpha_i})^{m_K}
$$
where the $m_i \ge 0$ and $\sum m_i \le K$, form a basis for $\Sym^K R$, as the $r_{\alpha_i}$ run over a basis for $R_0$.  Therefore there are no possible polynomial relations of degree $\le K$.

\end{proof}

\begin{proof}[Proof of the proposition]

Recall from the proof of lemma \ref{lemma_flat_adhm} that there is a Zariski open subset 
$$W \subset \mc{M}_{N\mid M,K}^{c}$$ 
which is the twisted cotangent bundle of an affine variety, say $V$, where 
$$
V \subset \left(  (\C^\times)^K \times \C^{ (N-1)K \mid MK} \right)/S_K 
$$ 
is the affine Zariski open subset where we remove the small diagonal from $(\C^\times)^K$.  We will describe this open subset $W$ bundle explicitly.    

Let us first recall the description of $\mc{M}^{c}_{N\mid M,K}$. This variety consists of the set of matrices $B_1,B_2 \in \mf{gl}(K)$, $I: \C^{N \mid M} \to \C^K$, and $J : \C^K \to \C^{N \mid M}$.  These satisfy, as usual, the moment map equation
$$
[B_1,B_2] + IJ = c \op{Id}, 
$$
and are taken up to the action of $GL(K)$.  We will let $I_\alpha \in \C^K$, and $J_{\alpha} : \C^K \to \C$, refer to the components of $I$ and $J$, where $\alpha \in \{1,\dots,N,\br{1},\dots,\br{M}\}$.  

To lie in the Zariski open subset $W$, we will constrain $B_1$ and $I_1$ as follows. We require that $B_1$ is diagonal, with distinct eigenvalues $\lambda_1,\dots,\lambda_K$; and we require that the first component $I_1 \in \C^K$ of $I$ is the vector $(1,\dots,1) \in \C^K$.  Everything is taken up to permutation by the symmetric group $S_K$.  

Let us fix these constraints on $B_1$ and $I_1$.  The moment map constraints will then fix the off-diagonal entries of $B_2$ and the first component $J_1$ of $J$.  The diagonal entries of $B_2$ and the remaining components $I_\alpha$, $J_\alpha$ of $I$ and $J$ will be free to vary.    

To fix $J_1$ and the off-diagonal elements of $B_2$, let us write out the moment map equation in coordinates. We find
$$
 (\lambda_i - \lambda_j) B_{2,ij} + J_{1,j} + \sum_{\alpha \neq 1} I_{i,\alpha} J_{\alpha,j} = c \delta_{ij}. \
$$ 
We can solve for $J_{1,j}$ by using the case $i = j$, in which case
$$
J_{1,j} = c -\sum_{\alpha \neq 1} I_{j,\alpha} J_{\alpha,j}. 
$$
We can then solve for $B_{2,ij}$ when $i \neq j$, to find
$$
B_{2,ij} = (\lambda_i - \lambda_j)^{-1} \left(\sum_{\alpha \neq 1} (I_{j,\alpha} - I_{i,\alpha})J_{\alpha,j} - c  \right).
$$
We let $\mu_i$ denote the diagonal entries of $B_{2}$, so that $B_{2,ii} = \mu_i$.  

In this way, we described the open subset $W \subset \mc{M}^{c}_{N \mid M,K}$ as the quotient by $S_K$ of an affine super variety with coordinates $\lambda_i$, $\mu_i$, $J_{\alpha,j}$ and $I_{i,\alpha}$ when $\alpha \neq 1$.  There are $2 K + 2(N-1) K$ bosonic coordinates and $2MK$ fermionic coordinates.   Since the $\lambda_i$ must be distinct, and we are quotienting by $S_K$, we find that 
$$
\Oo(W) \simeq  \C[\lambda_i,\mu_i, I_{i,\alpha}, J_{\alpha,j}, (\lambda_i - \lambda_j)^{-1}]^{S_K}.  
$$

Now let us consider the functions 
$$
F_{\alpha,\beta,k,l} = \op{Tr}_{\C^K} (I E_{\alpha \beta} J B_1^k B_2^l ) \in \Oo(W). 
$$
in these coordinates.  Our aim is to show that these functions do not satisfy any polynomial relation of degree $\le K$ in $\Oo(W)$. For this purpose, it is not necessary to obtain a complete expression for these functions in our coordinates.  We need only to understand a certain leading term.

Let grade the ring $\Oo(W)$ by giving the variables $\lambda_i$ degree $1$, $(\lambda_i - \lambda_j)$ degree $-1$, and the other generators degree $0$.  Let us expand each element $F_{\alpha,\beta,k,l}$ as a sum of homogeneous elements with respect to this grading, 
$$
F_{\alpha,\beta,k,l} = F^{(k)}_{\alpha,\beta,k,l} + F^{(k-1)}_{\alpha,\beta,k,l}+ \dots + F^{(l)}_{\alpha,\beta,k,l}.
$$
We can describe the term $F^{(k)}_{\alpha,\beta,k,l}$ of highest weight explicitly. We find 
\begin{align*} 
F^{(k)}_{1,1,k,l} &=\sum_i (c - \sum_{\alpha \neq 1} I_{i,\alpha}J_{\alpha,i}) \lambda_i^k \mu_i^l\\ 
F^{(k)}_{\alpha,\beta,k,l} &= \sum I_{i,\alpha}J_{\beta,i} \lambda_i^k \mu_i^l  \text{  if } \alpha, \beta \neq 1\\
 F^{(k)}_{1,\beta,k,l} &= \sum J_{\beta,i} \lambda_i^k \mu_i^l  \text{ if } \beta \neq 1\\    
 F^{(k)}_{\alpha,1,k,l} &= \sum I_{i,\alpha} (c - \sum_{\beta \neq 1} I_{i,\beta}J_{\beta,i}) \lambda_i^k \mu_i^l  \text{ if } \alpha \neq 1.
\end{align*}    
Note that in each case we find a power sum of the kind discussed in the previous lemma, and that if $c \neq 0$ all these power sums are linearly independent. (If $c = 0$, then there is a linear relation  $\sum_{\alpha} F_{\alpha,\alpha,k,l}^{(k)} = 0$ among these polynomials).

We are interested in potential polynomial relations among the elements $F_{\alpha,\beta,k,l}$.  Let us introduce the polynomial ring $\C[\mc{F}_{\alpha,\beta,k,l}]$ on free generators bearing the same indices as the elements $F_{\alpha,\beta,k,l}$.  Let us grade this polynomial ring by giving the generator $\mc{F}_{\alpha,\beta,k,l}$ weight $k$.  There is a map
$$
\rho : \C[\mc{F}_{\alpha,\beta,k,l}] \to \Oo(W)  
$$
sending the generator $\mc{F}_{\alpha,\beta,k,l}$ to the function $F_{\alpha,\beta,k,l}$.  This does not preserve the gradings on both sides. However, we can filter each algebra by saying that $F^n \C[\mc{F}_{\alpha,\beta,k,l}]$ consists of elements of weight $\le n$, and that $F^n \Oo(W)$ consists of elements of weight $\le n$, using the gradings we have defined on each algebra.  Then, $\rho$ is a map of filtered algebras.  

In each case, the filtration is split, so that we can identify the filtered algebra with the associated graded.  We thus have a map
\begin{align*} 
 \op{Gr} \rho : \C[\mc{F}_{\alpha,\beta,k,l}] & \to \Oo(W) \\
\mc{F}_{\alpha,\beta,k,l} & \mapsto F^{(k)}_{\alpha,\beta,k,l}.  
\end{align*} 

Now suppose that $P \in \C[\mc{F}_{\alpha,\beta,k,l}]$.  Let us expand $P$ as a sum
$$
P = P^{(n)} + P^{(n-1)} + \dots + P^{(0)}
$$
of homogeneous elements.  

Suppose that $P$ is of degree $\le K$ as a polynomial in the generators $\mc{F}_{\alpha,\beta,k,l}$, and that $\rho(P)= 0$.   We need to show that this leads to a contradition.  Since $\rho$ is a map of filtered algebras, we find
$$
\rho(P) = (\op{Gr} \rho) (P^{(n)} )  + \text{ lower order terms} 
$$
where the lower order terms are in $F^{n-1} \Oo(W)$.   Thus, if $\rho(P) = 0$, we must also have $(\op{Gr} \rho) (P^{(n)}) = 0$.  As we have seen, $(\op{Gr} \rho) (P^{(n)})$ is the polynomial $P^{(n)}$ applied to the power sums $F_{\alpha,\beta,k,l}^{(k)}$.  

Thus, we have found that there must be a polynomial relation of degree $\le K$ among these power sums.  However, the previous lemma tells us that this can not be the case.  

\end{proof} 

\section{Uniqueness of the deformation of $U(\op{Diff}(\C) \otimes \mf{gl}_N)$}
\label{section_uniqueness}

In this section we will prove the following theorem.
\begin{theorem} 
Consider the algebras $A_{N,\hbar,c}^{QFT}$ and $A_{N,\hbar,c}^{comb}$ over $\C[[\hbar]]$, each of which deforms $U(\Oo_{c}(\C^2) \otimes \mf{gl}_N)$.  Then, 
\begin{enumerate} 
\item The central element $1 \otimes \op{Id} \in \Oo_c(\C^2) \otimes \mf{gl}(N)$ quantizes to  central elements $\kappa \in A_{N,\hbar,c}^{QFT}$ and $\kappa \in A_{N,\hbar,c}^{comb}$. 
\item  There is an isomorphism of (topological) associative algebras  
$$
A_{N,\hbar_{QFT},1}^{QFT} \iso A_{N,\hbar_{comb},1}^{comb}.
$$
which sends
$$
\hbar_{QFT} \mapsto 2^{5} \pi^{2} \hbar_{comb} + f_2(\kappa) \hbar_{comb}^2  + \dots
$$
where the $f_i$ are polynomials of degree at most $i-1$. 
\end{enumerate}
Further, the isomorphism is unique provided it satisfies certain natural properties we will discuss below. 
\end{theorem}
This theorem thus says that, after reparameterizing $\hbar$ in a $\kappa$-dependent way, the two algebras are isomorphic.  The central element $\kappa \in A_{N,\hbar,c}^{comb}$ is such that under the homomorphism
$$
A_{N,\hbar,c} \to \Oo_{\hbar}(\mc{M}^c_{N,K})
$$
the central element $\kappa$ gets sent to $K$. This means that, from a physics point of view, we should interpret the parameter $\kappa$ as being the parameter $K$ in a large $K$ expansion.  Thus, our isomorphism occurs after reparameterizing $\hbar$ in a $K$-dependent way. 

I should also remark that there is some inherent ambiguity in the definition of $A_{N,\hbar,c}^{QFT}$, because our $5$-dimensional gauge theory does not admit a unique quantization. Rather, at each order in $\hbar$, we are free to add on two linearly independent terms to the Lagrangian.  When translated, by Koszul duality, into deformations of the algebra $U(\Oo_c(\C^2) \otimes \mf{gl}(N))$, these two linearly independent deformations are no longer linearly dependent. This is because, as discussed in section \ref{sec:Koszul duality}, the action of $\C^\times$ which scales $\C^2$ preserves the associative product in the $\R$-direction of $\R \times \C^2$. It also scales the parameters $c,\hbar$ with weight $2$. This allows us to scale away one of the two deformations.  

It follows from this and the discussion in section \ref{sec:Koszul duality} that different choices of quantization of our field theory lead to algebras $A_{N,\hbar,c}^{QFT}$ which are related by a change of coordinates of the parameter $\hbar$.  We can therefore normalize our choice of quantization so that the functions $f_i(\kappa)$, when $i > 1$, vanish when $\kappa = 0$. With this normalization, the algebras become isomorphic when we set $\kappa = 0$ and scale $\hbar$ appropriately. 

\subsection{}
The rest of this section will be devoted to the proof of this theorem.  The proof is entirely algebraic, but is rather non-trivial.  The basic strategy is to use deformation theory to relate different quantizations of the algebra $U(\Oo_c(\C^2) \otimes \mf{gl}(N))$. For this, we need some understanding of the Hochschild cohomology of this algebra.  This Hochschild cohomology is very difficult to compute: for example, in the case $N = 1$, this Hochschild cohomology is related to the Lie algebra cohomology of the Lie algebra $\Oo_c(\C^2) = \op{Diff}(\C)$, which is unknown \cite{Fuk12}. 

If $N$ is sufficiently large, however, the Hochschild cohomology computations simplify: there is a certain filtered piece of Hochschild cohomology that we can compute exactly for sufficiently large $N$, and as $N \to \infty$ we can compute all the Hochschild cohomology. 

We need to leverage this stabilization of Hochschild cohomology to constrain quantizations of $U(\Oo_c(\C^2) \otimes \mf{gl}(N))$.  For this, we will consider the super Lie algebras $\mf{gl}(N +R \mid R)$ and the algebra $U(\Oo_c(\C^2) \otimes \mf{gl}(N + R \mid R))$.  For $R$ sufficiently large, the Hochschild cohomology stabilizes as above.  Further, the Lie algebras $\mf{gl}(N +R \mid R)$ are related for different values of $R$: if we let $Q \in \mf{gl}(N +R \mid R)$ denote a fermionic matrix of rank $(0 \mid 1)$, then $\mf{gl}(N+R-1 \mid R-1)$ is the $Q$-cohomology of $\mf{gl}(N+R \mid R)$.  This relationship extends to all objects built from these super Lie algebras. In particular, we find that the $Q$-cohomology of $U(\Oo_c(\C^2) \otimes \mf{gl}(N+R \mid R))$  is $U(\Oo_c(\C^2) \otimes \mf{gl}(N+R-1\mid R_1))$. 

This suggests a strategy for relating different quantizations of $U(\Oo_c(\C^2) \otimes \mf{gl}(N+R \mid R))$.  We will look for families of quantizations $A_{N+R \mid R,\hbar,c}$ of $U(\Oo_c(\C^2) \otimes \mf{gl}(N+R \mid R))$, where the algebra $A_{N+R \mid R,\hbar,c}$ is the $Q$-cohomology of $A_{N+R+1\mid R+1,\hbar,c}$, as above.  It turns out that the obstruction and deformation groups controlling the construction of such families are the stable Hochschild cohomology groups of $U(\Oo_c(\C^2) \otimes \mf{gl}(N+R \mid R) ) $.  Since we can compute these stable groups, we then have good control over the problem.  

\subsection{}
Let us now introduce some algebraic language for discussing families of quantizations $A_{N+R \mid R,\hbar,c}$ related as $R$ varies in the way suggested above.  We will use the language of formal moduli problems \cite{Lur12}, which is a way of formalizing the ideas of deformation theory.  A formal moduli problem is something which assigns a  groupoid (or more generally a simplicial set) to every Artinian ring, in a way which is functorial and which satisfies certain additional axioms.  One can perform obstruction-theory arguments with any formal moduli problem.  The language of formal moduli problems isn't strictly necessary for what we're doing, but using this language makes the arguments clearer.   

Let us fix two associative algebras $B_1,B_2$ with a map $f : B_1 \to B_2$.  
\begin{definition} 
Define the formal moduli problem $\op{Def}(B_2,f)$ to be the formal moduli space of deformations of the pair $(B_2,f)$.
\end{definition}
More formally, given an Artinian dg ring $S$ with maximal ideal $m$, the $\infty$-groupoid of $S$-points of $\op{Def}(B_2,f)$ consists of  
\begin{enumerate} 
\item Algebras $B_{2,S}$, flat over $S$, with an isomorphism to $B_2$ when reduced modulo $m$. 
\item Maps $f_S : B_1 \otimes S \to B_{2,S}$ of $S$-algebras which reduce to $f$ modulo $m$.  
\end{enumerate} 
In this section, we will compute the tangent complex of this formal moduli problem in terms of certain Hochschild cohomology groups of $B_1$ and $B_2$.

As with any formal moduli problem \cite{Lur12}, we can write $\op{Def}(B_2,f) $ as the Maurer-Cartan functor associated to some $L_\infty$ algebra.  Let $CH^\ast(B_2)$ denote the Hochschild cochain complex of $B_2$.  For simplicity, we will assume that $B_1$ and $B_2$ are algebras concentrated in degree $0$. We will first describe the individual homotopy Lie algebras describing deformations of $B_2$ and deformations of the map $f$, and then explain how to put them together to get the homotopy Lie algebra controlling deformations of the pair $(B_2,f)$.   

Let $CH_+^\ast(B_2)$ denote the part of the Hochschild complex of $B_2$ concentrated in degrees greater than zero. This is the cone of the natural map from $CH^\ast(B_2)$ to $B_2$.  Then, $CH_+^\ast(B_2)[1]$ is the dg Lie algebra controlling deformations of $B_2$.       If $\til{B}_2$ is a quasi-free resolution of $B_2$, then there is a quasi-isomorphism of dg Lie algebras
$$
CH_+^\ast(B_2) \simeq \op{Der}^\ast(\til{B}_2)
$$
where on the right hand side we have the dg Lie algebra of derivations.

We can view $B_2$ as a $B_1$ bimodule via the map $f$. Thus, we can consider the Hochschild cochain complex $CH^\ast(B_1,B_2)$ of $B_1$ with coefficients in the bimodule $B_2$, and the subcomplex $CH^\ast_+(B_1,B_2)$ of cochains in positive degree.  We can define a natural dg associative algebra structure on $CH^\ast(B_1,B_2)$, as follows. If $\mu_1 \in CH^n(B_1,B_2)$ and $\mu_2 \in CH^m(B_1,B_2)$ then we define
$$
( \mu_1 \ast\mu_2) (\beta_1,\dots,\beta_{n+m} ) = \mu_1(\beta_1,\dots,\beta_n) \mu_2(\beta_{n+1},\dots,\beta_{n+m}).  
$$
The $\beta_i$ are in $B_1$ and on the right hand side we are using the product in $B_2$. This product is manifestly associative, and so the corresponding Lie bracket satisfies the Jacobi identity.  One can check that this product is compatible with the Hochschild differential, and makes $CH^\ast(B_1,B_2)$ into a dg associative algebra. We are interested in $CH^\ast_+(B_1,B_2)$, which forms a sub dg Lie algebra.  

This dg Lie algebra controls deformations of the map $f : B_1 \to B_2$ as a map of homotopy associative algebra. To see this, consider deforming $f$ to $f +  g$, where $g$ is an $A_\infty$ map and so has components $g_n : B_1^{\otimes n} \to B_2$ for $n \ge 1$. We can treat $g$ as an element of $CH^\ast_+(B_1,B_2)$.  One can check directly that $f+g$ defines an $A_\infty$ map if and only if $g$, viewed as an element of $CH^\ast_+(B_1,B_2)$, satisfies the Maurer-Cartan equation.  

So far, we have understood the formal moduli problems describing deformations of $B_2$ and of $f$. Next, we need to put these together to find a description of the formal moduli problem controlling deformations of the pair $(B_2,f)$. 
 
Since $\op{Def}(B_2,f)$ is a formal moduli problem, it's shifted tangent space $T[-1] \op{Def}(B_2,f)$ has a homotopy Lie algebra structure \cite{Lur12}.  There is a fibre sequence of formal moduli problems
$$
\op{Def}(f) \to \op{Def}(B_2,f) \to \op{Def}(B_2)
$$ 
which becomes a fibre sequence of homotopy Lie algebras
$$
CH^\ast_+ (B_1,B_2) \to T[-1] \op{Def}(B_2,f) \to CH^\ast_+(B_2,B_2)[1]. 
$$ 
Therefore, we can describe $T[-1] \op{Def}(B_2,f)$ with its homotopy Lie algebra structure as being built from the direct sum $CH^\ast_+(B_2,B_2) [1] \oplus CH^\ast_+(B_1,B_2)$ where the $L_\infty$ structure is deformed by ``triangular'' terms. 

The higher  terms will not be important for us. What will be important, however, is to notice that $l_1$ is the natural restriction map
$$
CH^\ast_+(B_2,B_2)[1] \to CH^\ast_+(B_1,B_2).
$$
We can see this as follows. Suppose we have a first order deformation of the algebra $B_2$ into an $A_\infty$ algebra, where only non-zero $A_\infty$ structure map that is deformed is
$$
m_n : B_2^{\otimes n} \to B_2.
$$
The deformation parameter is a parameter $\eps$ of square zero and degree $2-n$.  

We would like to know if the map $f$ deforms to an $A_\infty$ map $f+\eps g$ from $B_1$ to $B_2[\eps]/\eps^2$ with this deformed product.  We can do this if we can find some
$$
g:  B_1^{\otimes n-1} \to B_2
$$
such that 
$$
\d g = f^\ast m_n
$$
where $f^\ast m_n : B_1^{\otimes n} \to B_2$ is the pull-back of $m_n$ along the map $f$, and $\d g$ is the differential on $CH^\ast_+(B_1,B_2)$.  In other words, $f^\ast m_n$ is the obstruction to finding a deformation of the map $f$ compatible with a given deformation of the algebra $B_2$.

We have thus shown that  $T[-1] \op{Def}(B_2,f)$ is the cone of the natural restriction map $$
CH^\ast_+(B_2,B_2)  \to CH^\ast_+(B_1,B_2).
$$
Note also that we get the same cone of we replace $CH^\ast_+(B_1,B_2)$ by $CH^\ast(B_1,B_2)$ and $CH^\ast_+(B_2,B_2)$ by $CH^\ast(B_2,B_2)$. The point is that if we do this, we add a copy of $B_2$ to each complex and they cancel out when we form the cone.

This fact will be useful for later, so we will frame it as a lemma.
\begin{lemma}
\label{lemma_les} 
There is a long exact sequence in cohomology
$$
\dots \to H^i(T\op{Def}(B_2,f)  ) \to HH^{i+2}(B_2,B_2) \to HH^{i+2}(B_1,B_2) \to \dots
$$
where the map
$$
HH^{i+2}(B_2,B_2) \to HH^{i+2}(B_1,B_2) 
$$
is the natural restriction map.  
\end{lemma}
Another remark that will be useful later is that both $CH^\ast(B_1,B_2)$ and $CH^\ast(B_2,B_2)$ are modules over the center $Z(B_2)$.  The formula for the module structure is very simple: an element $\kappa \in Z(B_2)$ acts on a Hochschild cochain 
$$
\mu : B_i^{\otimes n} \to B_2
$$
(for $i = 1,2$) by post-composing with multiplication by $\kappa$.  One checks that this preserves the differential.  Evidently, the restriction map
$$
CH^\ast(B_2,B_2) \to CH^\ast(B_1,B_2)
$$
is a map of $Z(B_2)$-modules.  It follows from this that the tangent complex to the moduli problem $\op{Def}(B_2,f)$ is a module for $Z(B_2)$.   

\subsection{}
Let us specialize the deformation problem a bit to one closer to that we are ultimately interested in.  Let $B$ be any algebra, and consider the pair of super-algebras $U(\mf{gl}(N+R \mid R)\otimes B)$ and $U(\mf{sl}(N+R \mid R))$.  There is a homomorphism from the second algebra to the first which we can call $f$. 
\begin{definition} 
Let $\mc{M}_{N+R \mid R, B}$ be the formal moduli problem $\op{Def}( U(\mf{gl}(N+R \mid R) \otimes B), f)$.  
\end{definition}
Note that the tangent complex $T \mc{M}_{N+R \mid R,B}$ is, as I mentioned above, a module for the center $Z(U(\mf{gl}(N+R \mid R) \otimes B) ) $.  There is a homomorphism
$$
\C[\kappa] \to Z(U(\mf{gl}(N+R \mid R) \otimes B) )   
$$
sending $\kappa$ to $1 \otimes \op{Id}_{\mf{gl}(N+R \mid R)}$.  Thus, in particular, the tangent complex for $\mc{M}_{N+R \mid R,B}$ is a $\C[\kappa]$-module.

We will define a map of formal moduli problems
$$
\mc{M}_{N+R \mid R,B} \to \mc{M}_{N+R-1 \mid R-1,B} 
$$
implementing the $Q$-cohomology procedure we mentioned earlier.  To make this $Q$-cohomology construction work better with differential graded constructions, we will phrase it as adding $Q$ to the differential rather than as taking cohomology.

Thus, in general, suppose we have a bounded complex of super-vector spaces, say $V$, which is acted on by $\mf{sl}(1 \mid 1)$, and suppose the action of the one-dimensional maximal torus lifts to an action of the rank one maximal torus $H$ of the group $SL(1 \mid 1)$, and that the weights of the action of $H$ on $V$ are all even.    Let $Q \in \mf{sl}(1\mid 1)$ be a matrix of rank $(0\mid 1)$. Note that this maximal torus $H \iso \C^\times$ acts on $Q$ with weight $2$.   

Introduce a parameter $t$ which is of weight $-2$ under $H$, is odd under the super $\Z/2$-grading, and is of cohomological degree $1$.  Note that since $t$ is of both cohomological degree $1$ and odd in the super grading, it is a commuting as opposed to anti-commuting variable. Then, we can introduce a new complex $V((t))$ with differential $\d_V + t Q$.  This differential is of cohomological degree $1$, super-degree $0$ and is $H$-invariant.  We can then form the $H$-invariants $V((t))^{H}$.   This procedure has the effect of rearranging the grading of $V$ according to how $H$ acts, and then introducing a new differential (now of cohomological degree $1$) by adding $Q$ to the existing differential. 

Let us see how this produces maps of formal moduli problems
$$
\mc{M}_{N+R \mid R,B} \to \mc{M}_{N+R-1 \mid R-1,B}. 
$$
Suppose we have an Artinian ring $S$, and an $S$-point of $\mc{M}_{N+R \mid R,B}$. This consits of a super-algebra $A_{N+R \mid R,S}$ over $S$ which deforms $U(\mf{gl}(N+R \mid R)\otimes B)$, and a homomorphism
$$
f_S : U(\mf{sl}(N+R \mid R)) \to A_{N+R \mid R,S} 
$$
deforming the obvious homomorphism to $U(\mf{gl}(N+R \mid R)\otimes B)$.

Choosing a copy of $\mf{sl}(1 \mid 1)$ inside of $\mf{sl}(R \mid R)$, we find that $\mf{sl}(1 \mid 1)$ acts on both $A_{N+R \mid R,S}$ and of course on $U(\mf{sl}(N+R \mid R))$. We can apply the construction described above to find 
$$
f_S^Q : U^Q(\mf{sl}(N+R\mid R)) \to A^{Q}_{N+R \mid R,S}. 
$$ 
We find that $A^{Q}_{N+R \mid R,S}$ is a deformation of an algebra canonically quasi-isomorphic to $U(\mf{gl}(N+R -1 \mid R-1 )\otimes B)$, and that $U^Q(\mf{sl}(N+R \mid R))$ is canonically quasi-isomorphic to $U(\mf{sl}(N+R-1\mid R-1))$. In this way we have produced an $S$-point of the formal moduli problem $\mc{M}_{N+R-1\mid R-1,B}$.  
 
\begin{definition}
Let $\mc{M}_{N+\infty\mid \infty,B}$ denote the homotopy limit
$$
\mc{M}_{N+\infty \mid \infty,B} = \lim \mc{M}_{N+R \mid R,B} 
$$
of these formal moduli problems under these maps. 
\end{definition}
\begin{lemma}
The tangent complex $T \mc{M}_{N+\infty \mid \infty,B}$ is a $\C[\kappa]$-module.  
\end{lemma}
\begin{proof}
At finite $R$, the tangent complex $T \mc{M}_{N+R \mid R,B}$ are $\C[\kappa]$-modules. The maps 
$$
T \mc{M}_{N+R \mid R,B} \to T \mc{M}_{N+R-1 \mid R-1,B}  
$$
are maps of $\C[\kappa]$-modules, so that the $\C[\kappa]$-module structures survives the formation of the limit. 
\end{proof}

\subsection{}
The first algebraic theorem we will prove is a description of the cohomology of the tangent complex of $\mc{M}_{N+\infty \mid \infty,B}$.  Before we describe the result, we need to introduce a little notation.  Consider the algebra $B \oplus B^\vee[-1]$ which is the square-zero extension of $B$ by the bimodule $B^\vee$, the linear dual of $B$.  In the examples we will consider, $B$ will typically be a countable-dimensional complex vector space, and we should give $B^\vee$ it's natural topology.  Then, we can consider the cyclic cohomology
$$
HC^\ast(B \oplus B^\vee[-1]).
$$
The cyclic cochain complex is, of course, built from multi-linear maps
$$
(B \oplus B^\vee)^{\otimes n} \to \C. 
$$
We require these to be continuous multilinear maps.  If we do this, we find, for example, that multilinear maps from $(B^\vee)^{\otimes n}$ to $\C$ are the same as elements of $B^{\otimes n}$.  

The introduction of topological vector spaces is not strictly necessary, as one can describe the continuous cyclic cochain complex of $B \oplus B^\vee[-1]$ entirely in algebraic terms, but I feel that using the language of topological vector spaces makes the presentation more palatable. 

There's a natural map
$$
HH^\ast(B,B)[-1] \to HC^\ast(B\oplus B^\vee[-1])
$$
which is an isomorphism onto the summand in $HC^\ast(B \oplus B^\vee[-1])$ consisting of cyclic cochains which only involve one copy of $B^\vee[-1]$ and any number of copies of $B$.  We will let
$$
\kappa \in HC^\ast(B \oplus B^\vee[-1]) 
$$
denote the image of the identity element in $HH^0(B,B)$ under this map.  Multiplying by $\kappa$ gives an action of $\C[\kappa]$ on the symmetric algebra $\Sym^\ast ( HC^\ast(B \oplus B^\vee[-1] ) [-1] ) $.

Let us introduce a graded vector space
$$
V = t \C[[t]] \oplus \bigoplus_{n \ge 1} (B^{\otimes n})_{C_n}[1]
$$ 
where the subscript $C_n$ indicates cyclic coinvariants, and the parameter $t$ is of degree $2$.  We can identify 
$$
V = HC^\ast (\C \oplus B^\vee[-1] ) / \C
$$
where we quotient by the unit element in $HC^0$.   Let $\kappa \in V$ be the identity element $1 \in B$ viewed as an element of $V$.  Multiplication by $\kappa$ gives an action of $\C[\kappa]$ on $\Sym^\ast V[-1]$.  
\begin{theorem} 
There is a long exact sequence
$$
\dots \to H^i T \mc{M}_{N+\infty\mid \infty,B} \to H^{i+2} \left(  \Sym^\ast \left( HC^\ast ( B \oplus B^\vee[-1] )[-1] \right) \right) \to H^{i+2} \left(\Sym^\ast \left(V[-1] \right)  \right) \to \dots 
$$
This is a sequence of $\C[\kappa]$-modules, where the action of $\C[\kappa]$ on the tangent complex $T \mc{M}_{N+\infty \mid \infty,B}$ was discussed above.
\label{theorem_tangent_complex_cyclic} 
\end{theorem}
\begin{proof}
The first thing to note is that the formation of the tangent complex of a formal moduli problem commutes with limits, so that
$$
T \mc{M}_{N+\infty \mid \infty,B} = \lim T \mc{M}_{N+R\mid R,B}.
$$
Next, we can use the long exact sequence of lemma \ref{lemma_les} (or rather, it's cochain-level version) to show that
\begin{multline*} 
 T \mc{M}_{N+\infty \mid \infty,B}\\
 = \lim_R \op{Cone} \left( CH^\ast(U(\mf{gl}(N+R \mid R) \otimes B))  )[1] \to CH^\ast(U(\mf{sl}(N+R \mid R) ) , U(\mf{gl}(N+R  \mid R) \otimes B) )[1]   \right). 
\end{multline*} 
To prove the result, we thus need to show that
\begin{align*} 
\lim_R  CH^\ast(U(\mf{gl}(N+R \mid R) \otimes B) ) & \simeq \Sym^\ast \left(  CC^\ast ( B \oplus B^\vee[-1]) [-1]\right)  \\
\lim_R CH^\ast ( U(\mf{sl}(N+R \mid R) ,   U(\mf{gl}(N+R \mid R) \otimes B)  ) & \simeq \Sym^\ast \left( CC^\ast ( \C \oplus B^\vee[-1] )[-1] / \C[-1] \right).  
\end{align*} 
Here $CH^\ast$ and $CC^\ast$ indicate Hochschild and cyclic cochain complexes. 

Recall that for any Lie algebra $\mf{g}$, there is a quasi-isomorphism between Hochschild cochains of $\mf{g}$ and the Lie algebra cochains $C^\ast(\mf{g},U(\mf{g}))$.  Further, the PBW theorem gives us a $\mf{g}$-equivariant isomorphism
$$
\Sym^\ast \mf{g} \to U(\mf{g}). 
$$
Explicitly, the map $\Sym^n \mf{g} \to U(\mf{g})$ is given by the formula
$$
X_1 \otimes \dots \otimes X_n \to \sum_{\sigma \in S_n} X_{\sigma(1)} \dots X_{\sigma(n)} 
$$
where on the right hand side we are using the product in the universal enveloping algebra.  This map is symmetric, and so descends to $\Sym^n \mf{g}$, and is evidently $\mf{g}$-equivariant.   

This isomorphism allows us to identify the Hochschild cochains of $U(\mf{g})$ with $C^\ast(\mf{g}, \Sym^\ast \mf{g})$.   This, in turn, can be identified with the Lie algebra cochains of $\mf{g} \oplus \mf{g}^\ast [-1]$.

Applied to the situation at interest, we find a quasi-isomorphism
$$
CH^\ast(U(\mf{gl}(N+R \mid R) \otimes B) ) \simeq C^\ast\left( \mf{gl}(N+R \mid R) \otimes \left( B \oplus B^\vee[-1] \right) \right). 
$$ 
A similar argument gives us a quasi-isomorphism 
\begin{align*} 
 CH^\ast ( U(\mf{gl}(N+R \mid R) ,   U(\mf{gl}(N+R \mid R) \otimes B)  ) & \simeq C^\ast(\mf{gl}(N+R \mid R), \Sym^\ast (\mf{gl}(N+R \mid R) \otimes B) ) \\
&\simeq C^\ast\left( \mf{gl}(N+R \mid R) \otimes \left( \C \oplus B^\vee[-1] \right) \right). 
\end{align*}
We will modify this to obtain the complex computing the cohomology of $\mf{sl}(N+R \mid R)$ with coefficients in $U(\mf{gl}(N+R \mid R) \otimes B)$ at the end of our argument.  

Let us take the limit as $R \to \infty$ and see what we find.  There is a Lie algebra homomorphism 
$$
\mf{gl}(N+R \mid R) \to \mf{gl}(N+R +1 \mid R+1). 
$$
This leads to natural restriction cochain maps 
$$
 C^\ast\left(\mf{gl}(N+R +1 \mid R+1 ) \otimes (B \oplus B^\vee[-1] ) \right) \to  C^\ast\left(\mf{gl}(N+R \mid R) \otimes (B \oplus B^\vee[-1] ) \right) .  
$$
One can verify easily that this is the map that is relevant for the limit we are computing.

Let us now invoke the next lemma \ref{lemma_super_LQT}, which proves a   version of the Loday-Quillen-Tsygan theorem that applies to super groups.  This implies that there are quasi-isomorphisms
\begin{align*} 
 \Sym^\ast CC^\ast(B \oplus B^\vee[-1] ) [-1] & \to \lim_{R} C^\ast\left(\mf{gl}(N+R  \mid R ) \otimes (B \oplus B^\vee[-1] ) \right) \\ 
 \Sym^\ast CC^\ast(\C \oplus B^\vee[-1] ) [-1] &\to \lim_{R} C^\ast\left(\mf{gl}(N+R  \mid R ) \otimes (\C \oplus B^\vee[-1] ) \right). 
\end{align*}
This almost completes the proof. The final step to note is that
$$
HC^\ast(\C \oplus B^\vee[-1] ) = HC^\ast(\C) \oplus \bigoplus_{n \ge 1} (B^{\otimes n})_{C_n}[1] 
$$
and that $HC^\ast(\C) = \C[[t]]$.  Replacing the first occurrence of $\mf{gl}(N+R \mid R)$ by $\mf{sl}(N+R \mid R)$ in our computation of $C^\ast(\mf{gl}(N+R \mid R),\Sym^\ast (B \otimes \mf{gl}(N+R \mid R) ) ) $ has the effect of getting rid of $HC^0(\C)$, and so replacing $\C[[t]]$ by $t \C[[t]]$.

The fact that the sequence is a sequence of $\C[\kappa]$-modules is easy to verify.  
\end{proof}

\begin{lemma}
Let $A$ be any differential graded algebra.  Then, for any $N$, there is a quasi-isomorphism 
$$
\Sym^\ast \left( CC^\ast(A)[-1] \right) \simeq \lim_R C^\ast \left( A \otimes \mf{gl}(N+R \mid R)\right).  
$$
\label{lemma_super_LQT}
\end{lemma}
\begin{proof}
Let $A$ be any dg algebra.  The inclusion
$$
C^\ast(A \otimes \mf{gl}(N+R \mid R) )^{GL(N+R) \times GL(R) } \into C^\ast( A \otimes \mf{gl}(N+R \mid R) )
$$
is a quasi-isomorphism, since the action of $GL(N+R) \times GL(R)$ is semi-simple.  We can thus consider the limit
$$
\lim_R C^\ast(A \otimes \mf{gl}(N+R \mid R) )^{GL(N+R) \times GL(R) }. 
$$
Further, there is a natural restriction map
$$
C^\ast(A \otimes \mf{gl}(N+R \mid R) )^{GL(N+R) \times GL(R) } \to C^\ast(A \otimes \mf{gl}(R \mid R) )^{GL(R) \times GL(R) }.  
$$
This map is compatible with the maps in the inverse system, and gives a map 
$$
\lim_R C^\ast(A \otimes \mf{gl}(N+R \mid R) )^{GL(N+R) \times GL(R) } \to \lim_R C^\ast(A \otimes \mf{gl}(R \mid R) )^{GL(R) \times GL(R) }.  
$$
Invariant theory for the general linear group tells us that, this map is an isomorphism of cochain complexes.  We thus find a quasi-isomorphism of cochain complexes
$$
\lim_R C^\ast(A \otimes \mf{gl}(R \mid R) )^{GL(R) \times GL(R) } \simeq \lim_R  C^\ast( A \otimes \mf{gl}(N+R \mid R) ). 
$$

Next, note that the inclusion
$$
C^\ast(A \otimes \mf{gl}(R \mid R) )^{GL(R)\times GL(R)  } \into  C^\ast(A \otimes \mf{gl}(R \mid R) )^{GL(R) } 
$$
is quasi-isomorphism, where on the right hand side we are taking invariants with respect to the diagonal subgroup.  

Now, $\mf{gl}(R \mid R) = \mf{gl}(1 \mid 1) \otimes \mf{gl}(R)$.   Thus we find, invoking the standard Loday-Quillen-Tsygan theorem, that   
$$
\lim_R C^\ast(A \otimes \mf{gl}(R \mid R) )^{GL(R) } = \Sym^\ast \left(CC^\ast(A \otimes \mf{gl}(1 \mid 1) ) [-1] \right) . 
$$
The final step is to note that cyclic cohomology is Morita invariant, and that $A \otimes \mf{gl}(1 \mid 1)$ is Morita equivalent (as a super algebra) to just $A$.  It follows that we can replace $A \otimes \mf{gl}(1 \mid 1)$ in the above equation by just $A$.
\end{proof}

\subsection{} 
Next, we will specialize to the algebra $\op{Diff}(\C)$.  This algebra has a $\C^\times$-action where $z$ has weight $1$ and $\dpa{z}$ has weight $-1$.  We are interested in the $\C^\times$ fixed points of the moduli problem $\mc{M}_{N+\infty\mid\infty,\op{Diff}(\C)}$. We will indicate this by $\mc{M}^{\C^\times}_{N+\infty \mid \infty,\op{Diff}(\C)}$.  
  
The main algebraic theorem we will prove is the following. 
\begin{theorem} 
There are isomorphisms of $\C[\kappa]$-modules
\begin{align*}
H^0 T\mc{M}^{\C^\times}_{N+\infty \mid \infty,\op{Diff}(\C)}& = \C[\kappa] \\
 H^1 T\mc{M}^{\C^\times}_{N+\infty \mid \infty,\op{Diff}(\C)}& =0.
\end{align*}
Further, although $H^{-1} T\mc{M}^{\C^\times}_{N+\infty \mid \infty,\op{Diff}(\C)}$ is non-zero, the natural map
$$
H^{-1} T\mc{M}^{\C^\times}_{N+\infty \mid \infty,\op{Diff}(\C)} \to HH^1( U (\mf{gl}(N+R \mid R) \otimes \op{Diff}(\C) )  
$$
is zero for all $N$ and $R$.

\label{theorem_deformation_computation}
\end{theorem}
There is one more aspect to the story which we will phrase as a separate proposition.
\begin{proposition}
%\label{proposition_explicit_formula}
There is a unique generator for the $\C[\kappa]$ module $H^0 T\mc{M}^{\C^\times}_{N+\infty \mid \infty,\op{Diff}(\C) }$ with the following property. Let us view this generator as giving rise to a deformation of the algebra $U(\mf{gl}_N\otimes\op{Diff}(\C)$.  Let $E_{\alpha \beta} \in \mf{gl}(N)$ denote the elementary matrix.  Then, in this deformation, 
\begin{align*} 
[E_{\alpha\beta}, E_{\gamma \delta}] =& \delta_{\beta \gamma} E_{\alpha \delta} - \delta_{\gamma \alpha} E_{\gamma \beta} \\
 [E_{\alpha\beta}, z E_{\gamma \delta}] =& \delta_{\beta \gamma} z_i E_{\alpha \delta} - \delta_{\gamma \alpha} z E_{\gamma \beta} \\
[E_{\alpha\beta}, \partial E_{\gamma \delta}] =& \delta_{\beta \gamma} \partial E_{\alpha \delta} - \delta_{\gamma \alpha} \partial E_{\gamma \beta} \\
 [E_{\alpha\beta}\partial, E_{\gamma\delta} z] =& \delta_{\beta \gamma}E_{\alpha\delta}(\partial   z) - \delta_{\alpha\delta}E_{\gamma\beta}(z   \partial)\\ 
& + \hbar E_{\alpha \delta} E_{\gamma \beta} - \hbar \sum_{\mu} \delta_{\beta \gamma}  E_{\alpha \mu} E_{\mu \delta}  
- \hbar\sum_{\mu} \delta_{\alpha \delta} E_{\gamma \mu} E_{\mu \beta} .
\end{align*} 
Here $\hbar$ is the deformation parameter.  
\end{proposition}
We will prove this in the appendix.

This theorem will take a fair bit of work to prove.  However, we have some immediate consequences. 
\begin{corollary}
There is a $\C[[\hbar]]$-point of $\mc{M}^{\C^\times}_{N+\infty \mid \infty , \op{Diff}(\C)}$ which, to first order, is given by a generator of the $\C[\kappa]$-module $H^0 T \mc{M}^{\C^\times}_{N+\infty \mid \infty,\op{Diff}(\C)}$.  Explicitly, this is given by a sequence of algebras $U_{\hbar}(\mf{gl}(N+R \mid R) \otimes \op{Diff}(\C))$ deforming the universal enveloping algebra $U(\mf{gl}(N+R \mid R) \otimes \op{Diff}(\C))$, where to first order the deformation is as in the previous proposition. These algebras come equipped with a Lie algebra map from $\mf{sl}(N+R \mid R)$, and they are related as $R$ varies by taking $Q$-cohomology.

Further, any two such $\C[[\hbar]]$-points of $\mc{M}^{\C^\times}_{N+\infty \mid \infty,\op{Diff}(\C)}$ are related by a change of variables of the form
$$
\hbar \mapsto \lambda \hbar + f_2(\kappa) \hbar^2 + \dots
$$
where $\lambda$ is a non-zero constant and the $f_i(\kappa)$ are polynomials in $\kappa$. 
\label{corollary_unique_quantization} 
\end{corollary}
\begin{proof}
This is a formal consequence of the theorem above and of general obstruction theory arguments. 
\end{proof}
\subsection{}
   Let us give simple algebraic preliminaries we need before turning to the proof of the theorem. The first thing we need to do is to relate the moduli problem defined using the algebra $\op{Diff}(\C)$ to the moduli problem defined using some variant algebras. 

Let $B = \C[z_1,z_2,c,c^{-1}]$ where $[z_1,z_2] = c$.  Note that $B$ has a $\C^\times$-action, where the $z_i$ have weight $1$ and $c$ has weight $2$. We work over the base ring $\C[c,c^{-1}]$.  

We will first  derive some relations between the Hochschild cochain complexes of $U(\mf{gl}(N+R \mid R) \otimes B)$ and similar Hochschild complexes where the algebra $B$ is replaced by the algebra $\C[[z_1,z_2,c]]$ where $[z_1,z_2] = c$.

\begin{definition} 
 A $\C[[c]]$-linear restricted Hochschild cochain of $U(\mf{gl}(N+R \mid R) \otimes \C[[z_1,z_2,c]])$ is a Hochschild cochain which is a finite sum of cochains of definite weight under the $\C^\times$ action. 
\end{definition}
\begin{lemma} 
 The restricted $\C[[c]]$-linear Hochschild cochains of $U(\mf{gl}(N+R \mid R) \otimes \C[[z_1,z_2,c]])$ is the same the $\C[c]$-linear Hochschild cochains $U(\mf{gl}(N+R \mid R) \otimes \C[z_1,z_2,c])$.  
\end{lemma} 
\begin{proof} 
This follows from a simple general observation about graded vector spaces. Suppose that $V,W$ are vector spaces graded by the non-negative integers.  Let $V_i, W_i$ indicate the graded pieces, and $\br{V} = \prod V_i$, $\br{W} = \prod W_i$.  Then,
\begin{align*} 
\op{Hom}(V,W) &= \prod_i \left(\oplus_j\op{Hom} (V_i,  W_j)\right)\\
 \op{Hom}(\br{V},\br{W}) &= \prod_j \left(\oplus_i \op{Hom}(V_i,  W_j) \right)
\end{align*}
where on the second line we mean continuous linear maps. 
 
We can consider maps $V \to W$ of weight $k$ (meaning that they shift the grading by $k$), and also continuous maps $\br{V} \to \br{W}$ of weight $k$.  In either case we denote this space by $\op{Hom}_k$. Then, we find
\begin{align*} 
\op{Hom}_k(V,W) &= \prod_i \left(\op{Hom}( V_i,  W_{i+k}) \right)\\
 \op{Hom}_k(\br{V},\br{W}) &= \prod_i \left( \op{Hom}( V_{i-k} , W_i) \right)
\end{align*}
so that the spaces are the same. 

This observation applies to the case at hand, because the grading on $U(\C[[z_1,z_2,c]] \otimes \mf{gl}(N+R \mid R))$ coming from that on $\C[[z_1,z_2,c]]$ is entirely in non-negative degrees. 

\end{proof}
Next, we need another simple lemma.
\begin{lemma} 
The $\C[c,c^{-1}]$-linear restricted Hochschild cochain complex of $U( \C[z_1,z_2,c,c^{-1}] \otimes \mf{gl}(N+R \mid R))$ is obtained from the $\C[[c]]$-linear restricted Hochschild cochain complex of $U(\C[[z_1,z_2,c]]\otimes \mf{gl}(N+R \mid R))$ by tensoring over $\C [c]$  with $\C[c,c^{-1}]$.  
\end{lemma}
\begin{proof} 
This follows from first, applying the previous lemma, and then noting that we get the Hochschild cochains of $\C[z_1,z_2,c,c^{-1}]$ from those of $\C[z_1,z_2,c]$ by tensoring over $\C[c]$ with $\C[c,c^{-1}]$. 
\end{proof}
Next, we can put these together to get a lemma about the tangent complex to our moduli problem $\mc{M}_{N+\infty\mid\infty,B}$.   To state this lemma, as above let $B = \C[z_1,z_2,c,c^{-1}]$ and also let $\what{B} = \C[[z_1,z_2,c]]$.  The formal moduli problems $\mc{M}_{N+\infty \mid \infty,B}$ and $\mc{M}_{N+\infty \mid \infty, \what{B}}$ have $\C^\times$ actions.  It follows that we can consider the restricted part of the tangent complex, $T^{res}\mc{M}_{N+\infty \mid \infty,B}$, defined as the subspace consisting of finite direct sums of eigenvectors for the $\C^\times$-action. 
\begin{lemma} 
There is an isomorphism
$$
H^i \left( T^{res} \mc{M}_{N+\infty \mid \infty,B} \right)\iso H^i \left( T^{res}\mc{M}_{N+\infty\mid\infty,\what{B}} \right) [c^{-1}]. 
$$
\label{lem:B formal translation}
\end{lemma}
\begin{proof}
This follows from the previous lemmas.
\end{proof} 
The final lemma we need will translate between the tangent space of our moduli problem defined using $B = \C[z_1,z_2,c,c^{-1}]$ and of the corresponding moduli problem when we use $\op{Diff}(\C)$.  Putting these together gives us a dictionary allowing us to compute things about our moduli problem for $\op{Diff}(\C)$ using statements about the moduli problem for $\C[[z_1,z_2,c]]$.
\begin{lemma} 
Let us give $\op{Diff}(\C)$ a $\C^\times$-action under which $z$ has weight $1$ and $\dpa{z}$ has weight $-1$. 

Let $T^{(i)}\mc{M}_{N+\infty \mid \infty, \op{Diff}(\C)}$ denote the subcomplex of the tangent complex consisting of elements of weight $i$ under this $\C^\times$-action. 

Then, there is a quasi-isomorphism
$$
\oplus_{i \text{ even} } T^{(i)}\mc{M}_{N+\infty \mid \infty, \op{Diff}(\C)} \simeq T^{(0)} \mc{M}_{N+\infty \mid \infty, B}  
$$
where on the right hand side $T^{(0)}$ indicates the $\C^\times$-invariants in the tangent complex.
\label{lem:Diff(C) formal translation} 
\end{lemma}
\begin{proof}
 There a $\C^\times$-equivariant and $\C[c,c^{-1}]$-linear isomorphism of algebras
\begin{align*} 
\op{Diff}(\C) \otimes \C[c,c^{-1}]& \iso \C[z_1,z_2,c,c^{-1}] \\
z & \mapsto z_1 \\
\dpa{z}  &\mapsto c^{-1} z_2.  
\end{align*}
From this it is immediate that there is an isomorphism of complexes of $\C[c,c^{-1}]$-modules
$$
  T^{res}\mc{M}_{N+\infty \mid \infty, \op{Diff}(\C)}\otimes \C[c,c^{-1}]  =  T^{res}\mc{M}_{N+\infty \mid \infty, \op{Diff}(\C)\otimes \C[c,c^{-1}]} \iso T^{res} \mc{M}_{N+\infty\mid\infty,B} . 
$$
Passing to $\C^\times$-invariants on both sides and using the fact that $c$ is of weight $2$ gives the result.  
\end{proof}

\subsection{}
To compute the tangent complex of $\mc{M}_{N+\infty \mid \infty,\op{Diff}(\C)}$ it therefore suffices to compute the tangent complex of $\mc{M}_{N+\infty \mid \infty,\what{B}}$. There is a spectral sequence which allows us to relate this tangent complex to the tangent complex to $\mc{M}_{N+\infty\mid\infty,\C[[z_1,z_2]]}$. 

Invoking theorem \ref{theorem_tangent_complex_cyclic}, to compute the cohomology of this tangent complex we find that we need to compute the derived tensor product of the $\C[[z_1,z_2]]$-module $\C[[z_1,z_2]]^\vee$ with itself $n$ times.  In this case 
$$\C[[z_1,z_2]]^\vee = D_0 = \C[\partial_1, \partial_2]$$
is the space of constant-coefficient differential operators on $\C^2$. A constant-coefficient differential operator $D$ acts on the space of formal power series by sending a power series $f$ to $(Df)(0)$, the value at $0$ of $Df$. 

\begin{lemma}
There is an isomorphism of $\C[[z_1,z_2]]$-modules
$$
D_0^{\otimes^{\mbb L}_{\C[[z_1,z_2]]} n} \simeq D_0 [2(n-1) ].
$$
\end{lemma}
\begin{proof}
We will prove a more general result. Suppose we work in $k$ dimensions.  Let $\Oo(\what{D}^k) = \C[[z_1,\dots,z_k]]$ be the algebra of functions on the formal disc $\what{D}^k$.  Let $D_0(\what{D}^k)$ be the fibre at $0$ of the bundle of differential operators on $\what{D}^k$. This is the same as the constant-coefficient differential operators, which are the dual of $\Oo(\what{D}^k)$.  Then, we will see that
$$
D_0(\what{D}^k)^{\otimes^{\mbb L}_{\Oo(\what{D}^k)} n} \simeq D_0(\what{D}^k) [(n-1)k].
$$
Note that $D_0(\what{D}^k) = D_0(\what{D}^1)^{\otimes k}$, and that $\Oo(\what{D}^k) = \Oo(\what{D}^1)^{\what{\otimes} k}$, and that the tensor products we are computing are compatible with these decompositions.  So it suffices to prove the result when $k = 1$, and by induction to consider the case $n = 2$. 

In that case, note that we can identify $D_0(\what{D}^1) = \C[\partial]$ as the colimit of the $\C[[z]]$-modules $\C[z]/z^n$ under the maps  
$$
\C[z]/z^n \xto{z} \C[z]/z^{n+1}.
$$
Now, $\C[z]/z^n$ has a free resolution which is the two-term complex $\C[[z]] \xto{z^n} \C[[z]]$ situated in degrees $-1$ and $0$. We will compute the derived tensor product $\C[z]/z^n \otimes^{\mbb L}_{\C[[z]]} \C[\partial]$ using this free resolution of $\C[z]/z^n$ and then take the colimit over $n$.  We find that $\C[z]/z^n \otimes^{\mbb L}_{\C[[z]]} \C[\partial]$ is the cohomology of the two-term complex  $\C[\partial] \xto{z^n} \C[\partial]$ situated in degrees $-1$ and $0$. The differential in this complex is surjective, so the cohomology is the kernel, which consists of $\C[\partial]/\partial^{n}$ situated in degree $-1$. As a $\C[[z]]$-module, it is isomorphic to $\C[z]/z^n[1]$, where $[1]$ indicates that it's situated in degree $-1$.     

It remains to check that the maps in the diagram computing the colimit over $n$ are the same as the maps describing $\C[\partial]$ as a colimit of the modules $\C[z]/z^n$. This is easy to see, and tells us that 
$$\C[\partial] \otimes_{\C[[z]]}^{\mbb L} \C[\partial] \simeq \C[\partial][1]$$
as desired.

\end{proof}

Next, let us discuss the version of this lemma when we replace $\C[[z_1,z_2]]$ by a non-commutative deformation.
\begin{lemma} 
Let $A$ be the non-commutative algebra $\C[[z_1,z_2,c]]$ where $[z_1,z_2] = c$, and we view $A$ as an algebra over $\C[[c]]$. Let $A^\vee$ denote the $\C[[c]]$-linear dual.  Then, there is a quasi-isomorphism
$$
\underbrace{A^\vee \otimes^{\mbb L}_A \dots \otimes_A^{\mbb L} A^\vee }_{n \text{ copies of } A^\vee } \iso A^\vee[2(n-1)].  
$$
\label{lemma_noncommutative_tensor_powers}
\end{lemma}
\begin{proof} 
Consider the tensor product on the left hand side of the displayed equation. Clearly, it is a deformation of the corresponding tensor product when $c = 0$, and so some deformation of $D_0(\what{D}^2) [2(n-1)]$ where as above we identify the space $D_0(\what{D}^2)  = \C[\partial_1,\partial_2]$ with the linear dual of $\C[[z_1,z_2]]$.  This is a deformation of $D_0(\what{D}^2)$ as a $\C[[z_1,z_2]]$-bimodule, where the left and right actions coincide.  To prove the lemma, it suffices to show that there are no such deformations compatible with the symmetries of the problem. Thus, we need to compute the $\op{Ext}$-groups
$$
\op{Ext}^\ast_{\Oo(\what{D}^2 \times \what{D}^2)} (D_0(\what{D}^2),D_0(\what{D}^2)). 
$$
As in the proof of the previous lemma, this graded vector space is the tensor square of the corresponding graded vector space where we use the one-dimensional formal disc $\what{D}^1$ instead of $\what{D}^2$.  It suffice, therefore, to calculate
$$
\op{Ext}^\ast_{\Oo(\what{D}^1 \times \what{D}^1)} (D_0(\what{D}^1),D_0(\what{D}^1)). 
$$
We let $z,w$ be coordinates on the two copies of the formal disc.  As we saw in the proof of the previous lemma, $D_0 = \C[\partial]$ is the colimit over $n$ of the modules $\C[z]/z^n$ where $\C[[z]]$ and $\C[[w]]$ act in the same way.  We can take a free resolution of these modules, of the form
$$
\C[z]/z^n \simeq \C[[z,w,\eps_n,\delta]] 
$$
with differential $\d \delta = z-w$, $\d \eps_n = z^n$.  This allows us to compute that
\begin{align*} 
 \mbb{R}\op{Hom}_{\C[[z,w]] } (D_0,D_0) & \simeq \lim_n \op{Hom}_{\C[[z,w]]} ( \C[[z,w,\eps_n,\delta]], D_0 ) \\
   & = \lim_n D_0[\eps_n^\ast,\delta^\ast]. 
\end{align*}
Here $\eps_n^\ast,\delta^\ast$ are dual variables to $\eps_n,\delta$ and are of degree $1$.The differential is 
$$
z^n \eps_n^\ast + (z-w) \delta^\ast.
$$
Now, $z-w$ acts by $0$ on $D_0$, so this term in the differential vanishes.  Further, the map $z^n : D_0 \to D_0$ is surjective and has kernel isomorphic to $\C[z]/z^n$.   We find that
$$
\mbb{R}\op{Hom}_{\C[[z,w]] } (D_0,D_0) = \lim_n \left(\C[z,\delta^\ast]/z^n \right) . 
$$
The maps in the inverse system can be checked to be the natural surjective maps $\C[z]/z^{n+1} \to \C[z]/z^n$. Thus, the limit is $\C[[z , \delta^\ast]]$, and we have shown that
$$
\op{Ext}^\ast_{\C[[z,w]] } (D_0, D_0 ) = \C[[z,\delta^\ast]]  
$$
where, as before, $\delta^\ast$ is of degree $1$. 

Next, let's check the two-variable case.   Since this is the tensor square of the one-variable case, we find that
$$
\op{Ext}^\ast_{\C[[z_1,z_2,w_1,w_2]] } (D_0(\what{D}^2), D_0(\what{D}^2) ) = \C[[z_1,z_2,\delta_1^\ast,\delta_2^\ast]] . 
$$
As before, the variables $\delta_i^\ast$ are of degree $1$.  

As in the statement of the lemma, let $A = \C[[z_1,z_2,c]]$ be the non-commutative deformation of $\C[[z_1,z_2]]$ and let $A^\vee$ be the $\C[[c]]$-linear dual of $A$.  The results we have derived so far tell us that there is a spectral sequence
$$
\C[[z_1,z_2,\delta_1^\ast,\delta_2^\ast,c]] \Rightarrow \op{Ext}^\ast_{A \otimes A}(A^\vee, A^\vee). 
$$
Our goal is to analyze deformations of $A^\vee$ as an $A-A$ bimodule.  We are thus interested in elements of degree $1$, which in the representation on the left hand side contain precisely one $\delta_i^\ast$. 

We are more precisely interested in deformations of $A^\vee$ compatible with certain natural symmetries.  There is a $\C^\times$ which acts on everything under which the  variables $z_i$ have weight $1$ and $c$ has weight $2$.  One can check that the variables $\delta_i^\ast$ have weight $-1$.  Also, $\mf{sl}(2,\C)$ acts on everything in an evident way. 

There is at most one element of degree $1$ invariant under this action of $SL(2,\C)$ and $\C^\times$.  In the spectral sequence given above, it is represented by $z_1 \delta_1^\ast + z_2 \delta_2^\ast$. 

We ultimately want to prove that 
$$
A^\vee \otimes^{\mbb L}_{A} A^\vee \iso A^\vee[2]
$$
as $A-A$ bimodules.  It follows form the previous lemma that the right hand side is some deformation of the left hand side, and that this deformation is zero when we set $c = 0$.  Since the only non-zero class in $\op{Ext}^1_{A \otimes A}(A^\vee, A^\vee)$ which respects all the symmetries is non-zero when $c = 0$, it follows that the deformation of $A^\vee$ must be trivial.

\end{proof}

Now, note that $HH_\ast(A,A^\vee)$ is the dual of the Hochschild cohomology of $A$.  We let $HC^\ast(A)$ be the cyclic cohomology of $A$, which is the dual of the cyclic homology of $A$.   

Putting these computations together with theorem \ref{theorem_tangent_complex_cyclic} gives us an understanding of the tangent complex of the moduli space $\mc{M}_{N+\infty \mid \infty,A}$.

Let 
$$
M = HC^\ast(A)[-1] \oplus HH^\ast(A) \oplus HH^\ast(A)[-2]^{C_2} \oplus \dots 
$$
Let 
$$
V = t \C[[t,c]] [-1] \oplus \left( \bigoplus_{n \ge 1} A^{\otimes n}_{C_n} \right) [[c]].
$$
where the parameter $t$ has degree $2$, and the cyclic tensor powers of $A$ live in degree $0$. 
\begin{proposition}
\label{proposition_stable_Hochschild_differential_operators}
Let  $A = \C[[z_1,z_2,c]]$ be the non-commutative deformation of $\C[[z_1,z_2]]$. There is a long exact sequence  
$$
\dots \to H^i T \mc{M}_{N+\infty \mid \infty,A} \to H^{i+2} \Sym^\ast M \to H^{i+2} \Sym^\ast V \to \dots  
$$ 
\label{proposition_cyclic_noncommutative} 
\end{proposition}
\begin{proof}
This follows from Lemma \ref{lemma_noncommutative_tensor_powers} combined with Theorem \ref{theorem_tangent_complex_cyclic}. 
\end{proof}
Next, we will compute the cohomology of $M$.  Note that $M$ is concentrated in degrees $\ge 0$.  Since we are ultimately only interested in the low-degree cohomology of the tangent complex to $\mc{M}_{N+\infty\mid \infty,A}$, we need only compute explicitly the graded vector space $M$ in low degree (up to $3$).

Recall that, ultimately, we are interested in the \emph{restricted} Hochschild cohomology, which is by definition the subspace of the Hochschild cohomology which consists of finite sums of weight spaces for the $\C^\times$ action under which the variables $z_i$ have weight $1$ and $c$ has weight $2$.  We then want to invert $c$.  This means is that we can ignore $c$-torsion in $M$.  

Let us now compute the graded vector space $M$ modulo $c$-torsion. To do this, we need to recall a few facts.  Cyclic cohomology is the dual of cyclic homology, and the cyclic homology of $A$ can be computed using the Hochschild-Kostant-Rosenberg theorem.  We will use the grading convention that the differential is always of positive degree, so that cyclic homology is concentrated in negative degree. We have
\begin{align*} 
HC_0(\Oo(\what{D}^2) )  & = \Oo(\what{D}^2) \\
HC_{-1}( \Oo(\what{D}^2)  )&= \Omega^1(\what{D}^2) / \d \Omega^0(\what{D}^2) \\
&= \Omega^2(\what{D}^2) \\
HC_{-2} ( \Oo(\what{D}^2) ) &= \C. 
\end{align*} 
Now, when we make $\Oo(\what{D}^2)$ non-commutative by introducing the parameter $c$, we get an extra term in the differential of cyclic homology. One can calculate that this extra term is the map
\begin{align*} 
HC_{-1} = \Omega^2(\what{D}^2) &\mapsto \Oo(\what{D}^2) = HC_0\\
\alpha &\mapsto c \pi \vee \alpha    
\end{align*}
where $\pi$ is the Poisson tensor. It follows that, modulo $c$-torsion,  the only cyclic homology that appears is a copy of $\C$ in degree $-3$, coming from $HC_{-2}$.   Dually, this contributes a copy of $\C$ in degree $3$ to $M$.  This copy of $\C$ lives in $F^{-3} M$.  
  
Hochschild cohomology of $A$ can be computed with the help of the Hochschild-Kostant-Rosenberg theorem as well. We have
$$
HH^i(\Oo(\what{D}^2)) = \op{PV}^i(\what{D}^2)
$$
is the space of polyvector fields on the formal disc.    If we turn on the parameter $c$, making $\Oo(\what{D}^2)$ non-commutative, we introduce an extra term in the differential which is the map
\begin{align*} 
 HH^i(\Oo(\what{D}^2)) = \op{PV}^i(\what{D}^2) & \mapsto   \op{PV}^{i+1}(\what{D}^2) =  HH^{i+1}(\Oo(\what{D}^2)) \\
\alpha & \mapsto c \{\pi,\alpha\}  
\end{align*}
where, as above, $\pi$ is the Poisson bracket and $\{-,-\}$ is the Schouten bracket. 

It follows that when we ignore $c$-torsion, Hochschild cohomology of $A$ is concentrated in degree $0$ and consists of $\C[[c]]$.    

It's also easy to check that the action of the cyclic group on $HH^\ast(A)[-2]$ that appears on our formula for $M$ is trivial.  Thus, we find that $M$ contains a copy of $\C[[c]]$ in degree $0$, coming from $HH^0(A,A)$; and a second copy of $\C$ in degree $2$, coming from $HH^0(A, A[-2])$.  In sum, we have,
\begin{align*} 
 M^0 &= \C[[c]]\\
M^1 &= 0 \\ 
 M^2 &= \C[[c]]\\ 
 M^3 &= \C[[c]].  
\end{align*}  
The natural basis element for $M^0$ is the element $\kappa$. Recall that this description is modulo $c$-torsion. 

Finally, we are in a situation to compute the cohomology of the restricted tangent complex to $\mc{M}_{N+\infty \mid \infty,A}$, modulo $c$-torsion. 
\begin{theorem}
\label{theorem_hh_low_degree}
After inverting $c$, the cohomology of the restricted tangent complex to $\mc{M}_{N+\infty \mid \infty, A}$ is: 
\begin{align*} 
H^{0} (T^{res}\mc{M}_{N+\infty\mid \infty,A})[c^{-1}] &=  \C[c,c^{-1},  \kappa]   \\ 
H^1 (T^{res}\mc{M}_{N+\infty\mid \infty,A})[c^{-1}] = 0 
\end{align*}

\end{theorem}

Because of lemma \ref{lem:B formal translation}, the same result holds for the restricted tangent complex to the moduli space when $A$ is replaced by $\C[z_1,z_2,c,c^{-1}]$. Lemma \ref{lem:Diff(C) formal translation}  now implies a simialr result for the case when $A$ is replaced by $\op{Diff}(\C)$, thus proving theorem \ref{theorem_deformation_computation}.
\begin{proof}
Recall that we have a short exact sequence
$$
 H^{i+1} \Sym^\ast V\to   H^i T \mc{M}_{N+\infty \mid \infty,A} \to H^{i+2} \Sym^\ast M \to H^{i+2} \Sym^\ast V \to \dots  
$$
Recall that $V = t \C[[t]][-1] \oplus \bigoplus_{n \ge 1} A^{\otimes n}_{C_n}$. Thus $V$ is situated in degrees $0$ and degrees $\ge 3$. In the computations that follow, we will ignore all $c$ torsion, and invert $c$ at the end. We will also only pass to the restricted tangent complex at the end.

 Let us first compute that $H^1$ of the restricted tangent complex vanishes.  We have an exact sequence
$$
0 = H^2 \Sym^\ast V \to H^1 (T\mc{M}_{N+\infty\mid \infty,A})\to H^3 \Sym^\ast M \to H^3 \Sym^\ast V.   
$$
Now,
$$
H^3 \Sym^\ast M = \C[\kappa][[c]] 
$$
(modulo $c$-torsion).  And,
$$
H^3 \Sym^\ast V = t\cdot \Sym^\ast ( \bigoplus_{n \ge 1} A^{\otimes n}_{C_n} ).  
$$
One can calculate that the map
$$
H^3 \Sym^\ast M = \C[\kappa][[c]] \to H^3 \Sym^\ast V 
$$
is injective, showing that $H^1 T\mc{M}_{N+\infty \mid \infty,A} = 0$. 

Next, let us calculate $H^0$ of our tangent complex.  We find an exact sequence
$$
0 = H^1 \Sym^\ast V \to H^0 T \mc{M}_{N+\infty \mid \infty,A} \to H^2 \Sym^\ast M \to 0 = H^2 \Sym^\ast V.
$$
Thus, 
$$
H^0 T \mc{M}_{N+\infty \mid \infty,A} \iso H^2 \Sym^\ast M. 
$$
And $H^2 \Sym^\ast M = \C[\kappa][[c]]$.

Next, we will compute $H^{-1}$. We have an exact sequence
$$
H^0 \Sym^\ast M  \to H^0 \Sym^\ast V \to H^{-1} T \mc{M}_{N+\infty \mid \infty,A} \to H^1 \Sym^\ast M = 0.
$$
The map
$$
\C[\kappa][[c]] = H^0 \Sym^\ast M \to H^0 \Sym^\ast V 
$$
is injective.  Thus, 
$$
H^{-1} T\mc{M}_{N+\infty \mid \infty} = H^0 \Sym^\ast V / \C[\kappa][[c]]. 
$$

\end{proof}

\subsection{}
Finally, we can apply all this machinery to the problem we are interested in. 
\begin{theorem}
For all $N$, there is an isomorphism 
$$
U^{QFT}_{\hbar}(\op{Diff}(\C) \otimes \mf{gl}(N) ) \iso U^{comb}_{\hbar}(\op{Diff}(\C) \otimes \mf{gl}(N)) 
$$
which sends 
$$
\hbar \mapsto \lambda \hbar + f_2(\kappa) \hbar^2 + \dots
$$
where $\lambda$ is a non-zero constant and the $f_i(\kappa)$ are polynomials of $\kappa$. 
\end{theorem}
\begin{proof}
In light of corollary \ref{corollary_unique_quantization}, it suffices to show that 
\begin{enumerate} 
\item Both of our algebras arise from $\C[[\hbar]]$-points of $\mc{M}_{N+\infty \mid \infty, \op{Diff}(\C) }$. 
\item To first order in $\hbar$, the Hochschild cocycle defining the deformations is given by the formula in proposition \ref{proposition_explicit_formula}, up to a non-zero constant.  
\end{enumerate} 
Both algebras can be defined using the supergroups $\mf{gl}(N+R \mid R)$, in a way compatible with the natural maps for different values of $R$.   The only thing we have not verified, which we will check in the appendix, is that to first order in $\hbar$ the algebra $U^{QFT}_{\hbar}(\op{Diff}(\C))$ is given by the expression in proposition \ref{proposition_explicit_formula}.  
\end{proof}

\appendix
\section{ Non-triviality of the deformation $U_{\hbar}(\op{Diff}(\C) \otimes \mf{gl}_N)$} 
\label{appendix_non_trivial_deformation}

In this appendix we will show that the first-order deformations of $U(\mf{gl}_{N+R \mid R} \otimes \C[z_1,z_2])$ given by our two algebras $U^{comb}_\hbar(\mf{gl}_{N+R \mid R}\otimes \op{Diff}(\C))$ and $U^{QFT}_\hbar(\mf{gl}_{N+R \mid R} \otimes \op{Diff}(\C))$ are both non-trivial, and are proportional to each other at first order in $\hbar$.

\begin{proposition}
\label{proposition_explicit_formula}
There is a unique generator for the $\C[\kappa]$ module $H^0 T\mc{M}^{\C^\times}_{N+\infty \mid \infty,\op{Diff}(\C) }$ with the following property. Let us view this generator as giving rise to a deformation of the algebra $U(\mf{gl}_N\otimes\op{Diff}(\C)$.  Let $E_{\alpha \beta} \in \mf{gl}(N)$ denote the elementary matrix.  Then, in this deformation, 
\begin{align*} 
[E_{\alpha\beta}, E_{\gamma \delta}] =& \delta_{\beta \gamma} E_{\alpha \delta} - \delta_{\gamma \alpha} E_{\gamma \beta} \\
 [E_{\alpha\beta}, z E_{\gamma \delta}] =& \delta_{\beta \gamma} z_i E_{\alpha \delta} - \delta_{\gamma \alpha} z E_{\gamma \beta} \\
[E_{\alpha\beta}, \partial E_{\gamma \delta}] =& \delta_{\beta \gamma} \partial E_{\alpha \delta} - \delta_{\gamma \alpha} \partial E_{\gamma \beta} \\
 [E_{\alpha\beta}\partial, E_{\gamma\delta} z] =& \delta_{\beta \gamma}E_{\alpha\delta}(\partial   z) - \delta_{\alpha\delta}E_{\gamma\beta}(z   \partial)\\ 
& + \hbar E_{\alpha \delta} E_{\gamma \beta} - \hbar \sum_{\mu} \delta_{\beta \gamma}  E_{\alpha \mu} E_{\mu \delta}  
- \hbar\sum_{\mu} \delta_{\alpha \delta} E_{\gamma \mu} E_{\mu \beta} .
\end{align*} 
Here $\hbar$ is the deformation parameter.  
\end{proposition}
\begin{proof}
We have already constructed the explicit combinatorial algebra $U^{comb}_\hbar(\mf{gl}_N \otimes \op{Diff}(\C))$ (here we set $c = 1$ in the formulae for the relations in this algebra).  We have also seen in \ref{lemma_combinatorial_one_loop} that the commutators in $U^{comb}_\hbar(\mf{gl}_N \otimes \op{Diff}(\C))$ satisfy these relations modulo $\hbar^2$, as indeed do the commutators in the algebra $U_\hbar^{QFT}(\mf{gl}_N \otimes \op{Diff}(\C))$, as proved in proposition \ref{proposition_commutator_QFT}.    

We need to do two things. First, we need to verify that any deformation with such commutation relations among the lowest-lying generators is non-trivial.   Since the space $H^0 T\mc{M}^{\C^\times}_{N+\infty \mid \infty,\op{Diff}(\C) }$ is  a rank $1$ $\C[\kappa]$-module, this will tells us that this deformation will generate after we pass to the field of rational functions of $\kappa$.  Then, we have to verify that it generates directly as a $\C[\kappa]$-module.

Let us show that if it is non-trivial it must generate as a $\C[\kappa]$-module. To see this we introduce the concept of single-trace deformation. A deformation of $U(\mf{gl}_{N+R \mid R}\otimes \op{Diff}(\C))$ is single trace if, when we view it's corresponding Hochschild cocycle as a $\mf{gl}_{N+R \mid R}$ invariant tensor on some tensor powers of $\mf{gl}_{N+R \mid R}$, it involves only one trace. Our analysis of the problem of deforming the algebra $U(\mf{gl}_{N+R \mid R} \otimes \op{Diff}(\C))$ for all $R$ in a compatible way tells us  that the space of first-order deformations is a sum of contirbutions with $0$ traces, $1$ trace, $2$ traces. etc.  This is reflected in the description of the tangent complex to this space as a symmetric algebra on the cyclic cohomology of $\op{Diff}(\C) \oplus \op{Diff}(\C)^\vee[-1]$: the number of traces is reflected in which symmetric power of cyclic cohomology we are considering.  

Now, the deformation we are considering is (modulo $\hbar^2)$ single trace. Further, multilying any single-trace deformation by $\kappa^l$ results in a deformation with $l+1$ traces. It follows immediately that any non-trivial single-trace deformation must generate the space of deformations as a $\C[\kappa]$-module.  

To verify this deformation is non-trivial,  we will go a little into the guts of the analysis of deformations of the algebras $U(\mf{gl}_{N+R \mid R} \otimes A)$ where $A$ is an associative algebra.

Deformations of any associative algebra are described by Hochschild cohomology. For any Lie algebra $\mf{g}$, the Hochschild-Kostant-Rosenberg theorem describes the Hochschild cochains of $C^\ast(\mf{g})$ as $C^\ast(\mf{g}, \Sym^\ast \mf{g})$.  The deformation we are considering is an element of $C^2(\mf{g}, \Sym^2 \mf{g})$ for $\mf{g} = \mf{gl}_{N+R \mid R}\otimes \op{Diff}(\C)$.

We are looking at elements that are invariant under $\mf{gl}_{N+R \mid R}$ for $R$ sufficiently large.  As discussed in great detail in section \ref{section_uniqueness}, we can describe them using invariant theory for $\mf{gl}_{N+R \mid R}$. For any associative  algebra $A$, we can write the large $R$ Lie algebra cochain complex of $C^\ast(\mf{gl}_{N+R \mid R} \otimes A, \Sym^\ast \mf{gl}_{N+R \mid R} \otimes A)$ as the symmetric algebra cyclic cochains of the graded algebra $A \oplus A^\ast [-1]$.  Single-trace cochains correspond to the first symmetric power.  Since the cocycle corresponding to the deformation is single-trace, it can only be made exact by a single trace cochain. We thus need only consider the subcomplex of single-trace elements, which corresponds to the  cyclic cochains of $A \oplus A^\ast[-1]$ (where $A = \op{Diff}(\C)$ in our example). 

In our situation, our deformation class has a term which lies in the complex spanned by those cochains which involve two copies of $A^\ast[-1]$. This is because it is a  class which takes a commutator to a product of two generators.  The differential respects the grading by the numbers of elements of $A^\ast[-1]$ that a cyclic cochain depends on. Therefore we need only involve the subcomplex  which involves only two copies of $A^\vee[-1]$.  

Note also that we can restrict attention to cochains which are a particular weight under the $\C^\times$ action which scales the $z_i$, so that we can use the restricted dual of $\C[z_1,z_2]$. This restricted dual is spanned by elements $\partial_{z_1}^k \partial_{z_2}^l$ which send a polynomial $f \in \C[z_1,z_2]$ to $\partial_{z_1}^k \partial_{z_2}^{l} f (0)$. 

Explicitly,  an element of the cyclic cochain complex of $A \oplus A^\ast[-1]$ that depends on only two copies of $A^\ast$ can be written as a cyclic tensor of the form
$$
\alpha_1^\ast \otimes \dots \alpha_k^\ast \otimes a_1 \otimes \alpha_{k+1}^\ast \otimes \dots \alpha_n^\ast \otimes a_2 \in (A^\ast)^{\otimes k} \otimes A \otimes (A^\ast)^{\otimes n-k} \otimes A, 
$$  
where $\alpha_i^\ast \in A^\ast$ and $a_i \in A$.  By a standard trick in homological algebra,  we can write a variant complex with the same cohomology, called the reduced cochain complex, in which  the elements $\alpha_i^\ast \in A^\ast$ must be zero on the element $1 \in A$.  

Let us write down the differential on this complex. It is more convenient to write down the differential on the linear dual complex, which is spanned by elements 
$$
\alpha_1 \otimes \dots \alpha_k \otimes a_1^\ast \otimes \alpha_{k+1} \otimes \dots \alpha_n \otimes a_2^\ast \in A^{\otimes k} \otimes A^\ast \otimes A^{\otimes n-k} \otimes A^\ast. 
$$
If we use the reduced model, which has the same cohomology, the elements $\alpha_k$ are in the quotient of $A$ by the subspace spanned by the unit element. 

We view $A^\ast$ as being an $A-A$ bimodule in the natural way. Then, the differential is given by the expression 
\begin{multline*} 
 \d\left( \alpha_1 \otimes \dots \otimes \alpha_k \otimes a_1^\ast \otimes \alpha_{k+1} \otimes \dots \otimes \alpha_n \otimes a_2^\ast\right) \\
= \sum \pm \alpha_1 \otimes \dots\otimes  (\alpha_{i} \alpha_{i+1}) \otimes \dots \otimes \alpha_k \otimes a_1^\ast \otimes \alpha_{k+1} \otimes \dots\otimes \alpha_n \otimes a_2^\ast \\
\pm  \alpha_1 \otimes \dots \otimes (\alpha_k a_1^\ast) \otimes \alpha_{k+1} \otimes \dots \otimes \alpha_n \otimes a_2^\ast \\
\pm  \alpha_1 \otimes \dots \otimes \alpha_k \otimes (a_1^\ast \alpha_{k+1}) \otimes \dots \otimes \alpha_n \otimes a_2^\ast \\
\sum \pm  \alpha_1 \otimes \dots \otimes \alpha_k \otimes a_1^\ast \otimes \alpha_{k+1} \otimes (\alpha_j \alpha_{j+1}) \dots\otimes \alpha_n \otimes a_2^\ast\\
\pm  \alpha_1 \otimes \dots \otimes \alpha_k \otimes a_1^\ast\otimes \alpha_{k+1} \otimes \dots \otimes ( \alpha_n a_2^\ast) \\
\pm  \alpha_2 \otimes \dots \otimes \alpha_k \otimes a_1^\ast \otimes \alpha_{k+1} \otimes \dots \otimes \alpha_n \otimes (a_2^\ast \alpha_1). 
\end{multline*}
More abstractly, this differential is that on the Hochschild homology of $A$ with coefficients in the $A-A$ bimodule $A^\ast \otimes^{\mbb{L}}_A A^\ast$.  (We then must take the $S_2$ invariants). 

We will identify $\C[z,\partial_z]$ with the non-commutative algebra $\C[z_1,z_2]$ equipped with the Moyal product, where $\partial_z = z_1$, $z = z_2$.
 
We will use the basis of the dual of the vector space $\C[z_1,z_2]$ given by the elements $(z_1^k)^\ast(z_2^l)^\ast$. This is simply the dual basis to the basis of $\C[z_1,z_2]$ given by the elements $z_1^k z_2^l$. Thus, the complex we are considering is build from cyclic tensors involving some number of elements like $(z_1^k)^\ast(z_2^l)^\ast$ and exactly two  elements from the space $\C[z_i]$.

The  cochain coming from the deformation is of the form
\begin{equation}
 \eps_{ij} z_i^\ast \otimes z_j^\ast \otimes 1 \otimes 1 -  \eps_{ij}z_i^\ast \otimes 1 \otimes z_j^\ast \otimes 1  + \text{ higher order terms} 
\label{equation_cocycle}
\end{equation}
where the higher order terms are related to the commutators $[f(z_1,z_2) E_{\alpha \beta}, g(z_1,z_2) E_{\gamma \delta}]$ where one or both of $f,g$ has quadratic or higher order terms.   (In this equation $\eps_{ij}$ is the alternating symbol). 

To show that this element is non-zero in cohomology, it suffices to find an element in the dual complex which is closed and which pairs non-trivially with this element.

Let us write down such an element of the dual chain complex. We choose the element 
$$
 z_1 \otimes 1^\ast \otimes z_2 \otimes 1^\ast.
$$
This element clearly pairs non-trivially with the cocycle in equation \ref{equation_cocycle}.

Recall that in the dual chain complex, we are computing the differential in the cyclic homology of $A \oplus A^\vee[-1]$ where $A = \C[z_1,z_2]$.  Thus in the cyclic complex elements of $A$ will be treated as odd, and the differential picks up a sign when we move past them,  whereas elements of $A^\vee$ will be treated as even. 

We are using the notation $1^\ast$ to indicate the dual of the element $1 \in \C[z_1,z_2]$.  The bimodule structure on the dual of $\C[z_1,z_2]$ is such that is such that
\begin{align*} 
z_1 \cdot 1^\ast &= \tfrac{c}{2} z_2^\ast \\
 z_2 \cdot 1^\ast &= -\tfrac{c}{2} z_1^\ast \\
 1^\ast \cdot z_1 &=  - \tfrac{c}{2} z_2^\ast \\
 1^\ast \cdot z_2 &=   \tfrac{c}{2} z_1^\ast \\ 
\end{align*}
Then, with these conventions, we have 
\begin{align*}
\d (    z_1 \otimes 1^\ast \otimes z_2 \otimes 1^\ast ) =& 
 (z_1 \cdot 1^\ast) \otimes z_2 \otimes 1^\ast - z_1 \otimes (1^\ast \cdot z_2) \otimes 1^\ast -  z_1 \otimes 1^\ast \otimes \otimes (z_2 \cdot 1^\ast) + 1^\ast \otimes z_2 \otimes (1^\ast \cdot z_1) \\
=& z_2^\ast \otimes z_2 \otimes 1^\ast - z_1 \otimes z_1^\ast\otimes 1^\ast + z_1 \otimes 1^\ast \otimes z_1^\ast - 1^\ast \otimes z_2 \otimes z_2^\ast\\
=& 0 
\end{align*}
Thus, we have exhibited a closed element of the chain complex which pairs non-trivially with the cochain describing our deformation class. This tells us that this cochain can not be exact. 

\end{proof}


\begin{thebibliography}{KWWY14}

\bibitem[BDG15]{BulDimGai15}
M.~Bullimore, T.~Dimofte and D.~Gaiotto, \textsl{ The Coulomb Branch of 3d N=4
  Theories},
\newblock arXiv preprint arXiv:1503.04817  (2015).

\bibitem[Beh09]{Beh09}
K.~Behrend, \textsl{ Donaldson-Thomas type invariants via microlocal geometry},
\newblock Annals of Mathematics , 1307--1338 (2009).

\bibitem[BFN16]{BraFinNak16}
A.~Braverman, M.~Finkelberg and H.~Nakajima, \textsl{ Towards a mathematical
  definition of Coulomb branches of 3-dimensional N= 4 gauge theories, II},
\newblock arXiv preprint arXiv:1601.03586 , 8 (2016).

\bibitem[BJM13]{BusJoyMei13}
V.~Bussi, D.~Joyce and S.~Meinhardt, \textsl{ On motivic vanishing cycles of
  critical loci},
\newblock arXiv preprint arXiv:1305.6428  (2013).

\bibitem[BK06]{BruKle06}
J.~Brundan and A.~Kleshchev, \textsl{ Shifted Yangians and finite W-algebras},
\newblock Advances in Mathematics \textbf{ 200}(1), 136--195 (2006).

\bibitem[CG16]{CosGwi11}
K.~Costello and O.~Gwilliam,
\newblock \textsl{ Factorization algebras in quantum field theory},
\newblock Cambridge {U}niversity {P}ress, 2016.

\bibitem[CL15]{CosLi15}
K.~Costello and S.~Li, \textsl{ Quantization of open-closed BCOV theory, I},
\newblock (2015), {arXiv:1505.06703}.

\bibitem[CL16]{CosLi16}
K.~Costello and S.~Li, \textsl{ Twisted supergravity and its quantization},
\newblock (2016).

\bibitem[Cos11]{Cos11}
K.~Costello,
\newblock \textsl{ Renormalization and effective field theory},
\newblock Surveys and monographs, American Mathematical Society, 2011.

\bibitem[Cos13]{Cos13}
K.~Costello, \textsl{ Supersymmetric gauge theory and the {Y}angian},
\newblock (2013), {arXiv:1303.2632}.

\bibitem[Cos16]{Cos16}
\textsl{ Holography and Koszul duality: the example of the $M2$ brane},
\newblock (2016).

\bibitem[Eti04]{Eti04}
P.~Etingof, \textsl{ Cherednik and Hecke algebras of varieties with a finite
  group action},
\newblock arXiv preprint math/0406499  (2004).

\bibitem[Eti14]{Eti14}
P.~Etingof, \textsl{ Representation theory in complex rank, I},
\newblock arXiv preprint arXiv:1401.6321  (2014).

\bibitem[Fuk12]{Fuk12}
D.~Fuks,
\newblock \textsl{ Cohomology of infinite-dimensional Lie algebras},
\newblock Springer, 2012.

\bibitem[GHL09]{GuaHerLok09}
N.~Guay, D.~Hernandez and S.~Loktev, \textsl{ Double affine Lie algebras and
  finite groups},
\newblock Pacific journal of mathematics \textbf{ 243}(1), 1--41 (2009).

\bibitem[Gua07]{Gua07}
N.~Guay, \textsl{ Affine Yangians and deformed double current algebras in type
  A},
\newblock Advances in Mathematics \textbf{ 211}(2), 436--484 (2007).

\bibitem[GV98]{GopVaf98}
R.~Gopakumar and C.~Vafa, \textsl{ M-Theory and Topological Strings--II},
\newblock arXiv preprint hep-th/9812127  (1998).

\bibitem[GY16]{GuaYan16}
N.~Guay and Y.~Yang, \textsl{ On double deformed current algebras for simple
  Lie algebras},
\newblock (2016), {arXiv:1608.02900}.

\bibitem[JS08]{JoySon08}
D.~Joyce and Y.~Song, \textsl{ A theory of generalized Donaldson-Thomas
  invariants},
\newblock arXiv preprint arXiv:0810.5645  (2008).

\bibitem[KL12]{KieLi12}
Y.-H. Kiem and J.~Li, \textsl{ Categorification of Donaldson-Thomas invariants
  via perverse sheaves},
\newblock arXiv preprint arXiv:1212.6444  (2012).

\bibitem[KS10]{KonSoi10}
M.~Kontsevich and Y.~Soibelman, \textsl{ Cohomological Hall algebra,
  exponential Hodge structures and motivic Donaldson-Thomas invariants},
\newblock arXiv preprint arXiv:1006.2706  (2010).

\bibitem[KS16]{KorSci16}
P.~Koroteev and A.~Sciarappa, \textsl{ On elliptic algebras and large-n
  supersymmetric gauge theories},
\newblock Journal of Mathematical Physics \textbf{ 57}(11), 112302 (2016).

\bibitem[KWWY14]{KamWebWeeYac14}
J.~Kamnitzer, B.~Webster, A.~Weekes and O.~Yacobi, \textsl{ Yangians and
  quantizations of slices in the affine Grassmannian},
\newblock Algebra \& Number Theory \textbf{ 8}(4), 857--893 (2014).

\bibitem[Los16]{Los16}
I.~Losev, \textsl{ Representation theory of quantized Gieseker varieties, I},
\newblock arXiv preprint arXiv:1611.08470  (2016).

\bibitem[LQ84]{LodQui84}
J.-L. Loday and D.~Quillen, \textsl{ Cyclic homology and the Lie algebra
  homology of matrices},
\newblock Comment. Math. Helv. \textbf{ 59}(4), 569–591 (1984).

\bibitem[Lur12]{Lur12}
J.~Lurie,
\newblock \textsl{ Higher algebra},
\newblock 2012.

\bibitem[MW14]{MikWit14}
V.~Mikhaylov and E.~Witten, \textsl{ Branes And Supergroups},
\newblock (2014), {arXiv:1410.1175}.

\bibitem[Oko15]{Oko15}
A.~Okounkov, \textsl{ Lectures on $K$-theoretic computations in enumerative
  geometry},
\newblock (2015), {arXiv:1512.07363}.

\bibitem[Tsy83]{Tsy83}
B.~Tsygan, \textsl{ Homology of matrix algebras over rings and the Hochschild
  homology},
\newblock Uspeki Math. Nauk. \textbf{ 38}, 217–218 (1983).

\bibitem[Wit88]{Wit88a}
E.~Witten, \textsl{ Topological quantum field theory},
\newblock Comm. Math. Phys. \textbf{ 117}(2), 014, 21 pp. (electronic) (1988).

\bibitem[Wit98]{Wit98}
E.~Witten, \textsl{ Anti De Sitter Space And Holography},
\newblock (1998), {arXiv:hep-th/9802150}.

\end{thebibliography}
\end{document}